\documentclass[11pt,a4paper]{article}
\usepackage{amsmath}
\usepackage{amsfonts}
\usepackage{amssymb}
\usepackage{graphicx}
\usepackage{algorithm}

\usepackage{graphicx}
\usepackage{pgf}
\usepackage{paralist} 
\usepackage{wrapfig} 

\usepackage{epstopdf} 

\usepackage{url,hyperref}

\usepackage{smartref} 

\usepackage{placeins}

\usepackage{amsthm}
\usepackage{paralist}
\usepackage{tikz}
\usetikzlibrary{matrix,shapes,arrows}
\newtheorem{theorem}{Theorem}

\newtheorem{corollary}[theorem]{Corollary}
\newtheorem{lemma}[theorem]{Lemma}

\newtheorem{observation}[theorem]{Observation}

\newcommand{\ACK}{{\tt ACK}}
\newcommand{\NACK}{{\tt NACK}}

\newcommand{\send}{{\tt send}}
\newcommand{\receive}{{\tt receive}}

\newcommand{\maxser}{{\tt NS}}
\newcommand{\PUSH}{{\tt PUSH}}
\newcommand{\POP}{{\tt POP}}
\newcommand{\ENQUEUE}{{\tt ENQUEUE}}
\newcommand{\DEQUEUE}{{\tt DEQUEUE}}
\newcommand{\TOKEN}{{\sf TOKEN}}
\newcommand{\ENQ}{{\tt ENQ}}
\newcommand{\DEQ}{{\tt DEQ}}
\newcommand{\ENQT}{{\tt ENQ\_T}}
\newcommand{\ENQH}{{\tt ENQ\_H}}
\newcommand{\DEQT}{{\tt DEQ\_T}}
\newcommand{\DEQH}{{\tt DEQ\_H}}
\newcommand{\HEAD}{{\tt HEAD\_TOKEN}}
\newcommand{\TAIL}{{\tt TAIL\_TOKEN}}

\newcommand{\INSERT}{{\tt INSERT}}
\newcommand{\SEARCH}{{\tt SEARCH}}
\newcommand{\DELETE}{{\tt DELETE}}
\newcommand{\BDIRDELETE}{{\tt BlockDirDelete}}
\newcommand{\DIRDELETE}{{\tt DirDelete}}
\newcommand{\DIRINSERT}{{\tt DirInsert}}
\newcommand{\DIRSEARCH}{{\tt DirSearch}}

\newcommand{\CCSYNCH}{{\sf CC-Synch}}
\usepackage{smartref} 
\newcounter{linenum}
\addtoreflist{linenum}
\def\codeTabSpace{\hspace*{6mm}}
\newenvironment{code}%
{\begin{tabbing}%
\codeTabSpace \= \hspace*{70mm} \= \hspace*{42mm} \= \kill%
}%
{\end{tabbing}%
}
\newcounter{ind}
\newcommand{\cm}[1]{\hfill \`{\tiny/* #1 */}}
\newcommand{\ocm}[1]{\hfill \`{\tiny/* #1}}
\newcommand{\ccm}[1]{\hfill \`{\tiny#1 */}}

\newcommand{\n}{\addtocounter{ind}{3}\hspace*{3mm}}
\newcommand{\p}{\addtocounter{ind}{-3}\hspace*{-3mm}}
\newcommand{\nl}{\\\stepcounter{linenum}{\scriptsize \arabic{linenum}}\>\hspace*{\value{ind}mm}}
\newcommand{\ul}{\\\>\hspace*{\value{ind}mm}}
\newcommand{\bl}{\\[-2.0mm]\>\hspace*{\value{ind}mm}}
\newcommand{\firstline}{\stepcounter{linenum}{\scriptsize \arabic{linenum}}\>}
\newcommand{\lref}[1]{\linenumref{#1}} 
\newcommand{\lreset}{\setcounter{linenum}{0}}

\newcommand{\here}[1]{[[[{\bf #1}]]]\marginpar{***}}

\newcommand{\ignore}[1]{}
\newcommand{\lchanged}[1]{{\color{black} #1}\normalcolor}

\newcommand{\true}{{\tt true}}
\newcommand{\false}{{\tt false}}
\newcommand{\If}{{\tt if}}
\newcommand{\Elseif}{{\tt else if}}
\newcommand{\Else}{{\tt else}}

\newcommand{\While}{{\tt while}}
\newcommand{\For}{{\tt for}}
\newcommand{\Foreach}{{\tt for each}}
\newcommand{\Do}{{\tt do}}
\newcommand{\Break}{{\tt break}}
\newcommand{\Switch}{{\tt switch}}
\newcommand{\Case}{{\tt case}}

\newcommand{\return}{{\tt return}}

\newcommand{\bool}{\mbox{\tt boolean}}
\newcommand{\void}{\mbox{\tt void}}
\newcommand{\integer}{\mbox{\tt int}}

\newcommand{\CAS}{{\tt CAS}}

\newcommand{\DMA}{\textsc{Dma}}

\usepackage{subfigure}

\usepackage{listings}

\newcommand{\formic}{{\sf Formic}}

\newcommand{\GridGain}{{\sf GridGain}}
\newcommand{\Hazelcast}{{\sf Hazelcast}}
\newcommand{\CSTACK}{{\sf CStack}}
\newcommand{\ESTACK}{{\sf EStack}}
\newcommand{\BSTACK}{{\sf HStack}}
\newcommand{\TSTACK}{{\sf TStack}}
\newcommand{\DSTACK}{{\sf DStack}}
\newcommand{\CQUEUE}{{\sf CQueue}}
\newcommand{\TQUEUE}{{\sf TQueue}}
\newcommand{\DQUEUE}{{\sf DQueue}}
\newcommand{\HQUEUE}{{\sf HQueue}}

\newcommand{\CPush}{{\tt ClientPush}}

\newcommand{\CEnq}{{\tt ClientEnqueue}}
\newcommand{\CDeq}{{\tt ClientDequeue}}

\usepackage[final]{changes}
\definechangesauthor[name={Lena}, color=purple]{lk}
\definechangesauthor[name={Nick}, color=gray]{nk}

\usepackage{fullpage}

\title{Efficient Distributed Data Structures for 
		Future Many-core Architectures}
		
		\author{Panagiota Fatourou\\
			\small FORTH-ICS \& Univ. of Crete\\[-0.8ex]
			\small \texttt{faturu@ics.forth.gr}
			\and
			Nikolaos D. Kallimanis\\
			\small  ISI / Athena RC \& University of Ioannina\\[-0.8ex]
			\small \texttt{nkallima@isi.gr}
			\and
			Eleni Kanellou\thanks{Contact author. {\bf Address:} FORTH-ICS. N. Plastira 100, Vassilika Vouton. GR-70013 Heraklion (Crete), Greece. 
				{\bf Tel.:} +302810391709}\\
			\small  FORTH-ICS\\[-0.8ex]
			\small \texttt{kanellou@ics.forth.gr}
			\and
			Odysseas Makridakis\\
			\small Univ. of Crete\\[-0.8ex]
			\small \texttt{odmakryd@csd.uoc.gr}
			\and
			Christi Symeonidou\\
			\small Univ. of Crete\\[-0.8ex]
			\small \texttt{chsymeon@ics.forth.gr}
			\\
			\\
			\\
			\\
		}
\date{}

\begin{document}

\maketitle

\begin{abstract}
We study general techniques for implementing distributed data structures 
on top of future many-core architectures with {\bf non cache-coherent or 
partially cache-coherent memory}. With the goal of contributing towards 
what might become, in the future, the concurrency utilities package in 
Java collections for such architectures, we end up with a comprehensive 
collection of data structures by considering different variants of these 
techniques. 
To achieve scalability, we study a generic scheme which makes all our
implementations {\em hierarchical}. We consider a collection of known 
techniques for improving the scalability of concurrent data structures 
and we adjust them to work in our setting. 
We have performed experiments which illustrate that some of these 
techniques have indeed high impact on achieving scalability. Our 
experiments also reveal the performance and scalability power of the 
hierarchical approach. We finally present experiments to study energy 
consumption aspects of the proposed techniques by using an energy model 
recently proposed for such architectures.
\end{abstract}

\section{Introduction}
\label{sec:intro}
The dominant parallelism paradigm used by most high-level, high 
productivity languages such as Java, is that of threads and cache-coherent 
shared memory among all cores. However, cache-coherence does not scale 
well with the number of cores~\cite{novakovic2014scale}. So, future 
many-core architectures are not expected to support cache-coherence across 
all cores. They would rather feature multiple coherence islands, each 
comprised of a number of cores that share a coherent view of a part of the 
memory, but no hardware cache-coherence will be provided among cores of 
different islands. Instead, the islands will be interconnected using fast 
communication channels.	In recent literature, we meet even more aggressive 
approaches with Intel having proposed two fully non cache-coherent 
architectures, Runnemede~\cite{carter2013runnemede} and 
SSC~\cite{6077845}.  Additionally,~\cite{formic} presents the \formic\ 
board, a 512-core non cache-coherent prototype. 

Such architectures impose additional effort in programming them, since 
they require to explicitly code all communication and synchronization 
using messages between processors.
Previous works~\cite{mcilroy:oopsla10,yu1997java,zhu2002jessica2}
indicate the community's interest in bridging the gap between non 
cache-coherent or distributed architectures, and high-productivity 
programming languages by implementing runtime environments, like the 
Java Virtual Machine (JVM), for such architectures, which maintain the 
shared-memory abstraction. 

The difficulty in parallelizing many applications comes from
those parts of the computation that require communication and 
synchronization via data structures. Thus, the design of effective 
concurrent data structures is crucial for many applications.
On this avenue, \url{java.util.concurrent}, which is Java's concurrency 
utilities package (JSR 166)~\cite{lea:book06,javautils} provides a wide 
collection of concurrent data 
structures~\cite{lea:book06,MS96,Pugh,scherer:2006}. 
To run Java-written programs on a non cache-coherent architecture, 
Java's VM must be ported to that architecture, and some fundamental 
communication and synchronization primitives (such as \CAS, locks, and 
others) must be implemented. Normally, once this is done, it will be 
possible to execute applications that employ the concurrency utilities 
package without any code modification. However, the algorithms provided
in the package have been chosen to perform well on shared memory 
architectures. So, they take no advantage of the communication and 
synchronization features of non cache-coherent architectures and do not 
cope with load balancing issues or with the distribution of data among 
processors. Thus, they are expected to be inefficient when executing 
through JVMs ported for such an architecture, even if optimized 
implementations of locks, \CAS, and other primitives are provided on top 
of the architecture. Therefore, there is an urgent need to develop
novel data structures and algorithms, optimized for non cache-coherent 
architectures.

We study general techniques for implementing distributed data structures, 
such as stacks, queues, dequeues, lists, and sets, on many-core 
architectures with non or partially cache-coherent memory. With the goal 
of contributing to what might become, in the future, the concurrency 
utilities package in Java collections for such architectures, we end up, 
by considering different variants of these techniques, ending up with a 
comprehensive collection of data structures, richer than those provided in 
\url{java.util.concurrent}. Our collection, which is based on 
message-passing to achieve the best of performance, facilitates the 
execution of Java-written code on non-cache coherent architectures, 
without any modification, in a highly efficient way. 

To achieve {\em scalability}, that is, maintaining good performance as the 
number of cores increases~\cite{NA1991}, we study a generic scheme, which 
can be used to make all our implementations hierarchical~\cite{DMS15,FK12}.
The {\em hierarchical} version of an implementation exploits the memory 
structure and the communication characteristics of the architecture to 
achieve better performance. Specifically, just one core from each island, 
the {\em island master}, participates in the execution of the implemented 
distributed algorithm, whereas the rest submit their requests to this 
core. To execute the implemented distributed algorithm, the island masters 
must exploit the communication primitives provided for fast communication 
among different islands. In architectures with thousands of cores, we 
could employ a more advanced hierarchical structure, e.g. a tree-like 
hierarchy, of intermediate masters for better scalability. Depending on 
the implemented data structure, the island master employs 
elimination~\cite{Hendler:2004}, batching and other techniques to enhance 
scalability and performance. In the partially cache-coherent case, the 
cores of the same island may synchronize by employing 
combining~\cite{FK12,HIST10}. 

For efficiency, in some of our implementations,  we employ a highly 
scalable distributed hash table (DHT) which uses a simple standard 
technique~\cite{Devine93,hazelcast,shahmirzadiphd} to distribute the data 
on different nodes. Based on it and by employing counting 
networks~\cite{Count1,Count2}, we can come up with fully decentralized, 
scalable implementations of queues and stacks.  We also present 
implementations of (sorted and unsorted) lists, some of which support 
complex operations like range queries. To design our algorithms, we derive 
a theoretical framework which captures the communication characteristics 
of non cache-coherent architectures. This framework may be of independent 
interest. In this spirit, we further provide full theoretical proofs of 
correctness for some of the algorithms we present.

We have performed experiments on top of a non cache-coherent 512-core 
architecture, built using the \formic~\cite{formic} hardware prototyping 
board. We distill the experimental observations into a metric expressing 
the scalability potential of such implementations. The experiments 
illustrate nice scalability characteristics for some of the proposed 
techniques and reveal the performance and scalability power of the 
hierarchical approach.

\section{Related Work}
\label{sec:related}
Distributed transactional memory 
(DTM)~\cite{BAB08,D2STM,DhokePR15,TM2C,DiSTM,Man06,Saad2011,SR11}
is a {\em generic} approach for achieving synchronization, so 
data structures can be implemented on top of it. Transactional memory 
(TM)~\cite{herlihy:isca93,ST1995} is a programming paradigm which provides 
the transaction's abstraction; a {\em transaction} executes a piece of 
code containing accesses to {\em data items}. TM ensures that each 
transaction seems to execute sequentially and in isolation. 
However, it introduces significant performance overheads whenever reads 
from or writes to data items take place, and requires the programmer to 
write code in a transactional-compatible way. (When the transactions 
dynamically allocate data, as well as when they synchronize operations on 
dynamic data structures, compilers cannot detect all possible data races 
without trading performance, by introducing many false positives.) Our 
work is on a different avenue: towards providing a customized library of  
highly-scalable data structures, specifically tailored for non  
cache-coherent machines.

TM$^2$C~\cite{TM2C} is a DTM proposed for non cache-coherent machines. 
The paper presents a simple distributed readers/writers lock service where 
nodes are responsible for controlling access to memory regions. It also 
proposes two contention management (CM) schemes (Wholly and FairCM) that 
could be used to achieve starvation-freedom. However, in Wholly, the 
number of times a transaction $T$ may abort could be as large as the 
number of transactions the process executing $T$ has committed in past, 
whereas in FairCM, progress is ensured under the assumption that there is 
no drift~\cite{Attiya,lamport:cacm78} between the clocks of the different 
processors of the non cache-coherent machine. Read-only transactions in 
TM$^2$C can be slow since they have to synchronize with the lock service
each time they read a data item, and in case of conflict, they must
additionally synchronize with the appropriate CM module and may have to 
restart several times from scratch. Other existing 
DTMs~\cite{BAB08,D2STM,Saad2011a,SR11,SaadR11}, also impose common DTM 
overheads.

The data structure implementations we propose do not cause any space 
overhead, read-only requests are fast, since all nodes that store data of 
the implemented structure search for the requested key in parallel, and 
the number of steps executed to perform each request is bounded. We remark 
that, in our algorithms, information about active requests is submitted to 
the nodes where the data reside, and data are not statically assigned to 
nodes, so our algorithms follow neither the data-flow 
approach~\cite{BF10,SaadR11} nor the control-flow 
approach~\cite{BAB08,Saad2011a} from DTM research.

Distributed directory 
protocols~\cite{Combine,Attiya2015,Ballistic,Spiral,Relay} have been 
suggested for locating and moving objects on a distributed system. Most of 
the directory protocols follow the simple idea that each object is 
initially stored in one of the nodes, and as the object moves around, 
nodes store pointers to its new location. They are usually based either on 
a spanning tree~\cite{Arrow,Relay} or a hierarchical overlay 
structure~\cite{Combine,Ballistic,Spiral}. Remarkably, among them, 
COMBINE~\cite{Combine} attempts to cope with systems in which 
communication is not uniform. Directory protocols could potentially serve 
for managing objects in DTM. However, to implement a DTM system using a 
directory protocol, a contention manager must be integrated with the  
distributed directory implementation. As pointed out in~\cite{Attiya2015}, 
this is not the case with the current contention managers and distributed 
directory protocols. It is unclear how to use these protocols to get 
efficient versions of the distributed data structures we present in this 
paper. 

Previous research results~\cite{KBCL09,KPR+08,LDK+08,Kroll:sigmod94} 
propose how to support dynamic data structures on distributed memory 
machines. Some are restricted on tree-like data 
structures~\cite{LDK+08,Kroll:sigmod94}, other focus on data-parallel 
programs~\cite{KBCL09}, some favor code migration, whereas other focus 
on data replication. We optimize beyond simple distributed memory 
architectures by exploiting the communication characteristics of non 
cache-coherent multicore architectures. Some techniques   
from~\cite{KBCL09,KPR+08,LDK+08} could be of interest though to further 
enhance performance, fault-tolerance and scalability in our 
implementations.

Distributed data structures have also been 
proposed~\cite{Aspnes03,Gribble:2000,Hilford:97,Martin01,DBLP:journals/pvldb/AguileraGS08}
in the context of peer-to-peer systems or cluster computing, where   
dynamicity  and fault-tolerance are main issues. They tend to provide weak 
consistency guarantees. Our work is on a different avenue. 

\Hazelcast~\cite{hazelcast} is an in-memory data grid middleware which 
offers implementations for maps, queues, sets and lists from the Java 
concurrency utilities interface. These implementations are optimized for 
fault tolerance, so some form of replication is supported. Lists and sets 
are stored on a single node, so they do not scale beyond the capacity of 
this node. The queue stores all elements to the memory sequentially before 
flushing them to the datastore. Like \Hazelcast, 
\GridGain~\cite{gridgain}, an in-memory data fabric which connects 
applications with datastores, provides a distributed implementation of 
queue from the Java concurrency utilities interface. The queue can be 
either stored on a single grid node, or be distributed on different grid 
nodes using the datastore that exists below \GridGain. In a similar vein, 
Grappa~\cite{Nelson:2015:LSD:2813767.2813789} is a software system that 
can provide distributed shared memory over a cluster of machines. So-called 
delegate operations are used in order to access shared memory, relieving 
the programmer of having to reason about remote or local memory. However, 
Grappa does not provide the programmer with a data structure library.

We extend some of the ideas from hierarchical lock implementations 
and other synchronization protocols for NUMA cache-coherent 
machines~\cite{DMS15,FK12,D11} to get hierarchical implementations for a 
non cache-coherent architecture. Tudor {et al.}~\cite{TGT15} attempt to 
identify patterns on search data structures, which favor implementations 
that are portably scalable in cache-coherent machines. The patterns they 
came up with cannot be used to automatically generate a concurrent 
implementation from its sequential counterpart; they rather provide hints 
on how to apply optimizations when designing such implementations.

\section{Abstract Description of Hardware}
\label{sec:hw}
Inspired by the characteristics of non cache-coherent  
architectures~\cite{carter2013runnemede} and prototypes~\cite{formic},
we consider an architecture which features $m$ {\em islands} (or 
{\em clusters}), each comprised of $c$ cores (located in one or more 
processors). The main memory is split into modules, with each module 
associated to a distinct island (or core). A fast cache memory is located 
close to each core. No hardware cache-coherence is provided among cores of 
different islands: this means that different copies of the same variable residing on caches
of different islands may be inconsistent. The islands are interconnected 
with fast communication channels.

The architecture may provide cache-coherence for the memory modules of an 
island to processes executing on the cores of the island, i.e. the cores 
of the same island may see the memory modules of the island as 
cache-coherent shared memory. If this is so, we say that the architecture 
is {\em partially cache-coherent}; otherwise, it is {\em non 
cache-coherent}. A process can send messages to other processes by 
invoking \send\ and it can receive messages from other processes by 
invoking \receive. The following means of communication between cores (of 
the same or different islands) are provided:

\vspace*{.1cm}
\noindent
{\em Send/Receive mechanism}: Each core has its own mailbox which is 
a FIFO queue (implemented in hardware). (If more than one processes are 
executed on the same core, they share the same hardware mailbox; in this 
case, the functionality of \send\ and \receive\ can be provided to each 
process through the use of software mailboxes.) A process executing on the 
core can send messages to other processes by invoking \send, and it can 
receive messages from  other processes by invoking \receive. Messages are 
not lost and are delivered in FIFO order. An invocation of \receive\ 
blocks until the requested message arrives. The first parameter of an 
invocation to \send\ determines the core identifier to which the message 
is sent.

\vspace*{.1cm}
\noindent
{\em Reads and Writes through DMA}: A Direct Memory Access (DMA) engine 
allows certain hardware subsystems to access the system's memory without 
any interference with the CPU. We assume that each core can perform 
\DMA($A,B,d$)\ to copy a memory chunk  of size $d$ from memory address $A$ 
to memory address $B$ using a DMA (where $A$ and  $B$ may be addresses 
in a local or a remote memory module). We remark that DMA is not 
executed atomically. It rather consists of a sequence of atomic reads of 
smaller parts, e.g. one or a few  words, (known as {\em bursts}) of the 
memory chunk to be transferred, and atomic writes of these small parts to 
the other memory module. DMA can be used for performance optimization. 
Once the size of the memory chunk to be transferred becomes larger (by a 
small multiplicative factor) than the maximum message size supported by 
the architecture, it is more efficient to realize the transfer using DMA 
(in comparison to sending messages). Specifically, we denote by $MMS$ the 
maximum size of a message supported by the architecture. (Usually, this 
size is equal to either a few words or a cache line). Consider that the 
chunk of data that a core wants to send has size equal to $B$. To send 
these data using messages, the core must send $\frac{B}{MMS}$ messages. Each 
message has a cost $C_M$ to set it up. So, in total, to transfer the data 
using messages, the overhead paid is $\frac{B}{MMS}*C_M$. Setting up a DMA 
has a cost $C_D$, which in most architectures is by a small constant 
factor larger than $C_M$. However, $C_D$ is paid only once for the entire 
transfer of the chunk of data since it is done with a single DMA. 
Additionally, sending and receiving $\frac{B}{MMS}$ messages requires that 
a core's CPU will be involved $\frac{B}{MMS}$ times to send each message 
and another core's CPU will be involved $\frac{B}{MMS}$ times to receive 
these messages. This cost will be greatly avoided when using a single DMA 
to transfer the entire chunk of data. Thus, it is beneficial in terms of 
performance, to use DMA for transferring data whenever the size of the 
data to be transferred is not too small.

\section{Theoretical Framework}
\label{sec:model}
An implementation of a data structure (DS) stores its state in the memory 
modules and provides an algorithm, for each process, to implement each 
operation supported by the data structure. We model the submission and delivery of 
messages sent by processes by including incoming and outgoing message 
buffers in the state of each process (as described in standard 
books~\cite{Attiya,Lynch1996} on distributed computing).

We model each process as a state machine. We model a DMA request as a 
code block which contains a sequence of interleaved burst reads from a 
memory module and burst writes to a memory module. A DMA engine executes 
sequences of DMA requests, so it can also be modeled as a simple state 
machine whose state includes a buffer storing DMA requests that are to be 
executed. A {\em configuration} is a vector describing the state of each 
process (including its message buffers), the state of each DMA engine, the 
state of the caches (or the shared variables in case shared memory is  
supported among the cores of each island), and the states of the memory 
modules. In an {\em initial configuration}, each process and DMA engine is 
in an initial state, the shared variables and the memory modules are in 
initial states and all message and DMA buffers are empty. An {\em event} 
can be either a step by some process, a step by a DMA engine, or the  
delivery of a message; in one step, a process may either transmit exactly 
one message to some process and at least one message to every other 
process, or access (read or write) exactly one shared variable, or 
initiate a DMA transfer, or invoke an operation of the implemented data structure. In 
a {\em DMA step} a burst is read from or written to a memory module. All 
steps of a process should follow the process's algorithm. Similarly, all 
steps of a DMA engine should be steps of the code block that performs 
a DMA request submitted to this DMA engine.

An execution is an alternating sequence of configurations and events 
starting with an initial configuration. The {\em execution interval} of an 
instance of an operation $op$ that completes in an execution $\alpha$ is the subsequence 
of $\alpha$ starting with the configuration preceding the invocation of 
this instance of $op$ and ending with the configurations that follows its 
response. 
If an instance of an operation $op$ does not complete in $\alpha$, then the 
execution interval of $op$ is the suffix of $\alpha$ starting with the configuration 
preceding the invocation of this instance of $op$. An execution $\alpha$ 
imposes a partial order $\prec_{\alpha}$ on the instances of its operations, 
such that for two instances $op_1$ and $op_2$ of some operations in $\alpha$, 
$op_1 \prec_{\alpha} op_2$ if the response of $op_1$ precedes the invocation 
of $op_2$ in $\alpha$. If it neither holds that $op_1 \prec_{\alpha} op_2$ not 
that $op_2 \prec_{\alpha} op_1$, then we say that $op_1$ and $op_2$ are {\em concurrent}. If no two operation instances in $\alpha$ are concurrent, then 
$\alpha$ is a sequential execution. 

A step is {\em enabled} at a configuration $C$, if the process 
or DMA engine will execute this step next time it will be scheduled.
A finite execution $\alpha$ is fair if, for each process $p$ and each DMA 
engine $e$, no step by $p$ or $e$ is enabled at the final configuration 
$C$ of $\alpha$ and all messages sent in $\alpha$ have been delivered by 
$C$. An infinite execution $\alpha$ is fair if the following hold:

\begin{compactitem}
\item for each process $p$,
either $p$ takes infinitely many steps in $\alpha$, or there are infinitely many 
configurations in $\alpha$ such that in each of them (1) no step by $p$ is enabled,
(2) for every prefix of $\alpha$ that ends at such a configuration, all messages sent by $p$ have been delivered.
\item for each DMA engine $e$,
either $e$ takes infinitely many steps in $\alpha$, or there are infinitely many 
configurations in $\alpha$ such that in each of them no step by $e$ is enabled
(we remark that a DMA engine $e$ always have an enabled step as long as its DMA buffer is not empty).
\end{compactitem}

\vspace*{.1cm}
\noindent
{\bf Correctness.} For correctness, we consider {\em 
linearizability}~\cite{herlihy1990linearizability}. 
This means that, for every execution $\alpha$, there is a sequential 
execution $\sigma$, which contains all the completed operations in $\alpha$ 
(and some of the uncompleted ones) so that the response of each operation in 
$\sigma$ is the same as in $\alpha$ and $\sigma$ respects the partial order 
of $\alpha$.
In such a linearizable execution $\alpha$, 
one can assign a {\em linearization 
point} to each completed operation instance (and to some of the 
uncompleted ones) so that the linearization point of each operation occurs 
after the invocation event and before the response event the operation, 
and so that the order of the linearization points is the same as the order 
of the operations in $\sigma$.

\vspace*{.1cm}
\noindent
{\bf Progress.} We aim at designing algorithms that always terminate,
i.e. reach a state where all messages sent have been delivered and no step 
is enabled.  We do not cope with message or process failures. 

\vspace*{.1cm}
\noindent
{\bf Communication Complexity.}
Communication between the cores of the same island is usually faster than that across islands.
Thus, the communication complexity of an algorithm for a non cache-coherent architecture
is measured in two different levels, namely the intra-island communication 
and the inter-island communication. 
The {\em inter-island communication complexity} of an instance $inst$ of 
an operation $op$ 
in an execution $\alpha$ is the total number of messages sent by every core $c$ to cores residing 
on different islands from that of $c$ for executing $inst$ in $\alpha$.
The inter-island communication complexity of $op$ in $\alpha$
is the maximum, over all instances $inst$ of $op$ in $\alpha$, of the 
inter-island communication complexity of $inst$ in $\alpha$.
The inter-island communication complexity of $op$ for an implementation 
$I$ is the maximum, over all executions $\alpha$ produced by $I$,
of the inter-island communication complexity of $op$ in $\alpha$. 
We remark that communication can be measured in a more fine-grained way 
in terms of bytes transferred instead of messages sent,
as described in~\cite{Lynch1996}. For simplicity, we focus on the higher abstraction 
of measuring just the number of messages as described in~\cite{Attiya}.

If the architecture is non cache-coherent, 
then the {\em intra-island communication complexity} 
is defined as follows.
The {\em intra-island communication complexity} of an instance $inst$ of 
$op$ in $\alpha$ 
is the maximum, over all islands, of the total number of messages sent by every core $c$ of an island 
to cores residing on the same island as $c$ for executing $inst$ in $\alpha$.

If the architecture is partially cache-coherent,
then we measure the inter-island communication complexity 
in terms of cache misses following the cache-coherence (CC) shared-memory model (see e.g.~\cite{Herlihy:2008,mellor1991algorithms}).
Specifically, 
in the (CC) shared memory model, accesses to shared variables are performed via cached copies of them;
an access to a shared variable is a {\em cache miss} 
if the cached copy of this variable is invalid. In this case, a cache miss occurs
and a valid copy of the variable should be fetched in the local cache first
before it can be accessed.
Once the  cache miss is served
and as long as the variable is not updated by processes that are being executed
on other cores,
future accesses to the variable by processes that are being executed
on this core  do not lead to further cache misses.
In such a model, the {\em inter-island communication complexity} of an instance $inst$ of an operation $op$
is the maximum, over all islands, of the total number of cache-misses that the cores of the island 
experience to execute $inst$. 

We remark that independently of whether the architecture is partially or fully non cache-coherent,
the intra-island communication complexity of $op$ in $\alpha$
is the maximum, over all instances $inst$ of $op$ in $\alpha$, of the 
intra-island communication complexity of $inst$ in $\alpha$.
Moreover, the intra-island communication complexity of $op$ for an 
implementation $I$ is the maximum, over all executions $\alpha$ produced 
by $I$,
of the intra-island communication complexity of $op$ in $\alpha$. 

The {\em DMA communication complexity} of an instance $inst$ of an operation $op$,
is the total number of $DMA$ requests initiated by every process to execute $inst$;
in a more fine-grained model, we could instead measure the total number of bursts performed by these DMA requests. 
The DMA communication complexity of $op$ in $\alpha$
and of $op$ in $I$ are defined as for the other types of communication 
complexities.

\vspace*{.1cm}
\noindent {\bf Time complexity.} 
Consider a fair execution $\alpha$ of an implementation $I$.
A timed version of $\alpha$ is an enhanced version 
of $\alpha$ where each event has been associated to a non-negative real number, 
the time at which that event occurs. 
We define the delay of a message in a timed version of $\alpha$ 
to be the time that elapses between the computation event that sends 
the message and the event that delivers the message.
We denote by ${\cal T}_{\alpha}$ those timed versions of $\alpha$
for which the following conditions hold: 
(1) the times must start at 0, (2) must be strictly increasing 
for each individual process and the same must hold for each individual DMA engine, 
(3) must increase without bound if the execution is infinite,
(4) the timestamps of two subsequent events by the same process (or the same DMA engine) must differ by at most $1$, and
(4) the delay of each message sent must be no more than one time unit.
Let ${\cal T} = \cup_{\forall \alpha \mbox{ produced by } I}\{ {\cal T}_{\alpha} \}$.

The {\em time} until some event $\rho$ is executed in an execution $\alpha$
is the supremum of the times that can be assigned to $\rho$
in all timed versions of $\alpha$ in ${\cal T}_{\alpha}$. 
The {\em time} between two events in $\alpha$ is the supremum 
of the differences between the times in all timed versions of $\alpha$ in ${\cal T}_{\alpha}$.
The {\em time complexity} of an instance $inst$ of an operation $op$ in $\alpha$ 
is the time between the events of its invocation and its response. 
The {\em time complexity} of $op$ in $\alpha$ is the maximum, 
over all instances $inst$ of $op$ in $\alpha$,
of the time complexity of $inst$ in $\alpha$. The time complexity 
of an operation $op$ for $I$ is the maximum, over all executions $\alpha$
produced by $I$, of the time complexity of $op$ in $\alpha$.

\vspace*{.1cm}
\noindent
{\bf Space Complexity.}
The {\em space complexity} of $I$ is determined by the memory overhead introduced by $I$,
and by the number and type of shared variables employed (in case of 
partial non cache-coherence).

\section{Directory-based Stacks, Queues, and Deques}
\label{sec:directory}
The state of the data structure is stored in a highly-scalable distributed 
{\em directory} whose data are spread over the local memory modules of the 
\maxser\ servers. The directory supports the operations \DIRINSERT,  
\DIRDELETE, \BDIRDELETE, and \DIRSEARCH. \DIRDELETE\ and \BDIRDELETE\ 
both remove elements from the directory; however, \DIRDELETE\ returns 
$\bot$ if the requested element is not contained in the directory, 
while \BDIRDELETE\ blocks until the element is found in the directory.
To implement it, we employ a simple 
highly-efficient distributed hash table implementation (also met 
in~\cite{Devine93,hazelcast,shahmirzadiphd}) where hash collisions are 
resolved by using hash chains ({\em buckets}). Each server stores a number 
of buckets. For simplicity, we consider a simple hash function which 
employs $\mod$ and works even if the key is a negative integer. It returns 
an index which is used to find the server where a request must be sent, 
and the appropriate bucket at this server, in which the element resides 
(or must be stored). 

To perform an operation, each client must first access a fetch\&add object 
to get a unique sequence number which it uses as the key for the requested 
data. This object can be implemented using a designated server, called the 
{\em synchronizer} and denoted by $s_s$. The client then communicates with 
the appropriate server to complete its operation. The server locally 
processes the request and responds to the process that initiated it. This 
approach suites better to workloads where the state of the data structure 
is large.

The hash table implementation we use as our directory is presented in 
Section~\ref{app:dht}, for completeness. Similar hash table designs have 
been presented (or discussed) in~\cite{Devine93,shahmirzadiphd,hazelcast}. 
Section~\ref{app:stack-dir} presents the details of the directory-based 
distributed stack. The directory-based queue implementation is discussed 
in Section~\ref{app:queue-dir}. Section~\ref{app:deque-dir} provides the 
directory-based deque. We remark that our directory-based data structures 
would work even when using a different directory implementation.

\subsection{Distributed Hash Table}
\label{app:dht}

A {\em hash table} stores elements, each containing a key and a value (associated with the key). 
%
Each server stores hash table elements to a local data structure. 
This structure can be a smaller hash table 
or any other data structure (array, list, tree, etc.) and supports the operations 
\INSERT, \SEARCH\ and \DELETE.
To perform an operation on the DHT, a client  $c$
finds the appropriate server to submit its request by 
hashing the key value of interest. 
Then, it sends a message to this server, which performs the operation locally 
and sends back  the result to $c$. 
A server $s$ processes all incoming messages sequentially. 
\ignore{
Each message $s$ 
receives has four fields: (1) the $op$ field that denotes the type of the 
operation (\INSERT, \SEARCH\ or \DELETE), (2) the $key$ field that contains 
the key, (3) the $data$ field, which has a value in case of an insert and 
is equal to $\bot$ otherwise, (3) and the $cid$ field which contains the id of the 
client that initiated the transmission of the message. 
Event-driven pseudocode for a server $s$ is described in Algorithm 
\ref{alg20}. Upon receiving a message (line \lref{ln:128}), $s$ checks whether 
its type is \INSERT\ (line \lref{ln:129}), \SEARCH\ (line \lref{ln:131}) or 
\DELETE\ (line \lref{ln:133}). Depending on the type of the request, the server is 
going to invoke the appropriate function each time, and then send the function's 
result back to the client. 

\begin{algorithm}[!ht]
\caption{Events triggered in a hash table server.}
\label{alg20}
\begin{code}
  \firstline
     HashTable $buckets$ = $\varnothing$;                                      \ul
                                                                               \nl
     a message $\langle op, key, data, cid\rangle$ is received: \label{ln:128} \nl
  \n   \If\ ($op$ == \INSERT) $\lbrace$                         \label{ln:129} \nl
  \n     $status$ = insert($buckets, key, data$);               \label{ln:130} \nl
         send($cid, status$);                                                  \nl
  \p   $\rbrace$ \Elseif\ ($op$ == \SEARCH) $\lbrace$           \label{ln:131} \nl
  \n     $status$ = search($buckets, key$);                     \label{ln:132} \nl
         send($cid, status$);                                                  \nl
  \p   $\rbrace$ \Elseif\ ($op$ == \DELETE) $\lbrace$           \label{ln:133} \nl
  \n     $status$ = delete($buckets, key$);                     \label{ln:134} \nl
         send($cid, status$);                                                  \ul
  \p $\rbrace$                                                                     
  \p 
\end{code}
\end{algorithm}

If the request was for an insert (line \lref{ln:129}), the server calls the 
\texttt{insert()} function. This function searches the local buckets for a 
previously inserted element with the same key. If such an element is found, 
\texttt{insert()} returns a negative acknowledgement (\NACK), denoting that 
the new element is already in the hash table. If the key was not found, it stores it and 
returns an acknowledgement (\ACK). The response from \texttt{insert()} 
is returned to the client. 

For the \SEARCH\ and \DELETE\ messages the action sequence is the same. If the 
server receives a \SEARCH, it executes the function \texttt{search()} that 
searches for the key. If it is not found, \texttt{search()} returns \NACK\, and the 
value of the pair, otherwise. If the server receives a \DELETE, it is going to 
execute the function \texttt{delete()} that searches for the key in order to delete 
it. If the key is found, it deletes it and returns \ACK. Otherwise, it returns 
\NACK.

\begin{algorithm}[!ht]
\caption{Insert, search and delete operations of a client of the hash table.}
\label{alg21}
\begin{code}
 \firstline
    \bool\ \DIRINSERT(\integer\ $cid$, Data $data$, \integer\ $key$) $\lbrace$ \nl
 \n   $sid$ = hash\_function($key$);                                    \nl
      send($sid, \langle\INSERT, data, key, cid\rangle$); \label{ln:135}\nl
      $status$ = receive($sid$);                                        \nl
      \return\ $status$;                                \label{ln:135-1}\ul
 \p $\rbrace$                                                           \ul
                                                                        \nl
    \bool\ \DIRSEARCH(\integer\ $cid$, \integer\ $key$) $\lbrace$  \nl
 \n   $sid$ = hash\_function($key$);                                    \nl
      send($sid, \langle\SEARCH, \bot, key, cid\rangle$); \label{ln:136}\nl
      $status$ = receive($sid$);                                        \nl
      \return\ $status$;                                 \label{ln:137} \ul
 \p $\rbrace$                                                           \ul
                                                                        \nl
    \bool\ \DIRDELETE(\integer\ $cid$, \integer\ $key$) $\lbrace$  \nl
 \n   $sid$ = hash\_function($key$);                                    \nl
      send($sid, \langle\DELETE, \bot, key, cid\rangle$); \label{ln:138}\nl
      $status$ = receive($sid$);                                        \nl
      \return\ $status$;                                 \label{ln:139} \ul
 \p $\rbrace$
\end{code}
\end{algorithm}

The \texttt{\DIRINSERT()}, \texttt{\DIRSEARCH()} and 
\texttt{\DIRDELETE()} functions called by the clients are described in 
Algorithm \ref{alg21}. These functions have all the same instructions, but they 
differ in the type of message that the clients send towards the servers, as seen 
on lines \lref{ln:135}, \lref{ln:136}, and \lref{ln:138}, respectively. 
}
After receiving a response message from the server, all client functions return a 
boolean value depending on whether the operation was successful or not.

\subsection{Directory-Based Stack}
\label{app:stack-dir}
%

	\begin{minipage}{.5\textwidth}

		\begin{algorithm}[H]
			\small
			\caption{Events triggered in the synchronizer of the directory-based stack.}
			\label{alg11}
			\begin{code}
				\firstline
				\lreset
				\integer\ $top\_key=-1$;     \label{ln:90}  \ul 
				\nl
				a message $\langle op,\ cid\rangle$ is received: \label{ln:90.1}                      \nl
				\n   \If\ ($op == \PUSH$)                         \label{ln:91}  \nl
				\n     $top\_key++$;                                        \label{ln:92}  \nl
				\p       send($cid$, $top\_key$);                             \label{ln:96}  \nl
				\If\ ($op$ == \POP) $\lbrace$         \label{ln:96.1}     \nl 
				\n     \If\ ($top\_key \neq -1$)                   \label{ln:97} \nl
				$top\_key--$;                                      \label{ln:99}  \ul           
				\p     $\rbrace$                                                           
				\p    
			\end{code}
		\end{algorithm}
	\end{minipage}

To implement a stack, the synchronizer $s_s$ maintains a 
variable $top$ which stores 
an integer counting the number of elements that are currently in the stack. 

To apply an operation {\em op} a client sends a message to the synchronizer $s_s$. 
If {\em op} is a push operation, $s_s$ uses $top\_key$ variable to assign unique 
keys to the newly inserted data. Each time it receives a push request, $s_s$ 
sends the value stored in $top\_key$ to the client after incrementing it by one. 
Once a client receives a key from $s_s$ for the push operation it has initiated, 
it inserts the new element in the directory by invoking \DIRINSERT. 
Similarly, if {\em op} is a pop operation, $s_s$ sends the value stored in $top\_key$ 
to the client and decrements it by one. The client then invokes \DIRDELETE\ repeatedly, 
until it successfully removes from the directory the element with the received key. 
Notice that keys of elements are greater than or equal to $0$ and 
therefore, $top\_key$ has initial value $-1$, which indicates that 
the stack is empty.

Event-driven pseudocode for the synchronizer is described in Algorithm~\ref{alg11} 
and the code for the \texttt{ClientPush()} and \texttt{ClientPop()} operations, is 
presented in Algorithms~\ref{alg12}~and~\ref{alg13}.

\begin{figure}
	\begin{minipage}{.40\textwidth}\centering
		\begin{algorithm}[H]
			\small
			\caption{Push operation for a client of the directory-based stack.}
			\label{alg12}
			\begin{code}
				\firstline
				\void\ ClientPush(\integer\ $cid$, Data $data$) $\lbrace$ \nl
				\n   
				send($sid, \langle\PUSH, cid\rangle$);   \label{ln:47}  \nl
				$key$ = receive(sid);                  \label{ln:103.1} \nl
				$status$ = \DIRINSERT($key, data$);      \label{ln:103} \nl
				\return\ $status$;                                      \ul
				\p $\rbrace$
			\end{code}
		\end{algorithm}	
	\end{minipage}
\hfill
	\begin{minipage}{.40\textwidth}\centering
		\begin{algorithm}[H]
			\small
			\caption{Pop operation for a client of the directory-based stack.}
			\label{alg13}
			\begin{code}
				\firstline
				Data ClientPop(\integer\ $cid$) $\lbrace$            \nl
				\n     
				send($sid, \langle\POP, cid\rangle$);\label{ln:106}\nl
				$key$ = receive($sid$);     \label{ln:104}         \nl
				\If\ ($key$ == $-1$) 
				$status$ = $\bot$               \label{ln:104.2} \nl
				\Else\ $\lbrace$        \label{ln:104.3} \nl
				\n         \Do\ $\lbrace$              \label{ln:100.1}   \nl
				\n         	$status$ = \DIRDELETE($key$); \label{ln:100}   \nl
				\p     	   $\rbrace$ \While\ ($status == \bot$);\label{ln:101}\ul
				\p     $\rbrace$               			 \label{ln:104.4} \nl
				\return\ $status$;          \label{ln:105}         \ul
				\p $\rbrace$
			\end{code}
		\end{algorithm}
	\end{minipage}
\end{figure}

The synchronizer receives, processes, and responds to clients' messages. 
The messages have an $op$ field that represents the operation to be 
performed (\PUSH\ or \POP), and a $cid$ field with the client's 
identification number, needed for identifying the appropriate client to 
communicate with. 

The \texttt{ClientPop()} function, presented in Algorithm~\ref{alg13} is 
analogous to the push operation: it sends a \POP\ message to $s_s$ and waits 
for its response (line~\lref{ln:104}). Using the key that was received as 
argument,  \DIRDELETE() is repeatedly called (line~\lref{ln:100}). This is 
necessary since another client responsible for inserting the $key$ may not 
have finished yet its insertion. In this case \DIRDELETE\ returns $\bot$ 
(line~\lref{ln:101}). However, since the $key$ was generated previously 
by $s_s$, it is certain that it will be eventually inserted into the 
directory service.

\subsubsection{Proof of Correctness}
Let $\alpha$ be an execution of the directory-based stack implementation.
We assign linearization points to push and pop operations in $\alpha$ 
by placing the linearization point of an operation $op$ in the configuration 
resulting from the execution of line~\lref{ln:96} by the synchronizer for 
$op$. 

Let $op$ be a push or a pop operation invoked by a client $c$ in $\alpha$ 
and assume that the synchronizer executes line~\lref{ln:96} for it. By inspection 
of the pseudocode, we have that this line is executed after the synchronizer 
receives a message by $c$ (line~\lref{ln:90.1}) and before $c$ receives the 
synchronizer's response (line~\lref{ln:103.1} for a push operation, 
line~\lref{ln:104} for a pop operation). By the way the linearization points 
are assigned, we have the following lemma.

\begin{lemma}
The linearization point of a push (pop) operation 
$op$ is placed within its execution interval.
\end{lemma}

\ignore{
	as 
	follows: 
	The linearization point of a push operation $op$ is placed in 
the configuration resulting from the execution of line~\lref{ln:103} for 
$op$ by the client that invoked it. If $op$ is a pop operation and 
line~\lref{ln:98} is executed for $op$ by $s_s$, then the linearization 
point is placed in the resulting configuration. If $op$ is a pop operation 
for which line~\lref{ln:102} is executed by $s_s$, then we distinguish 
two cases. Let $op'$ be that push operation, which inserts into the 
directory the element that $op$ removes. If the linearization point of 
$op'$ occurs before or at the execution of line~\lref{ln:102} for $op$, 
then $op$ is linearized in the configuration resulting from the execution 
of this line. Otherwise, the linearization point of $op$ is placed 
right after the linearization point of $op'$.}


Denote by $L$ the sequence of operations (which have been assigned 
linearization points) in the order determined by their linearization 
points. Let $C_i$ be the configuration in which the $i$-th operation 
$op_i$ of $L$ is linearized. Denote by $\alpha_i$, the prefix of 
$\alpha$ which ends with $C_i$ and let $L_i$ be the prefix of $L$ 
up until the operation that is linearized at $C_i$. We denote by 
$top_i$ the value of the local variable $top\_key$ of $s_s$ at 
configuration $C_i$; 
Let $C_0$ be the initial configuration an let $top_0 = -1$. 

Notice that since only $s_s$ executes Algorithm \ref{alg11}, 
we have the following.
\begin{observation}
	\label{obs:seq-dstack}
	Instances of Algorithm~\ref{alg11} are executed 
	sequentially, i.e. their execution does not overlap.
\end{observation}

\ignore{
\begin{proof}
Inspection of the pseudocode easily shows that the claim holds 
for push operations, as the execution of the line after which 
the linearization point is placed, takes place after the invocation 
and before the response of the operation. 

Assume now that $op$ is a pop operation invoked by client $c$ and 
assume that $op$ removes an element with key $k$ from the directory. 
Let $op'$ be the push operation that inserts this element in the 
directory. Let $C$ be the configuration in the first \Do-\While\ 
loop iteration of lines~\lref{ln:100.1}~-~\lref{ln:101}, in which 
the execution of \DIRDELETE\ does not return $\bot$. Let $C'$ be 
the configuration resulting from the execution of \DIRINSERT\ on
line~\lref{ln:103} by $op'$, after which the element with key $k$ 
is inserted in the directory by $op'$. We consider two cases. 

First, assume that $C'$ precedes the execution of line~\lref{ln:102} 
for $op$ by $s_s$. In this case, the linearization point of $op$ is 
placed in the configuration resulting from the execution of 
line~\lref{ln:102} for $op$ by $s_s$. Inspection of the pseudocode 
shows that this line is executed by $s_s$ for $op$ after $s_s$ 
receives from $c$ the message that is sent by executing line~\lref{ln:106}, 
i.e. after {\tt ClientPop} is invoked. Further inspection shows that 
$c$ blocks (line~\lref{ln:104}) until it receives from $s_s$ the message 
sent on line~\lref{ln:102}. This means that {\tt ClientPop}, and therefore, 
$op$, does not respond before line~\lref{ln:102} is executed. The above 
implies that the linearization point of $op$ is included in its execution 
interval.  

Assume next that  $C'$ follows the execution of line~\lref{ln:102} 
for $op$ by $s_s$. Following the same argumentation as for the previous 
case, we have that the execution of that line occurs in the execution 
interval of $op$. From the definitions of $C$ and $C'$, we further have 
that $C'$ happens before $C$, since the element that $op'$ inserts 
in the directory by using \DIRINSERT, is the element that $op$ 
removes from the directory in $C$. Recall that by the way that the 
linearization points are assigned, the linearization point of $op'$ 
is placed in $C'$. Since $C$ is included in the execution interval 
of $op$ and $C'$ occurs after the execution of line~\lref{ln:102} 
and before $C$, and given that the linearization point of $op$ is 
in this case also placed in $C'$, it follows that the linearization 
point for $op$ is included in its execution interval.

Notice that the argument for the case where $op$ receives a response 
from $s_s$ because $s_s$ executes line~\lref{ln:98}, is analogous with 
the case where $C'$ precedes the execution of line~\lref{ln:102} for 
$op$ by $s_s$.

Thus, the claim holds for all cases.
\end{proof}
}
%
%
%




\ignore{
Further inspection of the pseudocode of Algorithm~\ref{alg11} indicates 
that the value of $top\_key$ is incremented before an element is inserted 
into the directory and decremented before one is removed from the directory. 
Furthermore, by Algorithm~\ref{alg13}, we see that if a client pop operation 
receives $\bot$ from the synchronizer in line~\lref{ln:104.2}, then the directory 
is not modified.
This implies the following observation.
\begin{observation}
\label{obs:empty-dstack}
The value of $top\_key$ is equal to $-1$ when as many elements have been 
inserted in the directory as have been removed. 
\end{observation}
}

By the way linearization points are assigned, further inspection of 
the pseudocode in conjunction with Observation~\ref{obs:seq-dstack} 
leads to the following. 
\begin{observation}
\label{obs:keys}
For each integer $i \geq 1$, the following hold at $C_i$:
Let $op_i$ and $op_{i+1}$ be two operations in $L$. 
\begin{itemize}
\item\label{st0} If $op_i$ is a push operation and $op_{i+1}$ is a push operation, then 
$top_{i+1} = top_i + 1$.
\item\label{st1} If $op_i$ is a push operation and $op_{i+1}$ is a pop operation, then 
$top_{i+1} = top_i$.
\item\label{st2} If $op_i$ is a pop operation and $op_{i+1}$ is a pop operation, then, 
if $top_i \neq -1$, $top_{i+1} = top_i - 1$.
\item\label{st3} If $op_i$ is a pop operation and $op_{i+1}$ is a push operation, then, 
if $top_i \neq -1$, $top_{i+1} = top_i$.
\end{itemize}
\end{observation}

Let $\sigma$ be the sequential execution that results by executing the 
operations in $L$ sequentially, in the order of their linearization points, 
starting from an initial configuration in which the queue is empty. 
Let $\sigma_i$ be the prefix of $\sigma$ that contains the operations in $L_i$.
Denote by $S_i$ the state of the sequential stack that results if the 
operations of $L_i$ are applied sequentially to an initially empty stack. 
Denote by $d_i$ the number of elements in $S_i$. We associate a sequence 
number with each stack element such that the elements from the bottommost 
to the topmost are assigned $0, \ldots, d_i-1$, respectively. Denote by 
$sl_{d_i}$ the $d_i$-th element of $S_i$.

Let $L'_i$ be the projection of $L_i$ that contains all operations in $L_i$ 
except those pop operations that return $\bot$. Denote those operations by 
$op'_i$. Denote by $C'_i$ the configuration in which the $i$-th operation of 
$L'_i$ is linearized. Let $S'_i$ be the state of the sequential stack after 
the operations of $L'_i$ have been applied to it, assuming an initially empty 
stack and denote by $d'_i$ the number of elements in $S'_i$. Again we associate 
a sequence number with each stack element such that the elements from the 
bottommost to the topmost are assigned $0, \ldots, d'_i-1$, respectively. 
Denote by $sl_{d'_i}$ the $d'_i$-th element of $S'_i$. We denote by 
$top'_i$ the value of the local variable $top\_key$ of $s_s$ at 
configuration $C'_i$. Let $k_i$ be the number of pop operations in $L'_i$.

By inspection of the pseudocode, it follows that:

\begin{observation}
	\label{obs:push-pop}
	If $op_i$ is a push operation, it inserts a pair $\langle key, data \rangle$ in the directory, 
	where $data$ is the argument of \CPush\ executed for $op_i$, and $key = top_i$.
	If $op_i$ is a pop operation then, if $top_i \neq -1$, it removes a pair $\langle key, data \rangle$ with $key = top_{i-1}$ from the directory and returns $data$; if $top_i = -1$, it does not remove any pair from the directory and returns $\bot$. 
\end{observation}

\begin{lemma}
For each integer $i > 0$, it holds that: 
\begin{itemize}
\item If $op'_i$ is a push operation, then it inserts element $d'_i$ 
into the stack and $top'_i = d'_i-1$.
\item If  $op'_i$ is a pop operation, then it removes element $d'_{i-1}$ 
from the stack, $top'_i = d'_{i-1} - 1$ and $top'_i$ is equal to the value of 
$top\_key$ of $s_s$ for the ${i - 2*k_i + 1}$-th push operation in $L'_i$ .
\end{itemize}
\end{lemma}

\begin{proof}
We prove the claim by induction on $i$. 

{\bf Base case.} We prove the claim for $i = 1$. 
Recall that $S'_0 = \epsilon$ and $top'_0 = -1$.
By definition of $L'_i$, $op'_1$ is a push operation. By observation of 
the pseudocode (line~\lref{ln:92}) we have that $top'_1 = -1 + 1 = 0$ 
and $op'_1$ inserts an element with sequence number $d'_1-1 = 0$ into 
the stack. Thus, the claim holds. 

{\bf Hypothesis.} Fix any $i$, $i > 0$, and assume that the claim 
holds for $C'_i$. 

{\bf Induction step.} We prove that the claim also holds at $C'_{i+1}$.
First, consider the case where $op'_{i+1}$ is a push operation. We distinguish 
two cases. Case (i): $op'_i$ is a push operation, as well. For $op'_i$, the 
induction hypothesis holds. Therefore, $op'_i$ inserts the $d'_i$-th element 
into the stack and $top'_i = d'_i - 1$. By Observation~\ref{obs:keys}, 
$top'_{i+1} = top'_i + 1$. Since $S'_i$ has $d'_i$ elements, $op'_{i+1}$ 
inserts element $d'_{i+1} = d'_i + 1$, with sequence number $d'_{i+1} - 1  = 
(d'_i + 1) - 1 = d'_i = top'_i + 1 = top'_{i+1}$ and the claim holds.
Case (ii): $op'_i$ is a pop operation, which removes an element from the 
stack. Since the induction hypothesis holds, $top'_i = d'_{i-1} - 1$. 
Conversely, $op'_{i+1}$ inserts an element into the stack. By 
Observation~\ref{obs:keys}, we have that $top'_{i+1} = top'_i$ and the 
claim holds.

\end{proof}

\ignore{
Further inspection of the pseudocode of Algorithm~\ref{alg11} indicates 
that the value of $top\_key$ is incremented before an element is inserted 
into the directory and decremented before one is removed from the directory. 
Furthermore, by Algorithm~\ref{alg13}, we see that if a client pop operation 
receives $\bot$ from the synchronizer in line~\lref{ln:104.2}, then the directory 
is not modified. This implies the following observation. 
\begin{observation}
\label{obs:no-mod}
A pop operation that returns $-1$ does not modify the state of the stack. 
\end{observation}

\begin{lemma}
For each integer $i > 0$, it holds that if $op_i$ is a pop operation 
that returns $\bot$, then $S_{i-1} = \epsilon$.
\end{lemma}

\begin{proof}
Initially, notice that, by the pseudocode of Algorithm~\ref{alg13}, 
we have that a pop operation returns $\bot$ only if it receives 
a value of $top\_key = -1$ from the synchronizer. We now prove the 
claim by induction on $i$. 

{\bf Base case.} We prove the claim for $i = 1$. 
Since $S_0 = \epsilon$ and $top_0 = -1$, the claim holds trivially.

{\bf Hypothesis.} Fix any $i$, $i > 0$ and assume that the claim 
holds for all $C_j$, $j \leq i$. 

{\bf Induction step.} We prove that the claim also holds at $C_{i+1}$.
By Observation~\ref{obs:seq-dstack}, we have that the value of $top\_key$ 
is not modified between the executions of $op_i$ and $op_{i+1}$. 
Assume first that $op_i$ is a push operation. By observation of the 
pseudocode of Algorithm~\ref{alg11}, we have that when the synchronizer 
executes it for $op_i$, it increments $top\_key$ before sending $top_i$ 
to the client that invoked $op_i$. Since $top\_key$ is not modified 
before $C_{i+1}$, $op_{i+1}$ receives $top_{i+1} \neq -1$ from the 
synchronizer, a contradiction to the assumption that it returns $\bot$, 
since, for this to occur, it must receive $top_{i+1} = -1$. 
Thus the claim holds. 

Assume now that $op_i$ is a pop operation. By the induction hypothesis, 
we have that at the claim holds at $C_i$. We proceed by case analysis. 
Assume first that $op_i$ return $\bot$. In this case, by the hypothesis, 
$S_{i-1} = \epsilon$. By Observation~\ref{obs:no-mod}, we have that 
$op_i$ does not modify the state of the stack. By observation of the 
pseudocode of Algorithm~\ref{alg11}, we have that the value of $top\_key$ 
is not modified by the synchronizer for $op_i$. Therefore, $top_{i+1} = -1$ 
and $S_i = S_{i-1} = \epsilon$ and the claim holds. Assume now that 
$op_i$ returns a value other than $\bot$. Inspection of the pseudocode of 
the client pop operation in Algorithm~\ref{alg13} indicates that this occurs 
when the client that invoked $op_i$ receives a value for $top\_key$ other 
than $-1$ from the synchronizer. Inspection of the pseudocode of the 
synchronizer in Algorithm~\ref{alg11} (lines~\lref{ln:96.1}~-~\lref{ln:99}) 
indicates that after a value other than $-1$ is sent to a client for a 
pop operation, the value of $top\_key$ is decremented. 
\end{proof}

%
%
%

\begin{lemma}
For each integer $i > 0$, it holds that if $op_i$ is a pop operation, 
then it returns the value of the field $data$ of $sl_{d_{i-1}}$ if 
$S_{i-1} \neq \epsilon$, or $\bot$ if $S_{i-1} = \epsilon$.
\end{lemma}

\begin{proof}
We prove the claim by induction on $i$. 

{\bf Base case.} 
We prove the claim for $i$ = $1$. Recall that at $C_0$, since no 
operation has been linearized, the equivalent sequential stack is 
empty. Recall also that at $C_0$ it holds that $top\_key = -1$. 
If $op_1$ is a push operation, the claim holds trivially. 
Let then $op_1$ be a pop operation. 
Since $op_i$ is the first operation to be linearized in $L$, $L_1 = op_1$ 
and $S_1 = \epsilon$. Since $op_1$ is the first operation linearized in $L$, 
then by Observation~\ref{obs:seq-dstack} and the way linearization points 
are assigned, we have that $op_1$ is the first time in $\alpha$ that 
Algorithm~\ref{alg11} is executed. Therefore, and by inspection of 
Algorithm~\ref{alg11}, we have that $top_0 = -1$. Since $top\_key$ is 
not modified by Algorithm~\ref{alg11} for $op_1$, $top_1 = top_0 = -1$.
By inspection of the client pseudocode in Algorithm~\ref{alg13} 
(lines~\lref{ln:104}~-~\lref{ln:104.2}), we see that when a pop 
operation receives $-1$ from $s_s$, it returns $\bot$ to the 
client. Thus, the claim holds. 

{\bf Hypothesis.} Fix any $i$, $i > 0$ and assume that the claim 
holds for all $C_j$, $j \leq i$. 

{\bf Induction step.} We prove that the claim also holds at $C_{i+1}$.
If $op_{i+1}$ is a push operation, the claim holds trivially. Let then 
$op_{i+1}$ be a pop operation. We proceed by case analysis. 

First, assume that $op_{i+1}$ is linearized after the execution of 
line~\lref{ln:98} by $s_s$. This implies that in the configuration 
in which $s_s$ evaluates the \If\ condition of line~\lref{ln:97}, 
it evaluates to \true. By Observation~\ref{obs:empty-dstack}, this 
means that for each push operation that has been linearized up to 
that configuration, there has been a matching pop operation that 
has been linearized as well. It follows that $S_i$ is empty and 
that the claim holds. 

Next, assume that $op_{i+1}$ is linearized in the configuration right 
after the execution of line~\lref{ln:102} by $s_s$. By definition, this 
means that $op_{i+1}$ removes an element from the directory that has 
been inserted into the directory by a push operation $op_j$, $j \leq i$, 
which has been linearized before the execution of this line, due to 
the way linearization points are assigned. We distinguish two cases. 

First assume that $op_i$ is a push operation and assume that $k_i$ 
is the value of $top\_key$ that it has received by $s_s$, i.e., $op_i$ 
inserts into the directory an element with key $k_i$. Since $op_i$ is 
linearized before the execution of line~\lref{ln:102} by $s_s$ for 
$op_{i+1}$ and by Observation~\ref{obs:seq-dstack}, we have that at 
the end of the execution of the instance of Algorithm~\ref{alg11} 
by $s_s$ for $op_i$, it holds that $top\_key = k_i$. Inspection of 
Algorithm~\ref{alg11} shows that a pop operation that follows a 
push operation receives the same value of $top\_key$ as the one 
that was sent to the push operation. Therefore, if no further instance
of Algorithm~\ref{alg11} is executed for some other operation by $s_s$ 
after it executes it for $op_i$ and before it executes it for $op_{i+1}$, 
then the claim follows straight-forwardly. Assume now that between 
$C_i$ and $C_{i+1}$, more instances of Algorithm~\ref{alg11} are 
executed by $s_s$ for other operations. Let $op'$ be that out of those 
operations for which Algorithm~\ref{alg11} is executed last before $C_{i+1}$ 
and assume that it is a push. Let $k'$ be the value of $top\_key$ 
at the end of this instance of Algorithm~\ref{alg11}. Then, at $C_{i+1}$, 
$s_s$ sends $k'$ to the client that invoked $op_{i+1}$. Then this client 
attempts to remove from the directory an element with key $k'$. However, 
since there is no further operation linearized between $C_i$ and $C_{i+1}$, 
this element is not in the directory at $C_{i+1}$. Thus, the push operation 
that inserts in the directory the value which $op_{i+1}$ removes, is 
linearized after $C_{i+1}$ -- a contradiction to the definition of linearization 
points. If $op'$ is a pop operation and it receives $k'$ as the value of 
$top\_key$ from $s_s$, then $op_{i+1}$ receives $k'-1$ as value of $top\_key$.
Then, $op_{i+1}$ attempts to remove from the directory an element with key $k'-1$.
Let $op''$ be the push operation that inserts an element with this key. 
If $op''$ is linearized after $C_{i+1}$, once more we arrive at a contradiction. 
If $op''$ is linearized before $C_{i+1}$, then by the induction hypothesis, 
implies that each of the pop operations between $C_i$ and $C_{i+1}$ removes 
the top-most element of the sequential stack. Thus, at $C_{i+1}$, the element 
inserted by $op''$ is the top-most one and the claim holds. 

Finally, assume that $op_{i+1}$ is linearized right after the linearization 
point of that push operation $op'$ whose value it removes from the directory. 
In this case, since no further operation is linearized between $op_{i+1}$ 
and $op'$, this means that the value inserted by $op'$ is indeed the top-most 
of $S_i$ when it is removed by $op_{i+1}$ and the claim holds. 
\end{proof}
}

\ignore{
\begin{lemma}
For each integer $i > 0$, the following hold:
\begin{compactenum}
\item\label{c1} If $op_i$ is a push operation, it holds that $top_i$ = $top_{i-1} + 1$. 
\item\label{c2} If $op_i$ is a pop operation,  
$top_i$ = $top_{i-1} - 1$ if $top_{i-1} \geq 0$, or $top_i = -1$ otherwise.
\item \label{c3} If $S_i \neq \epsilon$, $sl_{d_i}.key = top_i$.
\item \label{c4} If $op_i$ is a pop operation, then it returns the value of the field $data$ of $sl_{d_{i-1}}$ or \NACK\ if $S_{i-1} = \epsilon$.
\end{compactenum}
\end{lemma}

\begin{proof}
We prove the claims by induction on $i$. 

{\bf Base case.} 
We prove the claim for $i$ = $1$. Let $op_1$ be a push operation. 
By the pseudocode, $top_0 = -1$ (line \lref{ln:90}). By definition, 
$op_1$ is the first operation to be linearized in $\alpha$, therefore 
it is linearized at $C_1$. Notice that $top\_key$ is incremented 
before the step that brings about $C_1$, i.e. it is incremented before 
line \lref{ln:96} is executed. Since $op_1$ is the first operation to 
be linearized in $\alpha$, and since the initial value of $top\_key$ 
is $-1$, this implies that the value $0$ is sent to the client 
$c$ that invoked $op_1$. Thus, by definition of the linearization points, 
at $C_1$, $top_1 = 0 = (-1)+1 = top_0 + 1$, i.e. Claim \ref{c1} holds. 
Furthermore, by definition of $S_1$, it only contains one element at $C_1$, 
namely the value of $top\_key$ sent to $c$. Thus, $sl_{d_1}.key = top_1$, 
i.e. Claim \ref{c3} holds. 

Now let $op_1$ be a push operation. Since $op_1$ is the first operation 
to be linearized in $\alpha$, $top\_key$ has its initial value when line 
\lref{ln:97} is executed by $s_s$ for $op_1$. By definition, $top_0 = -1$, 
thus Claim \ref{c2} holds trivially. By observation of the pseudocode lines 
\lref{ln:104}-\lref{ln:105}, $op_1$ returns the value sent to $c$ from the 
execution of line \lref{ln:98}. Since $S_0 = \epsilon$, and $op_1$ is not 
successful, $S_i = \epsilon$, and Claims \ref{c3} and \ref{c4} hold.

{\bf Hypothesis.} Fix any $i$, $i > 0$ and assume that the lemma holds at $C_i$. 

{\bf Induction step.} We prove that the claims also holds at $C_{i+1}$. 

Let $op_{i+1}$ be a push operation. Variable $top\_key$ is incremented before 
it is sent to $c$ that requested the push operation (lines \lref{ln:92}-\lref{ln:96}).
Thus, the value that $c$ receives (line \lref{ln:103.1}) and stores in local variable 
$key$, is $top_i + 1$. Then it holds that $top_{i+1}$ = $top_i+1$, thus Claim \ref{c1} 
holds. 

In case that $op_{i+1}$ is a pop operation and $top_i \geq 0$, then the value of $top_i$ 
is reduced by one (line \lref{ln:99}). Thus, we have that $top_{i+1}$ = $top_i - 1$ and 
Claim \ref{c2} holds. In case that $top_i = -1$, then it holds that $top_{i-1} = top_{i} = -1$, 
as in this case, the value of $top\_key$ is not modified (line \lref{ln:97}-\lref{ln:98}) 
and thus Claim~\ref{c2} holds.

We now prove Claim~\ref{c3}. In case that $op_{i+1}$ is a push operation, the 
pseudocode(line~\lref{ln:96}) and Claim~\ref{c1} imply that $s_s$ sends a value 
equal to $top_{i+1} = top_i + 1$ to $c$. The pseudocode (line~\lref{ln:103})
implies that the client $c$ inserts a pair of $\langle top_{i+1}, data \rangle$ 
in the directory. Thus, Claim~\ref{c3} holds. In case that $op_{i+1}$ is a pop 
operation, the induction hypothesis implies that Claim~\ref{c3} holds.

We now prove, Claim~\ref{c4}. By the semantics of \DIRDELETE, if at the point 
that the instance of \DIRINSERT\ is executed in the \Do\ - \While\ loop of lines 
\lref{ln:99}-\lref{ln:101} for $op_{i+1}$, the instance of \DIRINSERT\ of $op_{i+1}$ 
has not yet returned, then \DIRDELETE\ returns $\langle \bot, -\rangle$. Notice 
that the parameter of \DIRDELETE\ is $top_i$. By induction hypothesis (Claim~\ref{c3}), 
$top_i$ is the $key$ of the last pair $sl_{d_i}$ in $S_i$. Therefore, when \DIRDELETE\ 
returns a $status \neq \bot$, it holds that it returns the $data$ field of the last 
element in $S_i$ and that it is sent to $c$ and is used as the return value of $op_{i+1}$.
Thus Claim \ref{c4} holds. 
\end{proof}
}

From the above lemmas we have the following.

\begin{theorem}
The directory-based distributed stack implementation is linearizable. 
\end{theorem}

\ignore{
\begin{lemma}
For each integer $i > 0$, the following hold:
\begin{enumerate}
\item\label{c1} If $op_i$ is a push operation, it holds that $top_i$ = $top_{i-1} + 1$. 
\item\label{c2} If $op_i$ is a pop operation,  
$top_i$ = $top_{i-1} - 1$ if $top_{i-1} > 0$, or $top_i = 0$ otherwise.
\item \label{c3} $sl_i.key = top_i$.
\item \label{c4} If $op_i$ is a pop operation, then it returns the value of the field $data$ of $sl_{i}$.
\end{enumerate}
\end{lemma}

\begin{proof}[Proof Sketch.]
We prove the claims by induction on $i$. 

{\bf Base case.} 
For $i = 0$ the claim holds trivially.

{\bf Hypothesis.} Fix any $i-1$, $i > 0$ and assume that the lemma holds 
for any $j$, $1 \geq j \geq i-1$.

{\bf Induction step.} We prove that the claims also holds for $i$. 

In case that $op_{i}$ is a push operation. Variable $top\_key$ is incremented before 
it is sent to $c$ that requested the push operation (lines \lref{ln:92}-\lref{ln:96}).
Thus, the value that $c$ receives (line \lref{ln:103.1}) and stores in local variable $key$, 
is $top_i$. It holds that $top_{i}$ = $top_i$, thus Claim \ref{c1} holds. 

In case that $op_{i+1}$ is a pop operation and $top_i > 0$, then the value of $top_i$ 
is is reduced by one. Thus, we have that $top_{i}$ = $top_i - 1$ and Claim \ref{c1} holds. 
In case that $top_i = 0$, then it holds that $top_{i-1} = top_{i} = 0$ and thus Claim~\ref{c2} 
holds.

We now prove Claim~\ref{c3}. In case that $op_{i}$ is a push operation, the pseudocode
(line~\lref{ln:96}) and Claim~\ref{c1} implies that the synchronizer server sends 
a value equal to $top_i = top_{i-1} + 1$ to $c$. The pseudocode (line~\lref{ln:103})
implies that the client $c$ inserts a pair of $<top_i , data>$ in the directory. 
Thus, Claim~\ref{c3} holds.
In case that $op_i$ is a pop operation, induction hypothesis implies that Claim~\ref{c3}
holds.

We now prove, Claim~\ref{c4}. By the semantics of \DIRDELETE, if at the point that the 
instance of \DIRINSERT\ is executed in the \Do\ - \While\ loop of lines \lref{ln:99}-\lref{ln:101} 
for $op_{i}$, the instance of \DIRINSERT\ of $op_{i}$ has not yet returned, then \DIRDELETE\ returns 
$\langle \bot, -\rangle$. Notice that the parameter of \DIRDELETE\ is $top_i$. 
By induction hypothesis (Claim~\ref{alg11}(3c)), $top_i$ is the $key$ of the last pair 
$sl_i$ in $S_i$. Therefore, when \DIRDELETE\ returns a $status \neq \bot$, 
it holds that it returns the $data$ field of the last element 
in $S_i$ and this is sent to $c$ and is used as the return value of $op_{i+1}$.
Thus Claim \ref{c3} holds. 
\end{proof}
}

\ignore{
\vspace*{.2cm}
\noindent
{\bf Proof of Correctness.}
For simplicity, we present in this section a proof where the communication 
with the directory in the case of pop operations is performed by the synchronizer.
Notice that in that case, lines \lref{ln:100.1}-\lref{ln:101} are executed 
by $s_s$.
Let $\alpha$ be an execution of the directory-based stack implementation.
We assign linearization points to push and pop operations in $\alpha$ as follows: 
The linearization point of a push operation $op$ is placed in the configuration 
resulting from the execution of line~\lref{ln:96} for $op$ by $s_s$.
The linearization point of a pop operation $op$ is placed in the configuration resulting 
from the execution of either line~\lref{ln:98} or line~\lref{ln:102} for $op$ (whichever 
is executed) by $s_s$. Recall that in this case, when $s_s$ executes line~\lref{ln:102}, 
it sends the client the element that was removed from the directory.
Denote by $L$ the sequence of operations (which has been assigned linearization points)
in the order determined by their linearization points.

\begin{lemma}
The linearization point of a push (pop) operation $op$ is placed within its execution interval.
\end{lemma}

\begin{proof}
Assume that $op$ is a push operation and let $c$ be the client that invokes it.
After the invocation of $op$, $c$ sends a message to $s_s$ (line \lref{ln:47}) 
and awaits a response from it. 
Recall that routine {\tt receive()} (line~\lref{ln:103.1}) blocks until 
a message is received.  
The linearization point of $op$ is placed at the configuration resulting 
from the execution of line \lref{ln:96} for $op$ by $s_s$. This line is executed after 
the request by $c$ is received, i.e. after $c$ invokes {\tt ClientPush}. Furthermore, it 
is executed before $c$ receives the response by the server and thus, before {\tt ClientPush} 
returns. Therefore, the linearization point is included in the execution interval of push.

The argumentation regarding pop operations is similar.
\end{proof}



Let $C_i$ be the configuration in which the $i$-th operation $op_i$ of $L$ is linearized. 
Denote by $\alpha_i$, the prefix of $\alpha$ which ends with $C_i$
and let $L_i$ be the prefix of $L$ up until the operation that is linearized at $C_i$.
We denote by $top_i$ the value of the local variable $top\_key$ of $s_s$
at configuration $C_i$; let $top_0 = 0$.
Denote by $S_i$ the sequential stack that results if the operations of $L_i$ are
applied sequentially to an initially empty stack. Denote by $d_i$ the number of 
elements in $S_i$. We consider a sequence number with each element of the stack 
such that the elements from the bottommost to the topmost are assigned $1, \ldots, d_i$, 
respectively. Denote by $sl_{d_i}$ the $d_i$-th element of $S_i$.

\begin{lemma}
For each integer $i > 0$, the following hold:
\begin{compactenum}
\item\label{c1} If $op_i$ is a push operation, it holds that $top_i$ = $top_{i-1} + 1$. 
\item\label{c2} If $op_i$ is a pop operation,  
$top_i$ = $top_{i-1} - 1$ if $top_{i-1} \geq 0$, or $top_i = -1$ otherwise.
\item \label{c3} If $S_i \neq \epsilon$, $sl_{d_i}.key = top_i$.
\item \label{c4} If $op_i$ is a pop operation, then it returns the value of the field $data$ of $sl_{d_{i-1}}$ or \NACK\ if $S_{i-1} = \epsilon$.
\end{compactenum}
\end{lemma}

\begin{proof}
We prove the claims by induction on $i$. 

{\bf Base case.} 
We prove the claim for $i$ = $1$. Let $op_1$ be a push operation. 
By the pseudocode, $top_0 = -1$ (line \lref{ln:90}). By definition, 
$op_1$ is the first operation to be linearized in $\alpha$, therefore 
it is linearized at $C_1$. Notice that $top\_key$ is incremented 
before the step that brings about $C_1$, i.e. it is incremented before 
line \lref{ln:96} is executed. Since $op_1$ is the first operation to 
be linearized in $\alpha$, and since the initial value of $top\_key$ 
is $-1$, this implies that the value $0$ is sent to the client 
$c$ that invoked $op_1$. Thus, by definition of the linearization points, 
at $C_1$, $top_1 = 0 = (-1)+1 = top_0 + 1$, i.e. Claim \ref{c1} holds. 
Furthermore, by definition of $S_1$, it only contains one element at $C_1$, 
namely the value of $top\_key$ sent to $c$. Thus, $sl_{d_1}.key = top_1$, 
i.e. Claim \ref{c3} holds. 

Now let $op_1$ be a push operation. Since $op_1$ is the first operation 
to be linearized in $\alpha$, $top\_key$ has its initial value when line 
\lref{ln:97} is executed by $s_s$ for $op_1$. By definition, $top_0 = -1$, 
thus Claim \ref{c2} holds trivially. By observation of the pseudocode lines 
\lref{ln:104}-\lref{ln:105}, $op_1$ returns the value sent to $c$ from the 
execution of line \lref{ln:98}. Since $S_0 = \epsilon$, and $op_1$ is not 
successful, $S_i = \epsilon$, and Claims \ref{c3} and \ref{c4} hold.

{\bf Hypothesis.} Fix any $i$, $i > 0$ and assume that the lemma holds at $C_i$. 

{\bf Induction step.} We prove that the claims also holds at $C_{i+1}$. 

Let $op_{i+1}$ be a push operation. Variable $top\_key$ is incremented before 
it is sent to $c$ that requested the push operation (lines \lref{ln:92}-\lref{ln:96}).
Thus, the value that $c$ receives (line \lref{ln:103.1}) and stores in local variable 
$key$, is $top_i + 1$. Then it holds that $top_{i+1}$ = $top_i+1$, thus Claim \ref{c1} 
holds. 

In case that $op_{i+1}$ is a pop operation and $top_i \geq 0$, then the value of $top_i$ 
is reduced by one (line \lref{ln:99}). Thus, we have that $top_{i+1}$ = $top_i - 1$ and 
Claim \ref{c2} holds. In case that $top_i = -1$, then it holds that $top_{i-1} = top_{i} = -1$, 
as in this case, the value of $top\_key$ is not modified (line \lref{ln:97}-\lref{ln:98}) 
and thus Claim~\ref{c2} holds.

We now prove Claim~\ref{c3}. In case that $op_{i+1}$ is a push operation, the 
pseudocode(line~\lref{ln:96}) and Claim~\ref{c1} imply that $s_s$ sends a value 
equal to $top_{i+1} = top_i + 1$ to $c$. The pseudocode (line~\lref{ln:103})
implies that the client $c$ inserts a pair of $\langle top_{i+1}, data \rangle$ 
in the directory. Thus, Claim~\ref{c3} holds. In case that $op_{i+1}$ is a pop 
operation, the induction hypothesis implies that Claim~\ref{c3} holds.

We now prove, Claim~\ref{c4}. By the semantics of \DIRDELETE, if at the point 
that the instance of \DIRINSERT\ is executed in the \Do\ - \While\ loop of lines 
\lref{ln:99}-\lref{ln:101} for $op_{i+1}$, the instance of \DIRINSERT\ of $op_{i+1}$ 
has not yet returned, then \DIRDELETE\ returns $\langle \bot, -\rangle$. Notice 
that the parameter of \DIRDELETE\ is $top_i$. By induction hypothesis (Claim~\ref{c3}), 
$top_i$ is the $key$ of the last pair $sl_{d_i}$ in $S_i$. Therefore, when \DIRDELETE\ 
returns a $status \neq \bot$, it holds that it returns the $data$ field of the last 
element in $S_i$ and that it is sent to $c$ and is used as the return value of $op_{i+1}$.
Thus Claim \ref{c4} holds. 
\end{proof}

From the above lemmas we have the following.

\begin{theorem}
The directory-based distributed stack implementation is linearizable. 
\end{theorem}

\ignore{
\begin{lemma}
For each integer $i > 0$, the following hold:
\begin{enumerate}
\item\label{c1} If $op_i$ is a push operation, it holds that $top_i$ = $top_{i-1} + 1$. 
\item\label{c2} If $op_i$ is a pop operation,  
$top_i$ = $top_{i-1} - 1$ if $top_{i-1} > 0$, or $top_i = 0$ otherwise.
\item \label{c3} $sl_i.key = top_i$.
\item \label{c4} If $op_i$ is a pop operation, then it returns the value of the field $data$ of $sl_{i}$.
\end{enumerate}
\end{lemma}

\begin{proof}[Proof Sketch.]
We prove the claims by induction on $i$. 

{\bf Base case.} 
For $i = 0$ the claim holds trivially.

{\bf Hypothesis.} Fix any $i-1$, $i > 0$ and assume that the lemma holds 
for any $j$, $1 \geq j \geq i-1$.

{\bf Induction step.} We prove that the claims also holds for $i$. 

In case that $op_{i}$ is a push operation. Variable $top\_key$ is incremented before 
it is sent to $c$ that requested the push operation (lines \lref{ln:92}-\lref{ln:96}).
Thus, the value that $c$ receives (line \lref{ln:103.1}) and stores in local variable $key$, 
is $top_i$. It holds that $top_{i}$ = $top_i$, thus Claim \ref{c1} holds. 

In case that $op_{i+1}$ is a pop operation and $top_i > 0$, then the value of $top_i$ 
is is reduced by one. Thus, we have that $top_{i}$ = $top_i - 1$ and Claim \ref{c1} holds. 
In case that $top_i = 0$, then it holds that $top_{i-1} = top_{i} = 0$ and thus Claim~\ref{c2} 
holds.

We now prove Claim~\ref{c3}. In case that $op_{i}$ is a push operation, the pseudocode
(line~\lref{ln:96}) and Claim~\ref{c1} implies that the synchronizer server sends 
a value equal to $top_i = top_{i-1} + 1$ to $c$. The pseudocode (line~\lref{ln:103})
implies that the client $c$ inserts a pair of $<top_i , data>$ in the directory. 
Thus, Claim~\ref{c3} holds.
In case that $op_i$ is a pop operation, induction hypothesis implies that Claim~\ref{c3}
holds.

We now prove, Claim~\ref{c4}. By the semantics of \DIRDELETE, if at the point that the 
instance of \DIRINSERT\ is executed in the \Do\ - \While\ loop of lines \lref{ln:99}-\lref{ln:101} 
for $op_{i}$, the instance of \DIRINSERT\ of $op_{i}$ has not yet returned, then \DIRDELETE\ returns 
$\langle \bot, -\rangle$. Notice that the parameter of \DIRDELETE\ is $top_i$. 
By induction hypothesis (Claim~\ref{alg11}(3c)), $top_i$ is the $key$ of the last pair 
$sl_i$ in $S_i$. Therefore, when \DIRDELETE\ returns a $status \neq \bot$, 
it holds that it returns the $data$ field of the last element 
in $S_i$ and this is sent to $c$ and is used as the return value of $op_{i+1}$.
Thus Claim \ref{c3} holds. 
\end{proof}
}
}
\subsection{Directory-Based Queue}
\label{app:queue-dir}

	\begin{minipage}{.45\textwidth}
		\begin{algorithm}[H]
			\small
			\caption{Events triggered in the synchronizer of the directory-based queue.}
			\label{alg17}
			\begin{code}
				\lreset
				\firstline
				\integer\ $head\_key=0, tail\_key=0$;                          \ul 
				\nl
				a message $\langle op, cid \rangle$ is received:               \nl
				\n   \If\ ($op$ == \ENQ) $\lbrace$                 \label{ln:116} \nl
				\n     $tail\_key$++;                              \label{ln:117} \nl
						send($cid, tail\_key$);                     \label{ln:118} \nl
				\p   $\rbrace$ \Elseif\ ($op$ == \DEQ) $\lbrace$                  \nl 
				\n		\If\ ($head\_key < tail\_key$) $\lbrace$   \label{ln:119} \nl
				\n      $head\_key$++;                            \label{ln:121} \nl
						send($cid, head\_key$);                   \label{ln:122} \nl
				\p     $\rbrace$    \Else\          \label{ln:120.1}\nl

				\n       send($cid, \NACK$);                       \label{ln:120} \ul
				\p     $\rbrace$                \label{ln:120.2}  \ul

				\p   $\rbrace$
				\p
			\end{code}
		\end{algorithm}
	\end{minipage}

The directory-based distributed queue is implemented in a way similar to the 
directory-based stack implementation. 
To implement a queue, the synchronizer $s_s$ maintains two counters, $head\_key$ and $tail\_key$, 
which store the key associated with the first and the last, respectively, element 
in the queue. 

In order to perform an enqueue or dequeue operation, a client calls 
\texttt{ClientEnqueue()} or \texttt{ClientDequeue()}, respectively 
(Algorithms~\ref{alg18}~and~\ref{alg19}, respectively).
To apply an operation $op$, a client sends a request to $s_s$, in order to 
receive the key of the element to be inserted or deleted. The client then calls 
\DIRINSERT\ to insert the new element in the directory 
(line~\lref{ln:118}).  
Each time $s_s$ receives an enqueue request, 
it increments 
$tail\_key$ and then sends the value stored in $tail\_key$ to the client.

Each time $s_s$ receives a dequeue request, if $head\_key \neq tail\_key$, then 
$s_s$ sends the value of $head\_key$ to the client and increments $head\_key$ 
to store the key of the next element in the queue. The client then uses this 
value as the key of the element to remove from the directory (line~\lref{ln:127}). 
Otherwise, in case $head\_key = tail\_key$, $s_s$ sends 
\NACK\ to $c$ without changing 
$head\_key$. 

A dequeue operation for which $s_s$ sends a key value to the client by executing line~\lref{ln:122}, is referred 
to as a successful dequeue operation. On the other hand, a dequeue operation for which 
$s_s$ sends \NACK\ to the client by executing line~\lref{ln:120}, and for which 
in turn, the client returns $\bot$, is referred to as an unsuccessul dequeue operation.



\begin{figure}
	\begin{minipage}{.45\textwidth}
		\begin{algorithm}[H]
\small
\caption{Enqueue operation for a client of the directory-based queue.}
\label{alg18}
\begin{code}
 \firstline
    \void\ \CEnq(\integer\ $cid$, Data $data$) $\lbrace$ \nl
 \n   $sid$ = get the server id;                                 \nl
      send($sid, \langle\ENQ, cid \rangle$);                     \nl
      $tail\_key$ = receive($sid$);               \label{ln:123} \nl
      \DIRINSERT($tail\_key, data$);              \label{ln:124} \ul
 \p $\rbrace$
\end{code}
\end{algorithm}
	\end{minipage}
\hfill
	\begin{minipage}{.45\textwidth}
		\begin{algorithm}[H]
\small
\caption{Dequeue operation for a client of the directory-based queue.}
\label{alg19}
\begin{code}
  \firstline
    Data \CDeq(\integer\ $cid$) $\lbrace$                     \nl
 \n   $sid$ = get the server id;                                      \nl
      send($sid, \langle \DEQ, cid\rangle$);                          \nl
      $head\_key$ = receive($sid$);                  \label{ln:125}   \nl
      \If ($head\_key == \NACK$)                     \label{ln:126}   \nl
 \n     \return\ $\bot$;                         \label{ln:126.1} \nl
 \p   
  $data$ = \BDIRDELETE($head\_key$);          \label{ln:127}   \nl
      \return\ $data$;                              \label{ln:127.1} \ul
 \p $\rbrace$ 
\end{code}
\end{algorithm}
	\end{minipage}
\end{figure}

\subsubsection{Proof of Correctness}
Let $\alpha$ be an execution of the directory-based queue implementation.
We assign linearization points to enqueue and dequeue operations in $\alpha$ as follows: 
The linearization point of an enqueue operation $op$ is placed in the configuration resulting 
from the execution of line~\lref{ln:118} for $op$ by $s_s$. 
The linearization point of a dequeue operation $op$ is placed in the configuration resulting 
from the execution of either line~\lref{ln:122} or line~\lref{ln:120} for $op$ (whichever is 
executed) by $s_s$. 

\begin{lemma}
The linearization point of an enqueue (dequeue) operation $op$ is placed within its execution interval.
\end{lemma}

\begin{proof}
Assume that $op$ is an enqueue operation and let $c$ be the client that invokes it.
The linearization point of $op$ is placed at the configuration resulting 
from the execution of line \lref{ln:118} for $op$ by $s_s$. This line is executed after 
the request by $c$ is received, i.e. after $c$ invokes {\tt ClientEnqueue}. Furthermore, it 
is executed before $c$ receives the response by the server and thus, before {\tt ClientEnqueue} 
returns. Therefore, the linearization point is included in the execution interval of enqueue.

The argument regarding dequeue operations is similar.
\end{proof}

Denote by $L$ the sequence of operations which have been assigned linearization 
points in $\alpha$ in the order determined by their linearization points. Let $C_i$ 
be the configuration in which the $i$-th operation $op_i$ of $L$ is linearized; 
denote by $C_0$ the initial configuration.  
Denote by $\alpha_i$, the prefix of $\alpha$ which ends with $C_i$ and let $L_i$ 
be the prefix of $L$ up until the operation that is linearized at $C_i$. 
We denote by $head_i$ the value of the local variable $head\_key$ of $s_s$
at configuration $C_i$, and by $tail_i$ the value of the local variable $tail\_key$ 
of $s_s$ at $C_i$. By the pseudocode, we have that the initial values of $tail\_key$ 
and $head\_key$ are $0$; therefore, we consider that $head_0 = tail_0 = 0$. 

Notice that since only $s_s$ executes Algorithm~\ref{alg17}, we have the following.

\begin{observation}
\label{obs:seq}
Instances of Algorithm \ref{alg17} are executed 
sequentially by $s_s$, i.e. their execution does not overlap.
\end{observation}

By inspection of Algorithm~\ref{alg17}, we have that for some instance of it, 
either lines~\lref{ln:116}~-~\lref{ln:117}, or lines~\lref{ln:119}~-~\lref{ln:122}, 
or lines~\lref{ln:120.1}~-~\lref{ln:120} are executed, where either $tail\_key$ or $head\_key$ is incremented. Then, by the 
way linearization points are assigned, and by Observation~\ref{obs:seq}, we 
have the following.

\begin{observation}
\label{obs:mod}
Given two configurations $C_i$, $C_{i+1}$, $i \geq 0$, in $\alpha$, there is at most one 
step in the execution interval between $C_i$ and $C_{i+1}$ that modifies either $head\_key$ 
or $tail\_key$. 
\end{observation}

More specifically, regarding the values of $head_i$ and $tail_i$, we obtain 
the following lemma.

\begin{lemma}
\label{lemma:ht}
For each integer $i \geq 1$, the following hold at $C_i$:
\begin{compactenum}
\item\label{cq1} If $op_i$ is an enqueue operation, 
then $tail_i = tail_{i-1} + 1$ and $head_i = head_{i-1}$.
\item\label{cq2} If $op_i$ is a dequeue operation and $head_{i-1} \neq  tail _{i-1}$, 
then $head_i = head_{i-1} + 1$ and $tail_i = tail _{i-1}$; 
otherwise $head_i = head_{i-1}$ and  $tail_i = tail_{i-1}$. 
\item\label{cq2.2} $head_i \leq  tail_i$.
\end{compactenum}
\end{lemma}

\begin{proof}
We prove the claims by induction. 

{\bf Base case.} We prove the claims for $i = 1$. 
Assume first that $op_1$ is an enqueue operation. 
Then, the linearization point of 
$op_1$ is placed in the configuration resulting from the execution of 
line~\lref{ln:118}. By inspection of the pseudocode, we have that $tail\_key$ 
is incremented by $s_s$ between $C_0$ and $C_1$, before the linearization 
point of $op_1$.
Notice also that because of Observation~\ref{obs:seq} no process other than 
$s_s$ modifies neither $tail\_key$ nor $head\_key$ between $C_0$ and $C_1$. 
Thus, $tail_1 = tail_0 + 1$. 
The value of $head\_key$ is not modified by enqueue operations (lines~\lref{ln:116}~-~\lref{ln:117}), 
therefore $head_1 = head_0$. Thus, claim~\ref{cq1} holds.

Next, assume that $op_1$ is a dequeue operation. Then, the linearization point 
of $op_1$ may be placed at the configuration resulting from the execution of 
line~\lref{ln:120} or line~\lref{ln:122}, whichever is executed by $s_s$ for it. 
By inspection of the pseudocode (line~\lref{ln:119}), line~\lref{ln:120} 
is executed only in case $head_0 < tail_0$. Since $head_0 = tail_0 = 0$, 
line~\lref{ln:120} is not executed and $op_1$ is linearized at the configuration resulting 
from the execution of 
line~\lref{ln:122}. By the pseudocode, line~\lref{ln:121} and by Observation~\ref{obs:seq}, 
$head\_key$ is incremented by $1$ in the execution step preceding the linearization point of $op_1$, i.e. between 
configurations $C_0$ and $C_1$. Thus, $head_1 = head_0 + 1$. The value of $tail\_key$ 
is not modified by dequeue operations (lines~\lref{ln:119}~-~\lref{ln:121}), 
therefore $tail_1 = tail_0$. By the above, claim~\ref{cq2} also holds.

From the previous reasoning, we have that in case $op_1$ is an enqueue operation, then 
$tail_1 = 1$ and $head_1 = 0$, while if $op_1$ is a dequeue operation, 
$tail_1 = 0$ and $head_1 = 0$. In either case, $head_1 \leq tail_1$, thus claim~\ref{cq2.2} also holds.

{\bf Hypothesis.} Fix any $i$, $i \geq 1$, and assume that the claims hold for $i-1$. 

{\bf Induction Step.} We prove that the claims also hold for $i$.
First, assume that $op_i$ is an enqueue operation. Then, the linearization point of 
$op_i$ is placed in the configuration resulting from the execution of 
line~\lref{ln:118}. By inspection of the pseudocode, we have that $tail\_key$ 
is incremented by $s_s$ between $C_{i-1}$ and $C_i$, before the linearization 
point of $op_i$.
Notice also that because of Observation~\ref{obs:seq} no process other than 
$s_s$ modifies neither $tail\_key$ nor $head\_key$ between $C_{i-1}$ and $C_i$. 
Thus, $tail_i = tail_{i-1} + 1$. 
The value of $head\_key$ is not modified by enqueue operations (lines~\lref{ln:116}~-~\lref{ln:117}), 
therefore $head_i = head_{i-1}$. Thus, claim~\ref{cq1} holds.

Next, assume that $op_i$ is a dequeue operation. Then, the linearization point 
of $op_i$ may be placed at the configuration resulting from the execution of 
line~\lref{ln:120} or line~\lref{ln:122}, whichever is executed by $s_s$ for it. 
Let $op_i$ be linearized at the execution of 
line~\lref{ln:120}. 
By the induction hypothesis, $head_{i-1} \leq tail _{i-1}$.
By inspection of the pseudocode (line~\lref{ln:119}), 
line~\lref{ln:120} is executed only in the case that $head_{i-1} = tail_{i-1}$. 
By inspection of the 
pseudocode (lines~\lref{ln:119}~-~\lref{ln:120}) and by Observation~\ref{obs:seq}, 
it follows that in this case $head\_key$ is not modified in the execution interval 
between $C_{i-1}$ and $C_i$. 
Therefore, $head_i = head_{i-1}$. Since a dequeue operation does not modify 
$tail\_key$, it also holds that $tail_i = tail_{i-1}$. 
Finally, let $op_i$ be  linearized at the execution of 
line~\lref{ln:122}. In this case, by the induction hypothesis and reasoning as previously, 
 $head_{i-1} \neq tail_{i-1}$. By the pseudocode, line~\lref{ln:121} and by Observation~\ref{obs:seq}, 
$head\_key$ is incremented by $1$ in the computation step preceding the linearization point of $op_i$, i.e. between 
configurations $C_{i-1}$ and $C_i$. Thus, $head_i = head_{i-1} + 1$. The value of $tail\_key$ 
is not modified by dequeue operations (lines~\lref{ln:119}~-~\lref{ln:121}), 
therefore $tail_i = tail_{i-1}$. By the above, claim~\ref{cq2} also holds.

By the induction hypothesis, we have that $head_{i-1} \leq tail_{i-1}$.
From the previous reasoning, we have that in case $op_i$ is an enqueue operation, 
then $tail_i = tail_{i-1} + 1$ and $head_i = head_{i-1}$.  It follows that 
$head_i \leq tail_i$. On the other hand, if $op_i$ is a dequeue operation, 
$tail_i = tail_{i-1}$ and $head_i = head_{i-1}$ in case $head_{i-1} = tail_{i-1}$. 
Otherwise, in case $head_{i-1} < tail_{i-1}$, then $tail_i = tail_{i-1}$ and $head_i = head_{i-1} + 1$. 
Since $head_{i-1} < tail_{i-1}$, and since both $head\_key$ and $tail\_key$ are both 
incremented in steps of $1$, it follows that $head_i \leq tail_i$.
In either of the previous cases, claim~\ref{cq2.2} also holds.
\end{proof}

Let $\sigma$ be the sequential execution that results by executing the 
operations in $L$ sequentially, in the order of their linearization points, 
starting from an initial configuration in which the queue is empty. 
Let $\sigma_i$ be the prefix of $\sigma$ that contains the operations in $L_i$.
Let $Q_i$ be the state  
of the queue after the operations of $L_i$ have been applied to an empty queue
sequentially. Let the size of $Q_i$ (i.e. the number of elements contained in $Q_i$) 
be $d_i$. Denote by $sl_i^j$ the $j$-th element of $Q_i$, $1 \leq j \leq d_i$. 
Consider  a sequence of elements $S$. If $e$ is the first element of $S$, 
we denote by $S \setminus e$ the suffix of $S$ that results by removing 
only element $e$ from the first position of $S$. 
We further denote by $S' = S \cdot e$ the sequence that 
results by appending some element $e$ to the end of $S$.

%

%


By inspection of the pseudocode, it follows that:

\begin{observation}
	\label{obs:enq-deq}
	If $op_i$ is an enqueue operation, it inserts a pair $\langle data, key \rangle$ in the directory, 
	where $data$ is the argument of \CEnq\ executed for $op_i$, and $key = tail_i$.
	If $op_i$ is a dequeue operation then, if $head_i \neq tail_i$, it removes a pair $\langle data, key \rangle$ with $key = head_i$ from the directory and returns $data$; if $head_i = tail_i$, it does not remove any pair from the directory and returns $\bot$. 
\end{observation}

\begin{lemma}
For each $i$, let $m_i$ be the number of enqueue operations in $L_i$ and $k_i$ be the number of 
successful dequeue operations in $L_i$. Then:
\begin{compactenum}
	 \item \label{cq3} $Q_i$ contains the elements inserted by the $m_i - k_i$ last enqueue operations in $L_i$, in order,
	 \item \label{cq4} $tail_i = m_i$,
	 \item \label{cq5} $head_i = k_i$,
	\item \label{cq6} if $op_i$ is a dequeue operation, then it returns the same response in $\alpha$ and $\sigma_i$. 
					If $op_i$ is a successful dequeue operation that removes the pair $\langle data, - \rangle$, from the 
					directory, then its response is the data inserted to the queue by the 
					$k_i$-th enqueue operation in $L_i$, whereas if $op_i$ is unsuccessful, then it returns $\bot$.
\end{compactenum}
\end{lemma}

\begin{proof}
We prove the claims by induction. 

{\bf Base case.} 
We prove the claim for $i$ = $1$.
First, assume that $op_1$ is an enqueue operation. Then, $m_1 = 1$ and $k_1 = 0$, and $Q_1$ contains a single element, 
namely the element inserted by the $(m_1 - k_1)$-th enqueue operation in $L_1$, which is $op_1$, thus claim~\ref{cq3} holds. 
By Lemma~\ref{lemma:ht}, we have that at $C_1$, 
$tail_1 = tail_0 + 1 = 1 = m_1$, and $head_1 = head_0 = 0 = k_1$. Thus, claims~\ref{cq4}~and~\ref{cq5} also hold. 
Claim~\ref{cq6} holds trivially. 

Next, assume that $op_1$ is a dequeue operation. Thus, $m_1 = 0$. 
By inspection of the pseudocode, $head_0 = tail_0 = 0$. 
By the way linearization points are assigned and by Observation~\ref{obs:mod}, 
$op_1$ is the operation for which the code of $s_s$ is executed for the first 
time. Thus, line~\lref{ln:120} is executed and $op_1$ is unsuccessful. Therefore, $k_1 = 0$. 
By Lemma~\ref{lemma:ht}, $head_1 = 0$ and $tail_1 = 0$. It follows that 
$head_1 = m_1$ and $tail_1 = k_1$. Thus, claims~\ref{cq4}~and~\ref{cq5} hold. 

Since $op_1$ is the first operation in $\sigma$ and since it is a dequeue operation, 
it is an unsuccessful dequeue operation in $\sigma$. Since $Q_0$ is empty, it follows 
that $Q_1 = \emptyset$, so claim~\ref{cq3} follows.  By inspection of the pseudocode, 
$op_1$ returns $\bot$ in $\alpha$. Since $op_1$ is an unsuccessful dequeue in $\sigma$, 
it also returns $\bot$ in $\sigma$. Thus, claim~\ref{cq6} also holds. 

{\bf Hypothesis.} Fix any $i$, $i > 0$ and assume that the claims hold 
for $i-1$.

{\bf Induction step.} We prove that the claims also hold 
for $i$.
First, assume that $op_i$ is an enqueue operation and let the $data$ argument of \CEnq\ for $op_i$ be $e$.
By the induction hypothesis, $Q_{i-1}$ contains the elements inserted by the  $m_{i-1} - k_{i-1}$ last enqueue 
operations in $L_{i-1}$, in order. Since $op_i$ is an enqueue operation, it follows that 
$Q_i = Q_{i-1} \cdot e$, $m_i = m_{i-1} + 1$, and 
claim~\ref{cq3} holds. Moreover, $k_i = k_{i-1}$. By Lemma~\ref{lemma:ht}, we have that 
$tail_i = tail_{i-1} + 1$ and that $head_i = head_{i-1}$.
By the induction hypothesis, $tail_{i-1} = m_{i-1}$ and $head_i = k_{i-1}$. 
It follows that $tail_i = m_{i-1} + 1 = m_i$ 
and that $head_i = head_{i-1} = k_{i-1} = k_i$. Thus, claims~\ref{cq4}~and~\ref{cq5} also hold. 
Claim~\ref{cq6} holds trivially. 

Now let $op_i$ be a dequeue operation. First, assume that $op_i$ is a successful 
dequeue operation, i.e. $s_s$ executes line~\lref{ln:122} for $op_i$. 
By inspection of the pseudocode, by Observation~\ref{obs:mod} and the way 
linearization points are assigned, it follows that when line~\lref{ln:119} 
is executed, $head\_key$ and $tail\_key$ have the values $head_{i-1}$ and 
$tail_{i-1}$ respectively. Since line~\lref{ln:122} is executed, it follows 
that $head_{i-1} < tail_{i-1}$. By the induction hypothesis, $head_{i-1} = k_{i-1}$ 
and $tail_{i-1} = m_{i-1}$. Since $op_i$ is a dequeue operation, $m_i = m_{i-1}$ 
and $k_i = k_{i-1} + 1$.  By Lemma~\ref{lemma:ht}, we have that 
$tail_i = tail_{i-1}$ and that $head_i = head_{i-1} + 1$. 
So, $tail_i = m_i$ and $head_i = k_i$. Thus, claims~\ref{cq4}~and~\ref{cq5} hold. 

By the induction hypothesis, $Q_{i-1}$ contains the elements inserted by the
$m_{i-1} - k_{i-1}$ last enqueue operations in $L_{i-1}$. 
Since $m_{i-1} - k_{i-1} = m_{i-1} - (k_i - 1)$,  $Q_{i-1}$ contains 
$m_{i-1} - k_i + 1$ elements. Since $op_i$ is a dequeue, it removes the first element of $Q_{i-1}$, 
so $Q_i$ contains $m_i - k_i$ elements, and claim~\ref{cq3} follows. 

Recall that $head_{i-1} < tail_{i-1}$. By the induction hypothesis 
(claims~\ref{cq4}~and~\ref{cq5}), it follows that $k_{i-1} < m_{i-1}$. 
Thus, $Q_{i-1}$ is not empty and $op_i$ is a successful dequeue operation 
also in $\sigma$. 

By the induction hypothesis, $Q_{i-1}$ contains the elements inserted by 
the $m_{i-1} - k_{i-1} = m_i - (k_i - 1) = m_i - k_i + 1$ last enqueue 
operations in $L_{i-1}$, in ascending order from head to tail. Since $m_i$ is 
the total number of successful enqueue operations in $L_i$ and $k_i$ the total 
number of successful dequeue operations, this means that in $\sigma$, $op_i$ 
removes the element inserted to the queue by the $k_i$-th enqueue operation. 
By inspection of the pseudocode, the client removes the key with value $head_i$ 
from the directory and returns the data that are associated with this key. 
Since $head_i = k_i$, $op_i$ returns the data from a $\langle data, key \rangle$
pair it removes from the directory. By inspection of the pseudocode, the 
key used to remove the data has the value $head_i$. Therefore, also in 
$\alpha$, $op_i$ returns the data associated with the $k_i$-th enqueue operation in $L_i$, 
and claim~\ref{cq6} holds.

Assume now that $op_i$ is an unsuccessful dequeue operation, i.e. $s_s$ executes line~\lref{ln:120} for $op_i$. 
By inspection of the pseudocode, by Observation~\ref{obs:mod} and the way 
linearization points are assigned, it follows that when line~\lref{ln:119} 
is executed, $head\_key$ and $tail\_key$ have the values $head_{i-1}$ and 
$tail_{i-1}$ respectively. Since line~\lref{ln:120} is executed, it follows 
that $head_{i-1} = tail_{i-1}$. By the induction hypothesis, $head_{i-1} = k_{i-1}$ 
and $tail_{i-1} = m_{i-1}$. Since $op_i$ is an unsuccessful dequeue operation, $m_i = m_{i-1}$ 
and $k_i = k_{i-1}$.  By Lemma~\ref{lemma:ht}, we have that in this case, 
$tail_i = tail_{i-1}$ and that $head_i = head_{i-1}$. 
So, $tail_i = m_i$ and $head_i = k_i$. Thus, claims~\ref{cq4}~and~\ref{cq5} hold. 
Furthermore, since $head_{i-1} = tail_{i-1}$, it follows that $m_{i-1} = k_{i-1}$, 
which means that $Q_{i-1}$ is empty, and thus, $op_i$ is unsuccessful in $\sigma$. 
As $m_i = m_{i-1}$ and $k_i = k_{i-1}$ and $op_i$ is unsuccessful, claim~\ref{cq3} holds.
Since $op_i$ is unsuccessful in $\alpha$, it removes no $\langle data, key \rangle$
pair from the directory and returns $\bot$. Thus, claim~\ref{cq6} also holds. 
\ignore{
Now, assume that $op_i$ is a dequeue operation. By the induction hypothesis, 
$Q_{i-1}$ contains the elements inserted by the  $m_{i-1} - k_{i-1}$ last enqueue operations 
in $L_{i-1}$. First, assume that $op_i$ is a successful dequeue operation, 
i.e. $s_s$ executed line~\lref{ln:122} for $op_i$.
%
By inspection of the pseudocode, $head_{i-1} = tail_{i-1}$.
By the induction hypothesis, $tail_{i-1} = m_{i-1}$ and $head_{i-1} < k_{i-1}$. 
Thus, $k_{i-1} < m_{i-1}$. So, $op_i$ is a successful dequeue also in $\alpha_i$.
In $\sigma_i$, $e$ is the first element of $Q_{i-1}$ and $Q_i = Q_{i-1} \setminus e$. 
So, the queue now contains the elements inserted by the $m_{i-1} - k_{i-1} - 1$ 
last enqueue operations in $L_{i-1}$, in order. Since $op_i$ is a successful dequeue operation, 
it follows that $m_i = m_{i-1}$, $k_i = k_{i-1} + 1$ and the number of elements in $Q_i$ are 
 $m_{i-1} - k_{i-1} - 1 = m_i - k_{i-1} - 1 = m_i - (k_{i-1} + 1)$, in order. 
 Thus, 
 claim~\ref{cq3} holds.
 By the induction hypothesis, Observation~\ref{obs:mod}, and Lemma~\ref{lemma:ht}, we have that 
 $tail_i = tail_{i-1} = m_{i-1} = m_i$ and that $head_i = head_{i-1} + 1 = k_{i-1} + 1 = k_i$. 
 Thus, claims~\ref{cq4}~and~\ref{cq5} also hold. 
 By the induction hypothesis, the elements of $Q_{i-1}$ are, from head to tail, 
 those inserted to the queue by enqueue operations 
 $k_{i-1}+1, k_{i-1}+2, \ldots, m_{i-1}-1, m_{i-1}$. Since, in $\sigma_i$, $op_i$ removes the first element of 
 $Q_{i-1}$, the element at the head  of $Q_i$ is that inserted by the $k_{i-1}+1$-th enqueue operation in $L_i$. 
 However, we have that  $k_{i-1}+1 = k_i$. In $\alpha_i$, $op_i$ removes from the directory 
 an element with key $head_i$. By Lemma~\ref{lemma:ht}, we have that $head_i = head_{i-1} + 1$. 
 By the induction hypothesis, we have that $head_{i-1} = k_{i-1}$. Therefore, 
 in $\alpha_i$, $op_i$ removes an element with key $head_i = k_{i-1} + 1 = k_i$ 
 and claim~\ref{cq6} holds.

 Next, assume that $op_i$ is an unsuccessful dequeue operation in $\alpha$. 
 Since $op_i$ is unsuccessful, by inspection of the pseudocode, the condition of 
 line~\ref{ln:119} evaluates to $\false$. By the induction hypothesis, $head_{i-1} = k_{i-1}$ 
 and $tail_{i-1} = m_{i-1}$, which means that in $\sigma_i$ there have been executed as many 
 successful dequeues as enqueues. Thus, $Q_{i-1}$ is empty and $op_i$ is an unsuccessful dequeue 
 also in $\sigma_i$ and claim~\ref{cq3} holds. 
 Furthermore, $m_i = m_{i-1}$ and $k_i = k_{i-1}$, so claims~\ref{cq4}~and~\ref{cq5} hold. 
 By Lemma~\ref{lemma:ht}, $tail_i = tail_{i-1}$ and $head_i = head_{i-1}$. 
 Thus, $m_i = m_{i-1}$ and $k_i = k_{i-1}$ and claims~\ref{cq3}~to~\ref{cq5} hold. 
 Thus, at $C_i$ there have been an equal number of successful dequeue and enqueue operations,
 which in turn implies that the queue is empty and claim~\ref{cq6} holds. 
}
\end{proof}

\ignore{
\begin{lemma}
\label{lemma:seq}
At $C_i$, $i \geq 1$, the following hold:
\begin{compactenum}
\item \label{cq3} If $op_i$ is an enqueue operation, then $tail_i = sl_i^{d_i}.key$.  
\item \label{cq4} If $op_i$ is a dequeue operation, then if $Q_{i-1} \neq \epsilon$, 
$head_i = sl_{i-1}^1.key$. If $Q_{i-1} = \epsilon$, then $head_i = tail_i$.
\end{compactenum}
\end{lemma}

\begin{proof}
We prove the claims by induction. 

{\bf Base case.} 
We prove the claim for $i$ = $1$.
Consider the case where $op_1$ is an enqueue operation. 
Then, $d_1 = 1$ and $Q_1$ contains only the pair $\langle 0, data \rangle$.
By Observation \ref{obs:seq}, 
it is the first operation in $\alpha$ for which an instance of Algorithm~\ref{alg17} 
is executed by $s_s$. Therefore, by Lemma~\ref{lemma:ht}, $tail_1 = tail_0 = 0$. 
Thus, $tail_1 = sl_1^{d_1}.key$

Now consider the case where $op_1$ is a dequeue operation. By Observation~\ref{obs:seq}, 
$op_1$ is the first operation in $\alpha$ for which an instance of Algorithm~\ref{alg17} 
is executed by $s_s$. Notice that then, $Q_1 = \epsilon$. Therefore, by Lemma~\ref{lemma:ht}, 
$head_1 = head_0 = 0$. By the same reasoning, $tail_1 = tail_0 = 0$. Thus, $head_1 = tail_1$, 
so claim~\ref{cq4} holds. 

{\bf Hypothesis.} Fix any $i$, $i > 0$ and assume that the lemma holds at $C_i$. 

{\bf Induction step.} We prove that the claims also hold at $C_{i+1}$. 
Assume that $op_{i+1}$ is an enqueue operation. We examine two cases. First, 
consider that $op_i$ is an enqueue operation as well. By the induction hypothesis, 
$sl_i^{d_i}.key = tail_i$. By Lemma~\ref{lemma:ht}, we have that $tail_{i+1} = tail_i + 1$. 
By Observation~\ref{obs:enq-deq}, we have that the client $c$ that initiated $op_{i+1}$ 
inserts a pair with $key = tail_{i+1} = tail_i + 1$ into the directory. By definition, 
$sl_{i+1}^{d_{i+1}}.key =  sl_i^{d_i}.key + 1$. Thus, $sl_{i+1}^{d_{i+1}}.key = tail_i + 1$, 
and claim~\ref{cq3} holds. 

Next, consider that $op_i$ is a dequeue operation. By Lemma~\ref{lemma:ht}, 
dequeue operations do not modify $tail\_key$. This lemma further implies that 
$tail_{i+1} = tail_j + 1$, where $op_j$ is the last enqueue operation preceding 
$op_{i+1}$ in $L_{i+1}$. By definition and by Observation~\ref{obs:enq-deq}, 
$op_j$ enqueues a pair with $key = tail_j$ to $Q_j$. Furthermore, by definition 
of $op_j$, all other operations in $L_{i+1}$ that have a linearization point 
between that of $op_j$ and $op_{i+1}$, are dequeue operations. Therefore, no 
further element is appended to $Q$ between $C_j$ and $C_{i+1}$, i.e. $sl_j^{d_j} = sl_i^{d_i}$. 
Notice that $sl_j^{d_j}. key = tail_j$. By Observation~\ref{obs:enq-deq}, $c$ inserts a 
pair with $key = tail_{i+1}$ into the directory and by definition, 
$sl_{i+1}^{d_{i+1}}.key =  sl_i^{d_i}.key + 1$. Thus, since $tail_{i+1} = tail_j + 1$, 
it follows that $sl_{i+1}^{d_{i+1}}.key = tail_j + 1 = sl_j^{d_j}.key + 1 = sl_i^{d_i}.key + 1$, 
and claim~\ref{cq3} holds. 

Now let $op_{i+1}$ be a dequeue operation. Again we examine two cases. 
First, consider that $op_i$ is a dequeue operation as well.  By the induction 
hypothesis, $op_i$ dequeues an element with key $sl_{1_{i-1}}.key = head_i$. 
Since claim~\ref{cq4} holds at $C_i$, $sl_{1_i}.key = head_i+1$. By Lemma~\ref{lemma:ht}, 
$head_{i+1} = head_i + 1$. Thus, $op_{i+1}$ removes from 
the sequential queue the element with key equal to $head_i + 1$. Since 
claims~\ref{cq3}~and~\ref{cq4} hold at $C_i$ by the induction hypothesis, 
we have that this element is $sl_{1_i}$, i.e. claim \ref{cq4} also holds 
at $C_{i+1}$. 

Next consider that $op_i$ is an enqueue operation. By Lemma~\ref{lemma:ht}, 
enqueue operations do not modify $head\_key$. This lemma further implies that 
$head_{i+1} = head_j + 1$, where $op_j$ is the last dequeue operation preceding 
$op_{i+1}$ in $L_{i+1}$. By definition and by Observation~\ref{obs:enq-deq}, 
$op_j$ dequeues a pair with $key = tail_j$ from $Q_{j-1}$. Furthermore, by definition 
of $op_j$, all other operations in $L_{i+1}$ that have a linearization point 
between that of $op_j$ and $op_{i+1}$, are enqueue operations. Therefore, no 
further element is removed from $Q$ between $C_j$ and $C_{i+1}$, i.e. $sl_j^1 = sl_i^1$. 
Notice that $sl_j^1. key = head_j$. By Observation~\ref{obs:enq-deq}, $c$ removes a 
pair with $key = head_{i+1}$ from the directory and by definition, 
$sl_{i+1}^1.key =  sl_i^1.key + 1$. Thus, since $head_{i+1} = head_j + 1$, 
it follows that $sl_{i+1}^1.key = head_j + 1 = sl_j^1.key + 1 = sl_i^1.key + 1$, 
and claim~\ref{cq4} holds. 
\end{proof}


By Lemma~\ref{lemma:ht} and by inspection of the pseudocode, we have that 
at $C_i$, $i > 0$, the value of $tail\_key$ indicates 
the number of enqueue operations on $Q_i$ that have been linearized in 
$\alpha_i$, and the  value of $head\_key$ indicates the number of successful 
dequeue operations (i.e. dequeue operations that do not return $\bot$) on $Q_i$ 
that have been linearized in $\alpha_i$. Thus, the following corollary holds.

\begin{corollary}
\label{cor:empty}
$Q_i = \epsilon$ if and only if $head_i = tail_i$. 
\end{corollary}

\begin{lemma}
If $op_i$ is a dequeue operation, then it returns the value of the field $data$ of 
$sl_{i-1}^1$ or $\bot$ if $Q_{i-1} = \epsilon$.
\end{lemma}

\begin{proof}[Proof Sketch]

Consider the case where $Q_{i-1} \neq \epsilon$. By definition of $Q_i$, we have that 
$Q_i = Q_{i-1} \setminus \{sl_{i-1}^1\}$. Let $op_j$ be the enqueue operation that is 
linearized before $op_i$ and inserts an element with key $head_i$ to the queue. Notice 
by the pseudocode, line \lref{ln:127}, that the parameter of \DIRDELETE\ is $head_i$. 
By the semantics of \DIRDELETE, if at the point that the instance of \DIRDELETE\ is 
executed in the \Do\ - \While\ loop of  lines~\lref{ln:127}~-~\lref{ln:127.1} for $op_i$, 
the instance of \DIRINSERT\ of $op_j$ has not yet returned, then \DIRDELETE\ returns 
$\langle \bot, -\rangle$. 

By Lemma~\ref{lemma:seq}, and since $head\_key$ is not modified by the execution of 
line~\lref{ln:122} by the server, $head_i$ is the $key$ of the first pair $sl_{i-1}^1$ 
in $Q_{i-1}$. Therefore, when \DIRDELETE\ returns a $status \neq \bot$, it holds 
that it returns the $data$ field of $sl_{i-1}^1$, the first element in $Q_{i-1}$, 
as the return value of $op_i$, i.e. the claim holds. 

Now consider the case where $Q_{i-1} = \epsilon$.
Since, by Corollary~\ref{cor:empty}, when this is the case, $head_i = tail_i$, 
\NACK\ is sent to the client that invoked $op_i$ and, by inspection of the 
pseudocode, $op_i$ returns $\bot$, i.e. the claim holds. 
\end{proof}

}

From the above lemmas we have the following:

\begin{theorem}
The directory-based queue implementation is linearizable.  
\end{theorem}

\subsection{Queues with Special Functionality}
{\bf Synchronous queue.}
A synchronous queue $Q_S$ is an implementation of the queue data type. 
Instead of storing elements, a synchronous queue matches instances of $Dequeue()$ with instances of 
$Enqueue()$ 
operations. Thus, if $op_e$ is an instance of an $Enqueue(x, Q_S)$ operation 
and $op_d$ an instance of a $Dequeue(Q_S)$ such that 
$op_d$ returns the element $x$ enqueued by $op_e$, then 
a synchronous queue ensures that 
the execution intervals of $op_e$ and $op_d$ are overlapping.

In order to derive a distributed synchronous queue from the directory-based queue proposed here, 
$s_s$ must respond to a dequeue request with the value of $head$ and increment $head$, even if 
$head = tail$. Moreover, $s_s$ must use a local queue to store active enqueue requests together 
with the keys it has assigned to them (notice that there can be no more such requests that the 
number of clients); $s_s$ must send the key $k$ for each such enqueue request to the client that 
initiated it, at the time that $head$ becomes equal to $k$. In this way, the execution interval 
of an enqueue operation for element $e$ overlaps that of the dequeue operation which gets $e$ as 
a response, as specified by the semantics of a synchronous queue.

\vspace*{.1cm}
\noindent
{\bf Delay queue.}
A delay queue $Q_D$ implements the queue abstract data type. Each element $e$ of a delay queue 
is associated with a delay value $t_e$ that represents the time that $e$ must remain in the queue before it can be removed from it. 
Thus, an $Enqueue(e, t_e, Q_D)$ inserts an element $e$ with time-out value $t_e$ to $Q_S$. $Dequeue(Q_D)$ returns the element $e$ 
residing at the head of $Q_D$ if $t_e$ has expired and blocks (or performs spinning) if this is not the case. Notice that this implementation 
can easily be provided by associating each element inserted in the directory with a time-out value. We also have to change the way
that the directory works so that it takes into consideration the delay of each element before removing it.
\subsection{Directory-Based Double-Ended Queue (Deque)}
\label{app:deque-dir}

The implementation of the directory-based deque follows similar principles 
as the stack and queue implementations. In order to implement a deque, $s_s$ 
also maintains two counters, $head$ and $tail$, which 
store the key associated with the first and the last, respectively, element in the deque.
However, in this case,
counters $head$ and $tail$ may store negative integers and are incremented or decremented based on the operation to be performed.

\subsubsection{Algorithm Description}
Event-driven pseudocode for the synchronizer $s_s$ is presented in 
Algorithm \ref{alg-dirdeq}; $s_s$ now performs a combination of 
actions presented for the synchronizers of the stack and the queue 
implementations (Algorithms ~\ref{alg11}
and ~\ref{alg17}). 

The synchronizer $s_s$ has two counters, $head\_key$ and $tail\_key$ 
(line \lref{ln:alg-dirdeq-1}), that store the key associated with the first and the 
last, respectively, element in the deque. The $head\_key$ is modified when 
operations targeting the front are received by $s_s$ and the $tail\_key$ is modified 
when operations targeting the back are received by $s_s$. Because each endpoint of a deque 
behaves as a stack, the actions for enqueuing and dequeuing are similar as in Algorithm~\ref{alg11}. 

	\begin{minipage}{.45\textwidth}
\begin{algorithm}[H]
\small
\caption{Events triggered in the synchronizer of the directory-based deque.}
\label{alg-dirdeq}
\begin{code}
  \lreset
  \firstline
     \integer\ $head\_key=0$, $tail\_key=0$;             \label{ln:alg-dirdeq-1} \ul 
                                                                                 \nl
     a message $\langle op,\ cid\rangle$ is received:    \label{ln:alg-dirdeq-2} \nl
  \n   \Switch\ ($op$) $\lbrace$                         \label{ln:alg-dirdeq-3} \nl
  \n     \Case\ \ENQT:                                   \label{ln:alg-dirdeq-4} \nl
  \n       $tail\_key++$;                                \label{ln:alg-dirdeq-5} \nl
           send($cid$, $tail\_key$);                     \label{ln:alg-dirdeq-6} \nl
           \Break;                                                               \nl
  \p     \Case\ \DEQT:                                   \label{ln:alg-dirdeq-7} \nl
  \n       \If\ ($tail\_key == head\_key$) $\lbrace$     \label{ln:alg-dirdeq-8} \nl
  \n         send($cid$, \NACK);                         \label{ln:alg-dirdeq-9} \nl
  \p       $\rbrace$ \Else\ $\lbrace$                                            \nl
  \n         \Do\ $\lbrace$                             \label{ln:alg-dirdeq-11} \nl
  \n           $status$ = \DIRDELETE($tail\_key$);      \label{ln:alg-dirdeq-12} \nl
  \p         $\rbrace$ \While\ ($status == \bot$);                               \nl
             $tail\_key--$;                             \label{ln:alg-dirdeq-10} \nl
             send($cid$, $status$);                     \label{ln:alg-dirdeq-13} \ul 
  \p       $\rbrace$                                                             \nl
           \Break;                                                               \nl
  \p     \Case\ \ENQH:                                  \label{ln:alg-dirdeq-14} \nl 
  \n       send($cid$, $head\_key$);                    \label{ln:alg-dirdeq-16} \nl
           $head\_key--$;                               \label{ln:alg-dirdeq-15} \nl
           \Break;                                                               \nl
  \p     \Case\ \DEQH:                                  \label{ln:alg-dirdeq-17} \nl
  \n       \If\ ($tail\_key == head\_key$) $\lbrace$    \label{ln:alg-dirdeq-18} \nl
  \n         send($cid$, \NACK);                        \label{ln:alg-dirdeq-19} \nl
  \p       $\rbrace$ \Else\ $\lbrace$                                            \nl
  \n         $head\_key++$;                             \label{ln:alg-dirdeq-20} \nl
             \Do\ $\lbrace$                             \label{ln:alg-dirdeq-21} \nl
  \n           $status$ = \DIRDELETE($head\_key$);      \label{ln:alg-dirdeq-22} \nl
  \p         $\rbrace$ \While\ ($status == \bot$);                               \nl
             send($cid$, $status$);                     \label{ln:alg-dirdeq-23} \nl 
  \p       $\rbrace$                                                             \nl
           \Break;                                                               \ul
  \p\p $\rbrace$
  \p
\end{code}
\end{algorithm}
\end{minipage}

Upon a message receipt, if $s_s$ receives a request \ENQT\ (line \lref{ln:alg-dirdeq-4}) 
it increments $tail\_key$ by one (line \lref{ln:alg-dirdeq-5}), 
and then sends the current value of $tail\_key$ to the client (line \lref{ln:alg-dirdeq-6}). 
The client uses the value  that $s_s$ sends to it, as the key for the data to insert in the directory. 
Likewise, if $s_s$ receives a request \ENQH\ (line \lref{ln:alg-dirdeq-14}), 
it sends the current value of $head\_key$ to the client (line  \lref{ln:alg-dirdeq-16}), 
and then decrements $head\_key$ by one (line \lref{ln:alg-dirdeq-15}). 

When a message of type \DEQT\ arrives (line \lref{ln:alg-dirdeq-7}), $s_s$ 
first checks whether the deque is empty (line \lref{ln:alg-dirdeq-8}).
If this is so, $s_s$ sends a \NACK\ to the client (line \lref{ln:alg-dirdeq-9}). 
Otherwise, the synchronizer repeatedly calls \texttt{\DIRDELETE($tail\_key$)} 	
to remove the element corresponding to a key equal to the value of $tail\_key$ from the directory
(line \lref{ln:alg-dirdeq-12}), and then decrements $tail\_key$ (line 
\lref{ln:alg-dirdeq-10}). Finally, $s_s$ 
sends the data to the client (line \lref{ln:alg-dirdeq-13}). 
The synchronizer performs similar actions for a \DEQH\ message, but instead of 
decrementing the $tail\_key$, it increments the $head\_key$
(line \lref{ln:alg-dirdeq-20}).

The code for the clients operations for enqueue, is presented in Algorithm 
\ref{alg-cenq}.
For enqueuing to the back of the deque, the client sends an \ENQT\ message 
to $s_s$ and blocks waiting for its response. When it receives the unique
key from  $s_s$ , the client is free to insert the element lazily. For enqueuing 
to the front of the deque, the client sends an \ENQH\ message and performs the 
same actions as for enqueuing to the back.

\begin{figure}
	\begin{minipage}{.45\textwidth}
		\begin{algorithm}[H]
\small
\caption{Enqueue operations for a client of the directory-based deque.}
\label{alg-cenq}
\begin{code}
  \firstline
     \void\ EnqueueTail(\integer\ $cid$, Data $data$) $\lbrace$			 \nl
  \n   $sid$ = get the synchronizer id;                \label{ln:ent01}          \nl
       send($sid, \langle\ENQT, cid\rangle$);          \label{ln:ent02}          \nl
       $key$ = receive(sid);                           \label{ln:ent03}          \nl
       \DIRINSERT($key, data$);                        \label{ln:ent04}          \ul
  \p $\rbrace$                                         \label{ln:ent05}          \nl
                                                       \label{ln:ent06}          \nl
     \void\ EnqueueHead(\integer\ $cid$, Data $data$) $\lbrace$ \nl
  \n   $sid$ = get the synchronizer id;                \label{ln:enh01}          \nl
       send($sid, \langle\ENQH, cid\rangle$);          \label{ln:enh02}          \nl
       $key$ = receive(sid);                           \label{ln:enh03}          \nl
       \DIRINSERT($key, data$);                        \label{ln:enh04}          \nl
  \p $\rbrace$
\end{code}
\end{algorithm}
	\end{minipage}
\hfill
\begin{minipage}{.45\textwidth}
\begin{algorithm}[H]
	\small
	\caption{Dequeue operation for a client of the directory-based deque.}
	\label{alg-cdeq}
	\begin{code}
		\firstline
		Data DequeueTail(\integer\ $cid$) $\lbrace$   \nl
		\n   $sid$ = get the synchronizer id;            \label{ln:dqt01}          \nl
		send($sid, \langle\DEQT, cid\rangle$);      \label{ln:dqt02}          \nl
		$status$ = receive($sid$);                  \label{ln:dqt03}          \nl
		\return\ $status$;                          \label{ln:dqt04}          \ul
		\p $\rbrace$                                     \ul
		\nl
		Data DequeueHead(\integer\ $cid$) $\lbrace$  \nl
		\n   $sid$ = get the synchronizer id;            \label{ln:dqh01}          \nl
		send($sid, \langle\DEQH, cid\rangle$);      \label{ln:dqh02}          \nl
		$status$ = receive($sid$);                  \label{ln:dqh03}          \nl
		\return\ $status$;                          \label{ln:dqh04}          \ul
		\p $\rbrace$
	\end{code}
\end{algorithm}
\end{minipage}
\end{figure}

The client code for dequeue to the front and dequeue to the back,
is presented in Algorithm \ref{alg-cdeq}. For dequeuing to 
the back of the deque, the client sends an \DEQT\ message to $s_s$ and blocks 
waiting for its response. The synchronizer performs the dequeue itself and 
sends back the response. For dequeuing to the front of the deque, 
the client sends an \DEQH\ message and performs the same actions as for 
enqueue to the back.

%

\subsubsection{Proof of Correctness}
Let $\alpha$ be an execution of the directory-based deque implementation.
We assign linearization points to enqueue and dequeue operations in $\alpha$ as follows: 

The linearization point of an enqueue back operation $op$ is placed in the configuration 
resulting from the execution of line~\lref{ln:alg-dirdeq-6} for $op$ by $s_s$. 
The linearization point of a dequeue back operation $op$ is placed in the configuration 
resulting from the execution of either line~\lref{ln:alg-dirdeq-9} or 
line~\lref{ln:alg-dirdeq-13} for $op$ (whichever is executed) by $s_s$. 
The linearization point of an enqueue front operation $op$ is placed in the configuration 
resulting from the execution of line~\lref{ln:alg-dirdeq-16} for $op$ by $s_s$. 
The linearization point of a dequeue front operation $op$ is placed in the configuration 
resulting from the execution of either line~\lref{ln:alg-dirdeq-19} or 
line~\lref{ln:alg-dirdeq-23} for $op$ (whichever is executed) by $s_s$.

\begin{lemma}
The linearization point of an enqueue (dequeue) operation $op$ executed by client $c$ is placed within its execution interval.
\end{lemma}

\begin{proof}
Assume that $op$ is an enqueue front (back) operation and let $c$ be the client that invokes it.
After the invocation of $op$, $c$ sends a message to $s_s$ (line \lref{ln:enh02}) and awaits a 
response from it. Recall that routine {\tt receive()} (line~\lref{ln:enh03}) blocks until a message 
is received.  The linearization point of $op$ is placed at the configuration resulting from the 
execution of line \lref{ln:alg-dirdeq-16} for $op$ by $s_s$. This line is executed after the request 
by $c$ is received, i.e. after $c$ invokes {\tt  EnqueueHead} ({\tt  EnqueueTail}). Furthermore, it 
is executed before $c$ receives the response by the server and thus, before {\tt  EnqueueHead} 
({\tt  EnqueueTail}) returns. Therefore, the linearization point is included in the execution 
interval of enqueue front (back).

The argumentation regarding dequeue front (back) operations is similar.
\end{proof}

Denote by $L$ the sequence of operations which have been assigned linearization 
points in $\alpha$ in the order determined by their linearization points. Let $C_i$ 
be the configuration in which the $i$-th operation $op_i$ of $L$ is linearized; 
denote by $C_0$ the initial configuration.  Denote by $\alpha_i$, the prefix of 
$\alpha$ which ends with $C_i$ and let $L_i$ be the prefix of $L$ up until the 
operation that is linearized at $C_i$. We denote by $head_i$ the value of the 
local variable $head\_key$ of $s_s$ at configuration $C_i$, and by $tail_i$ the 
value of the local variable $tail\_key$ of $s_s$ at $C_i$. By the pseudocode, 
we have that the initial values of $tail\_key$ and $head\_key$ are $0$; therefore, 
we consider that $head_0 = tail_0 = 0$. 

By analogous reasoning as the one followed in the case of the directory-based queue, 
inspection of the pseudocode leads to the following observations.

\begin{observation}
\label{obs:deq:seq}
Instances of Algorithm \ref{alg-dirdeq} are executed 
sequentially, i.e. their execution does not overlap.
\end{observation}

\begin{observation}
\label{obs:deq:mod}
Given two configurations $C_i$, $C_{i+1}$, $i \geq 0$, in $\alpha$, there is at most one 
step in the execution interval between $C_i$ and $C_{i+1}$ that modifies $tail\_key$. 
\end{observation}

Denote by $D_i$ the sequential deque that results if the operations of $L_i$ are 
applied sequentially to an initially empty queue. Let the size of $D_i$ (i.e. the 
number of elements contained in $D_i$) at $C_i$ be $d_i$. Denote by $sl_i^j$ the 
$j$-th element of $D_i$, $1 \leq j \leq d_i$. 
Each element of $D_i$ is a pair of type $\langle key, data \rangle$ where
the elements from the bottommost to the topmost are assigned integer keys 
as follows: Let $f_i$ be the key of element $sl_i^1$ and $l_i$ be the key 
of element $sl_i^{d_i}$ in some configuration $C_i$. 
We denote the $key$ field of the $\langle data, key \rangle$ pair that 
comprises some element $sl_i^j$, $1 \leq j \leq d_i$, of $D_i$ by $sl_i^j.key$.
Then, if $d_i > 1$, $sl_i^{j+1}.key = sl_i^j.key + 1$, $1 \leq j \leq d_i$.
We consider that if $op_1$ is an enqueue front operation, then $f_1 = l_1 = 0$, 
while if it is an enqueue back operation, then $f_1 = l_1 = 1$.
Notice that $l_i - f_i + 1 = d_i$.

Consider  a sequence of elements $S$. If $e$ is the first element of $S$, 
we denote by $S \setminus_f e$ the suffix of $S$ that results by removing 
only element $e$ from the first position of $S$. If $e$ is the last element 
of $S$, we denote by $S \setminus_b e$ the prefix of $S$ that results by 
removing only element $e$ from the last position of $S$. If $e$ is an element 
not included in $S$, we denote by $S' = S \cdot e$ the sequence that results 
by appending element $e$ to the end of $S$, and by $S'' = e \cdot S$ the 
sequence that results by prefixing $S$ with element $e$.

\begin{lemma}
\label{lem:tail}
For each integer $i \geq 1$, the following hold at $C_i$:
\begin{compactenum}
\item\label{cdt1} If $op_i$ is an enqueue back operation, then $tail_i$ = $tail_{i-1} + 1$. 
\item\label{cdt2} If $op_i$ is a dequeue back operation, then it holds that $tail_i$ = $tail_{i-1} - 1$ if $tail_{i-1} \neq head_{i-1}$; otherwise, $tail_i = tail_{i-1}$.
\end{compactenum}
\end{lemma}

\begin{proof}
Fix any $i \geq 1$. 
If $op_i$ is an enqueue back operation, the linearization point of $op_i$ is placed at the 
configuration resulting from the execution of line \lref{ln:alg-dirdeq-6}. By inspection of 
the pseudocode (lines \lref{ln:alg-dirdeq-5}-\lref{ln:alg-dirdeq-6}), we have that in the 
instance of Algorithm \ref{alg-dirdeq} executed for $op_i$, $tail\_key$ is incremented before 
it is sent to the client $c$ that invoked $op_i$. By Observations \ref{obs:deq:seq} and 
\ref{obs:deq:mod}, this is the only increment that occurs on $tail\_key$ between $C_{i-1}$ 
and $C_i$. Thus, Claim \ref{cdt1} holds.

If $op_i$ is a dequeue back operation, the linearization point of $op_i$ is placed either 
at the configuration resulting from the execution of line \lref{ln:alg-dirdeq-9} or at the 
configuration resulting from the execution of line \lref{ln:alg-dirdeq-13}. 
%
Let $op_{i-1}$ be a dequeue back operation that is linearized 
at the execution of line \lref{ln:alg-dirdeq-9}. By inspection of the pseudocode (line \lref{ln:alg-dirdeq-8}), 
this occurs only in case $head_{i-1} = tail_{i-1}$. Since the execution of this line does not modify $tail\_key$, 
$tail_i = tail_{i-1}$ and Claim \ref{cdt2} holds. 

Now let $op_i$ be a dequeue back operation that is linearized at the configuration resulting from the execution of 
line \lref{ln:alg-dirdeq-13}. By the pseudocode (lines  \lref{ln:alg-dirdeq-10}-\lref{ln:alg-dirdeq-13}) and by 
Observation \ref{obs:seq}, it follows that the execution of line \lref{ln:alg-dirdeq-10} is the only step in which 
$tail\_key$ is modified in the execution interval between $C_{i-1}$ and $C_i$. 
Since line \lref{ln:alg-dirdeq-13} is executed, it holds that the condition of the \If\ 
clause of line \lref{ln:alg-dirdeq-8} evaluates to \false, i.e. it holds that $head_{i-1} \neq tail_{i-1}$. 
Furthermore, because of the execution of line \lref{ln:alg-dirdeq-10}, $tail_i = tail_{i-1} - 1$. 
Thus, Claim \ref{cdt2} holds.
\end{proof}


\begin{lemma}
\label{lem:head}
For each integer $i \geq 1$, the following hold at $C_i$:
\begin{compactenum}
\item\label{cdh1} In case $op_i$ is an enqueue front operation, then, if $i > 1$ and 
$op_{i-1}$ is an enqueue front operation, it holds that $head_i$ = $head_{i-1} - 1$; 
otherwise $head_i = head_{i-1}$. 
\item\label{cdh2} In case $op_i$ is a dequeue front operation, then, 
if $head_{i-1} \neq tail_{i-1}$, $i>1$, and $op_{i-1}$ is not an enqueue front operation, 
it holds that $head_i$ = $head_{i-1} + 1$ ; otherwise $head_i = head_{i-1}$.
\end{compactenum}
\end{lemma}

\begin{proof}
Fix any $i \geq 1$. 
Let $op_i$ be an enqueue front operation. If $i=1$, then by inspection of the pseudocode, 
we have that $head\_key$ is not modified before the execution of line \lref{ln:alg-dirdeq-16}.
Since $head_0 = 0$ and the execution of line \lref{ln:alg-dirdeq-16} does not modify $head\_key$, 
it follows that $head_1 = 0 = head_0$ and Claim \ref{cdh1} holds. Now let $i>1$. By inspection 
of the pseudocode and by Observation \ref{obs:deq:seq} we have that $head\_key$ is not modified 
by enqueue back and dequeue back operations. By the pseudocode, Observation \ref{obs:deq:seq}
and the way linearization points are assigned, we have that although $head\_key$ is modified by 
dequeue front operations only before the configuration in which the operation is linearized, it 
is modified by enqueue front operations in the step (line \lref{ln:alg-dirdeq-15}) right after 
the configuration in which an enqueue front operation is linearized. Therefore, if $op_{i-1}$ 
is an enqueue front operation, then $head\_key$ is decremented once (line \lref{ln:alg-dirdeq-15}) 
in the execution interval between $C_{i-1}$ and $C_i$. Thus, if $op_{i-1}$ is an 
enqueue front operation, then $head_i = head_{i-1}+1$, while if $op_{i-1}$ is any other 
type of operation, $head_i = head_{i-1}$. Thus, Claim \ref{cdh1} holds. 

Now let $op_i$ be a dequeue operation. If $i=1$, then by inspection of the pseudocode, 
we have that $head\_key$ is not modified before the execution of line \lref{ln:alg-dirdeq-10}.
By the pseudocode and by Observation \ref{obs:deq:seq}, $tail\_key$ is not modified as well 
before the execution of line \lref{ln:alg-dirdeq-10}. Thus, $head_0 = 0 = tail_0$ and the \If\ 
condition of line \lref{ln:alg-dirdeq-8} evaluates to \true. Then, $op_1$ is linearized in the 
configuration resulting from the execution of line \lref{ln:alg-dirdeq-19}. Notice that the 
execution of this line does not modify $head\_key$. It follows that $head_1 = 0 = head_0$ and 
that Claim \ref{cdh2} holds. 

Now let $i>1$. The linearization point of $op_i$ may be placed at the configuration resulting 
from the execution of line \lref{ln:alg-dirdeq-19} or line \lref{ln:alg-dirdeq-23}, whichever 
is executed by $s_s$ for it. Let the linearization point be placed in the configuration resulting 
from the execution of line \lref{ln:alg-dirdeq-19}. In that case, $head_{i-1} = tail_{i-1}$. 
Notice that the execution of that line does not modify $head\_key$. Therefore, $head_i = head_{i-1}$, 
and Claim \ref{cdh2} holds. Now let the linearization point be placed in the configuration resulting 
from the execution of line \lref{ln:alg-dirdeq-23}. In case $op_{i-1}$ is an enqueue back or dequeue 
tail operation, $head\_key$ is not modified by it. Therefore, since line \lref{ln:alg-dirdeq-20} is 
executed before line \lref{ln:alg-dirdeq-23}, $head_i = head_{i+1} + 1$ and Claim \ref{cdh2} holds. 
The same also holds if $op_{i-1}$ is a dequeue front operation. If $op_{i-1}$ is an enqueue front 
operation, then by inspection of the pseudocode (line \lref{ln:alg-dirdeq-15}), we have that 
$head\_key$ is decremented in the step following the configuration in which $op_{i-1}$ is 
linearized. Therefore, in this case and by Observation \ref{obs:deq:seq}, $head\_key$ is 
decremented and then incremented 
once in the execution interval between $C_{i-1}$ and $C_i$. This in turn implies that 
$head_i = head_{i-1} -1 +1 = head_{i-1}$ and Claim \ref{cdh2} holds.
\end{proof}

Recall that $sl_i^j.key = f_i+j-1$ 
or $sl_i^j.key = l_i-j+1$ .
By inspection of the pseudocode (lines \lref{ln:alg-dirdeq-16}/\lref{ln:alg-dirdeq-6}), 
we see that, when $op_i$ is an enqueue front/back operation, $head_i$/$tail_i$ is sent 
by $s_s$ to the client $c$ that invoked $op_i$. By further inspection of the pseudocode 
(lines \lref{ln:enh03}-\lref{ln:enh04}/\lref{ln:ent04}-\lref{ln:ent05}), we see that $c$ 
uses $head_i$/$tail_i$ as the $key$ field of the element it enqueues, i.e. uses it as argument 
for auxiliary function {\tt DirInsert()/DirDelete()}. When $op_i$ is a dequeue front/back 
operation, by inspection of the pseudocode (lines \lref{ln:alg-dirdeq-19}/\lref{ln:alg-dirdeq-9}), 
we have that when $head\_key = tail\_key$, $s_s$ sends \NACK\ to $c$, and that when $c$ receives 
\NACK, it does not enqueue any element and instead, returns $\bot$ (lines \lref{ln:dqh03}-\lref{ln:dqh04}/\lref{ln:dqt03}-\lref{ln:dqt04}). 
When $head\_key \neq tail\_key$, $s_s$ uses $head_i$/($tail_i+1)$ as the $key$ field in order 
to determine which element to dequeue (lines \lref{ln:alg-dirdeq-22}/\lref{ln:alg-dirdeq-12}). 
Then, the following observation holds.

\begin{observation}
\label{obs:deq:enq-deq}
If $op_i$ is an enqueue back operation, it inserts a pair with $key = tail_i$ into the directory.
If $op_i$ is a dequeue back operation, then, if $head_i \neq tail_i$, it removes a pair with $key = tail_i+1$ from the directory; 
if $head_i = tail_i$, it does not remove any pair from the directory.
If $op_i$ is an enqueue front operation, it inserts a pair with $key = head_i$ into the directory.
If $op_i$ is a dequeue front operation then, if $head_i \neq tail_i$, it removes a pair with $key = head_i$ from the directory; 
if $head_i = tail_i$, it does not remove any pair from the directory.
\end{observation}

\begin{lemma}
\label{lemma:deq:seq1}
At $C_i$, $i \geq 1$, the following hold:
\begin{compactenum}
\item \label{cd13} If $op_i$ is an enqueue back operation, then $tail_i = sl_i^{d_i}.key$.  
\item \label{cd14} If $op_i$ is a dequeue back operation, then if $D_{i-1} \neq \epsilon$, 
$tail_i = sl_{i-1}^{d_{i-1}}.key$. If $D_{i-1} = \epsilon$, then $head_i = tail_i$.
\item \label{cd23} If $op_i$ is an enqueue front operation, then $head_i = sl_i^1.key$.  
\item \label{cd24} If $op_i$ is a dequeue front operation, then if $D_{i-1} \neq \epsilon$, 
$head_i = sl_{i-1}^1.key$. If $D_{i-1} = \epsilon$, then $head_i = tail_i$.
\end{compactenum}
\end{lemma}

\begin{proof}
We prove the claims by induction. 

{\bf Base case.} 
We prove the claim for $i$ = $1$.

Consider the case where $op_1$ is an enqueue back operation. 
Then, $d_1 = 1$ and by definition, $D_1$ contains only the pair $\langle 1, data \rangle$.
By Observation \ref{obs:deq:seq}, it is the first operation in $\alpha$ for 
which an instance of Algorithm \ref{alg-dirdeq} is executed by $s_s$. Therefore, 
by Lemma \ref{lem:tail}, $tail_1 = tail_0 + 1 = 1$. Thus, $tail_1 = sl_1^{d_1}.key$ 
and Claim \ref{cd13} holds.

Next, consider the case where $op_1$ is a dequeue back operation. By Observation \ref{obs:deq:seq}, 
$op_1$ is the first operation in $\alpha$ for which an instance of Algorithm \ref{alg-dirdeq} 
is executed by $s_s$. Notice that then, $Q_1 = \epsilon$. Therefore, by Lemma \ref{lem:tail}, 
$tail_1 = tail_0 = 0$. Since $head\_key$ is not modified by dequeue back operations, 
$head_1 = head_0 = 0$. Thus, $head_1 = tail_1$, so Claim \ref{cd14} holds. 

Next, consider the case where $op_1$ is an enqueue front operation. 
Again, by definition, $d_1 = 1$ and $D_1$ contains only the pair  $\langle 0, data \rangle$. 
By Observation \ref{obs:deq:seq}, it is the first operation in $\alpha$ for 
which an instance of Algorithm \ref{alg-dirdeq} is executed by $s_s$. Therefore, 
by Lemma \ref{lem:head}, $head_1 = head_0 = 0$. Thus, $head_1 = sl_1^1.key$ 
and Claim \ref{cd23} holds.

Finally, consider the case where $op_1$ is a dequeue front operation. 
By Observation \ref{obs:deq:seq}, $op_1$ is the first operation in $\alpha$ for which 
an instance of Algorithm \ref{alg-dirdeq} is executed by $s_s$. 
Notice that then, $Q_1 = \epsilon$. Therefore, by Lemma \ref{lem:head}, $head_1 = head_0 = 0$. 
Since $tail\_key$ is not modified by dequeue front operations, $tail_1 = tail_0 = 0$. 
Thus, $head_1 = tail_1$ and Claim \ref{cd24} holds.

{\bf Hypothesis.} Fix any $i$, $i > 0$ and assume that the lemma holds at $C_i$. 

{\bf Induction step.} We prove that the claims also hold at $C_{i+1}$. 
Assume that $op_{i+1}$ is an enqueue back operation. By the induction hypothesis, 
if $op_i$ is an enqueue back operation, then $sl_i^{d_i}.key = tail_i = l_i$. 
Similarly, if $op_i$ is a dequeue back operation, then by the induction hypothesis,
$sl_{i-1}^{d_{i-1}}.key = tail_i$. Since the dequeue back operation removes the 
last element in $D_{i-1}$, it follows that the last element $sl_i^{d_i}$ of $D_i$ 
is $sl_{i-1}^{d_{i-1}}$. Thus, here also, $tail_i = sl_i^{d_i}.key = l_i$.
Notice that enqueue front and dequeue front operations do not modify $tail\_key$. 
Since these types of operation do not affect the back of the sequential dequeue, 
it still holds that $tail_i = sl_i^{d_i}.key = l_i$.
Since $op_{i+1}$ is an enqueue back operation, by Lemma \ref{lem:tail}, we have that 
$tail_{i+1} = tail_i + 1$. By Observation \ref{obs:deq:enq-deq}, we have that the client 
$c$ that initiated $op_{i+1}$ inserts a pair with $key = tail_{i+1} = tail_i + 1$ into 
the directory. By definition, $sl_{i+1}^{d_{i+1}}.key = sl_i^{d_i}.key + 1$. 
Thus, $sl_{i+1}^{d_{i+1}}.key = tail_i + 1$, and Claim \ref{cd13} holds. 

Now assume that $op_{i+1}$ is a dequeue back operation.
We examine two cases. First, let $D_i \neq \epsilon$. 
By Lemma \ref{lem:tail}, it then holds that $tail_{i+1} = tail_i - 1$.
By Observation \ref{obs:deq:enq-deq}, we have that a pair with $key = tail_{i+1} = tail_i - 1$ 
is removed from the directory. By definition, we have that $D_{i+1} = D_i \setminus_b sl_i^{d_i}$. 
Also by definition, we have that $sl_i^{d_i}.key = sl_i^{d_i-1}.key + 1$. Because of $op_{i+1}$, 
$sl_i^{d_i-1} = sl_{i+1}^{d_{i+1}}$. Since $tail_{i+1} = tail_i - 1$, Claim \ref{cd14} holds.
Now let $D_i = \epsilon$. In this case, $op_{i+1}$ cannot have any effect on the state of the 
deque. By inspection of the pseudocode, this corresponds to the operation being linearized in the 
configuration resulting from the execution of line \lref{ln:alg-dirdeq-9}. Notice that in order
for this to be the case, the \If\ condition of line \lref{ln:alg-dirdeq-8} must evaluate to \true.
This occurs if $head_i = tail_i$, thus Claim \ref{cd14} holds.

Next assume that $op_{i+1}$ is an enqueue front operation.
By the induction hypothesis, if $op_i$ is an enqueue front operation, 
then $sl_i^1.key = head_i = f_i$. By Lemma \ref{lem:head}, it holds 
then that $head_{i+1} = head_i-1$. Since $op_{i+1}$ is an enqueue front 
operation, it prepends an element to $D_i$ and therefore, $sl_i^1 = sl_{i+1}^2$. 
By definition of $D_{i+1}$, $sl_{i+1}^1.key = sl_{i+1}^2.key-1$.
Since $sl_{i+1}^2.key = head_i$, $sl_{i+1}^1.key = head_{i+1}$ and 
Claim \ref{cd23} holds. 

On the other hand, if $op_i$ is a dequeue front operation, then by the 
induction hypothesis, $sl_{i-1}^1.key = head_i$. By Lemma \ref{lem:head}, 
it also follows that in this case, $head_{i+1} = head_i$. Notice that 
by definition, $op_i$ removes element $sl_{i-1}^1$ from $D_{i-1}$. 
Then, for element $sl_i^1$ of $D_i$, by definition, $sl_i^1.key = sl_{i-1}^1.key+1$. This means that $sl_{i-1}^1.key = head_i = sl_i^1.key-1$. 
Since $head_{i+1} = head_i$, Claim \ref{cd23} holds.

Notice that enqueue back and dequeue back operations do not modify 
$head\_key$.

Finally, assume that $op_{i+1}$ is a dequeue front operation. We examine two cases. First, let $D_i \neq \epsilon$. 
By Lemma \ref{lem:head}, it then holds that $head_{i+1} = head_i$.
By Observation \ref{obs:deq:enq-deq}, we have that a pair with $key = head_{i+1} = head_i$ 
is removed from the directory. By definition, we have that $D_{i+1} = D_i \setminus_b sl_i^1$. 
Also by definition, we have that $sl_i^1.key = sl_i^2.key-1$. Because of $op_{i+1}$, 
$sl_i^2 = sl_{i+1}^1$. Since $head_{i+1} = head_i$, Claim \ref{cd24} holds.
Now let $D_i = \epsilon$. In this case, $op_{i+1}$ cannot have any effect on the state of the 
deque. By inspection of the pseudocode, this corresponds to the operation being linearized in the 
configuration resulting from the execution of line \lref{ln:alg-dirdeq-19}. Notice that in order
for this to be the case, the \If\ condition of line \lref{ln:alg-dirdeq-18} must evaluate to \true.
This occurs if $head_i = tail_i$, thus Claim \ref{cd24} holds.
\end{proof}

%
%

From the above lemma, we have the following corollary.

\begin{corollary}
\label{cor:deq:empty}
$D_i = \epsilon$ if and only if $head_i = tail_i$. 
\end{corollary}

\begin{lemma}
\label{lem:lin-b}
If $op_i$ is a dequeue back operation, then it returns the value of the field $data$ 
of $sl_{i-1}^{d_{i-1}}$ or $\bot$ if $D_{i-1} = \epsilon$.
\end{lemma}

\begin{proof}
Consider the case where $D_{i-1} \neq \epsilon$. By definition of $D_i$, we have that 
$D_i = D_{i-1} \setminus_b  sl_{i-1}^{d_{i-1}}$. Let $op_j$ be the enqueue operation that 
is linearized before $op_i$ and inserts an element with key $tail_i+1$ to the queue. Notice 
by the pseudocode (lines \lref{ln:alg-dirdeq-11}-\lref{ln:alg-dirdeq-13}), that the parameter 
of \DIRDELETE\ is $tail_i+1$. 
By the semantics of \DIRDELETE, if at the point that the instance of \DIRDELETE\ is 
executed in the \Do\ - \While\ loop of  lines \lref{ln:alg-dirdeq-11}-\lref{ln:alg-dirdeq-12} 
for $op_i$, the instance of \DIRINSERT\ of $op_j$ has not yet returned, then \DIRDELETE\ returns 
$\langle \bot, -\rangle$.

By Lemma \ref{lem:tail}, $tail_i+1$ is the $key$ of the last pair $sl_{i-1}^{d_{i-1}}$ 
in $D_{i-1}$. Therefore, when \DIRDELETE\ returns a $status \neq \bot$, it holds 
that it returns the $data$ field of $sl_{i-1}^{d_{i-1}}$, the last element in $D_{i-1}$.
Notice that this value is sent to the client $c$ that invoked $op_i$ (line \lref{ln:alg-dirdeq-13}) 
and that $c$ uses this value as the return value of $op_i$ (lines \lref{ln:dqt03}-\lref{ln:dqt04}).
Thus, the claim holds. 

Now consider the case where $D_{i-1} = \epsilon$.
Since, by Corollary \ref{cor:deq:empty}, when this is the case, $head_i = tail_i$, 
\NACK\ is sent $c$ and, by inspection of the pseudocode, $op_i$ returns $\bot$, 
i.e. the claim holds. 
\end{proof}

In a similar fashion, we can prove the following.

\begin{lemma}
\label{lem:lin-f}
If $op_i$ is a dequeue front operation, then it returns the value of the field 
$data$ of $sl_{i-1}^1$ or $\bot$ if $D_{i-1} = \epsilon$.
\end{lemma}

From the above lemmas we have the following:

\begin{theorem}
The directory-based deque implementation is linearizable.  
\end{theorem}
\subsection{Hierarchical approach, Elimination, and Combining.}
In this section, we outline how the hierarchical approach, 
described in Section~\ref{sec:impl-paradigms}, is applied to the 
directory-based designs.

Each island master $m_i$ performs the necessary communication between the clients of its island
and $s_s$. 
In the stack implementation,
each island master applies elimination 
before communicating with $s_s$. 
To further reduce communication with $s_s$, 
$m_i$ applies a technique known as combining~\cite{ultracomputer1982}.
In the case of stack, once elimination has been applied, there is only one type of requests
that must be sent to the synchronizer; for all these requests, $m_i$ sends just one message
containing their number $f$  and 
their type to the synchronizer. 
In case of push operations, this method allows the synchronizer to directly 
increment $top$ by $f$ and respond to $m_i$ with the value $g$ that $top$ had 
before the increment. Once $m_i$ receives $g$, it informs the clients (which initiated these requests)
that the keys for their requests are $g, g+1, \ldots, g+f-1$. 
In the case of queue, 
each message of $m_i$ to $s_s$ contains two counters counting 
the number of active enqueue and dequeue requests from clients of island $i$. 
When $s_s$ receives such a message it responds with a message containing the current
values of $tail$ and $head$. It then increments $tail$ and $head$ 
by the value of the first and second counter, respectively. Server $m_i$ 
assigns unique keys to active enqueue and dequeue operations, based on the value of 
$tail$ and $head$ it receives, in a way similar as in stacks.
Combining can be used for deques (in addition to elimination) in ways similar 
to those described above.

\section{Token-based Stacks, Queues, and Deques}
\label{app:token}
\ignore{
The servers are arranged on a logical ring, based on their id's. 
The servers have been placed on the ring so that the server $s_i$ with id $i$, $0 < i < \maxser-1$, 
is followed by the server $s_{i+1}$ with id $i+1$ and is preceded by the server $s_{i-1}$ with id $i-1$. 
The server $s_0$ with id $0$ is followed by the server $s_1$ with id $1$ and is preceded by 
the server $s_{\maxser -1}$ with id $\maxser-1$; similarly,
server $s_{\maxser-1}$ is followed by server $s_0$ and is preceded by server $s_{\maxser -2}$.
This order is known to the servers, and remains unchanged during the algorithm's execution. 
Each server $s$ has allocated a memory chunk in its local memory for storing 
elements of the implemented data structure. However, only if a server holds a token
is allowed to simulate operations. Such a server is called {\em token server}. 
}

We start with an informal description of the token-based technique that 
we present in this section. We assume that the servers are numbered from 
$0$ to $\maxser\ - 1$ and form a logical ring. 
Each server has allocated a chunk of memory (e.g. one or a few pages) of a predetermined size, where it stores elements 
of the implemented DS. 
A DS implementation 
employs (at least) one token which identifies the server $s_t$, called the {\em token server}, 
at the memory chunk of which newly inserted elements are stored. 
(A second token is needed in cases of queues and deques.)
When the chunk of memory allocated by the token server becomes full, 
the token server gives up its role and appoints another (e.g. the next) server 
as the new token server. 
A client remembers the server that served its last request 
and submits the next request it initiates
to that server;
so, each response to a client contains the id of the server that served the client's request. 
Servers that do not have the token for handling a request,
forward the request to subsequent servers; this is done until the request reaches the appropriate token server.
A server allocates a new (additional) chunk of memory 
every time the token reaches it (after having completed one more round of the ring)
and gives up the token when this chunk becomes full.

Section~\ref{app:stack-token} presents the details of the token-based distributed stack. 
The token-based queue implementation appears in Section~\ref{app:token-queue}.
Section~\ref{app:token-deque} provides the token-based deque.
We start by presenting {\em static versions} of the implementations, 
i.e. versions in which the total memory allocated for the data structure is predetermined during an execution 
and once it is exhausted the data structure becomes full and no more insertions of elements
can occur. We then describe in Section~\ref{app:directory-dyn}, how to take dynamic
versions of the data structures from their static analogs.

\subsection{Token-Based Stack}
\label{app:stack-token}

To implement a distributed stack, each server uses its allocated memory chunk to maintain a local stack, $lstack$. 
Initially, $s_t$ is the server with id $0$. 
To perform a push (or pop), a client $c$ 
sends a push (or pop) request to the server that has served $c$'s last request (or, initially, to server $0$)
and awaits for a response. 
If this server is not the current token server at the time that it receives the request, 
it forwards the request to its next or previous server,
depending on whether its local stack is full or empty, respectively. This is repeated until
the request reaches the server $s_t$ that has the token
which pushes the new element onto its local stack and sends an \ACK\ to $c$. 
If $s_t$'s local stack does not have free space to accommodate the new element, 
it sends the push request of $c$, together with an indication that 
it gives up its token, to the next server. 
\POP\ is treated by $s_t$ in a similar way.

\subsubsection{Algorithm Description}

	\begin{minipage}{.45\textwidth}
		\begin{algorithm}[H]
\caption{Events triggered in a server of the token-based stack.}
\label{alg8}
\begin{code}
  \lreset
  \firstline
     LocalStack $lstack$ = $\varnothing$;                                           \nl
     \integer\ $my\_sid$;                          \cm{each server has a unique id} \nl
     \integer\ $token = 0$;                                                         \ul
                                                                                    \nl
     a message $\langle op, data, id, tk\rangle$ is received:                       \nl
  \n   \Switch\ ($op$) $\lbrace$                                                    \nl
  \n     \Case\ \PUSH:                                                \label{ln:63} \nl
  \n       \If\ ($tk$ == \TOKEN) $token = my\_sid$;                   \label{ln:60} \nl
           \If\ ($token \neq my\_sid$) $\lbrace$                      \label{ln:64} \nl
  \n         send($token, \langle op, data, id, tk\rangle$);          \label{ln:65} \nl
             \Break;                                                                \ul
  \p       $\rbrace$                                                                \nl
           \If\ (!IsFull($lstack$)) $\lbrace$                         \label{ln:66} \nl
  \n         push($lstack, data$);                                   \label{ln:lp3} \nl
             send($id, \langle\ACK, my\_sid\rangle$);                \label{ln:lp1} \nl
  \p       $\rbrace$ \Elseif\ ($my\_sid \neq$ \maxser-1) $\lbrace$    \label{ln:67} \nl 
  \n         $token$ = find\_next\_server($my\_sid$);                 \label{ln:68} \nl
             send($token, \langle op, data, id, \TOKEN\rangle$);      \label{ln:70} \nl
  \p       $\rbrace$ \Else\ 
                     \cm{It's the last server in the order, thus the stack is full} \nl
  \n          send($id, \langle\NACK, my\_sid\rangle$);               \label{ln:71} \nl
  \p       \Break;                                                                  \nl
  \p     \Case\ \POP:                                                 \label{ln:72} \nl
  \n       \If\ ($tk$ == \TOKEN) $token = my\_sid$;                   \label{ln:61} \nl
           \If\ ($token \neq my\_sid$) $\lbrace$                      \label{ln:73} \nl
  \n         send($token, \langle op, data, id, tk\rangle$);          \label{ln:74} \nl
             \Break;                                                                \ul
  \p       $\rbrace$                                                                \nl
           \If\ (!IsEmpty($lstack$)) $\lbrace$                        \label{ln:75} \nl
  \n         $data$ = pop($lstack$);                                \label{ln:75.1} \nl
             send($id, \langle data, my\_sid\rangle$);               \label{ln:lp2} \nl
  \p       $\rbrace$ \Elseif\ ($my\_sid \neq 0$) $\lbrace$            \label{ln:76} \nl
  \n         $token$ = find\_previous\_server($my\_sid$);             \label{ln:77} \nl
             send($token, \langle op, data, id, \TOKEN\rangle$);      \label{ln:79} \nl
  \p       $\rbrace$ \Else\ 
                   \cm{It's the first server in the order, thus the stack is empty} \nl
  \n         send($id, \langle\NACK, my\_sid\rangle$);                \label{ln:80} \nl
  \p       \Break;                                                                  \ul
  \p\p $\rbrace$                                                                     
  \p
\end{code}
\end{algorithm}
\end{minipage}

Initially the elements are stored 
in the memory space allocated by server $s_0$, the first server in the 
ring. At this point, $s_0$ is the token server;
the token server manages the top of the stack. 
Once the memory chunk of the token server becomes full, 
the token server notifies the next server ($s_1$) 
in the ring to become the new token server.

The pseudocode for the server is presented in Algorithm \ref{alg8}. Each server $s$, apart 
from a local stack ($lstack$), maintains also a local variable $token$ which identifies 
whether $s$ is the token server. The messages that are transmitted during the execution 
are of type \PUSH\ and \POP, which are sent from clients that want to 
perform the mentioned operation to the servers, or are forwarded from any 
server towards the token server. Each message has four fields: (1) $op$ 
with the operation to be performed, (2) $data$, containing data in 
case of \ENQ\ and $\bot$ otherwise, (3) $id$ that contains the id 
of the sender and (4) a one-bit flag $tk$ which is set to \TOKEN\ only when a 
forwarded message denotes also a token transition. 

If the message is of type \PUSH\ (line \lref{ln:63}), $s$ first checks whether 
the message contains a token transition. If $tk$ is marked with \TOKEN, 
$s$ changes the $token$ variable to contain its id (line \lref{ln:60}). 
If $s$ is not a token server, it just forwards the message to the next
server (line \lref{ln:65}). Otherwise, it 
checks if there is free space in $lstack$ to store the new request 
(line \lref{ln:66}). If there is such space, the server pushes the data to the 
stack, and sends back an \ACK\ to the client. In this implementation, the 
\texttt{push()} function (line \lref{ln:lp3}) does not need to return any value, 
since the check for memory space has already been performed by the server on 
line \lref{ln:66}, hence \texttt{push()} is always successful.

If $s$ does not have any free space, it must notify the next server to become 
the new token server. More specifically, if $s$ is not the server with id 
$\maxser-1$ (line \lref{ln:67}), it forwards to the next server the \PUSH\ 
message it received from the client, after setting the $tk$ field to 
\TOKEN\ (line \lref{ln:70}). On the other hand, if $s$ is the server with id 
$\maxser -1$,  all previous servers have no memory space available to store a 
stack element. In this case, $s$ sends back to the client a message \NACK 
(line \lref{ln:71}).

If the message is of type \POP\ (line \lref{ln:72}) similar actions take 
place: $s$ checks whether the message contains a token transition and if 
its true, it changes its local variable $token$ appropriately. Then $s$ 
checks if it is the token server (line \lref{ln:73}). If not, it just 
forwards the message towards the server it considers as the token server 
(line \lref{ln:74}). If $s$ is the token server, it checks if its local 
stack is empty (line \lref{ln:75}). If it is not empty, the pop operation 
can be executed normally. At the end of the operation, $s$ sends to the 
client the data of the previous top element (line~\lref{ln:lp2}). In case 
of an empty local stack, if $s$ is  not $s_0$ (line \lref{ln:76}), 
it forwards to the previous server the client's \POP\ message, after setting 
the $tk$ field to \TOKEN\ (line \lref{ln:79}). On the other hand, 
if the server that received the \POP\ request is $s_0$ ($id==0$), 
then all the servers have empty stacks and the server sends back to the client 
a \NACK\ message (line \lref{ln:80}).

\begin{figure}
	\begin{minipage}{.45\textwidth}
		\begin{algorithm}[H]
\caption{Push operation for a client of the token-based stack.}
\label{alg9}
\begin{code}
  \firstline
     $sid=0$;                                                 \label{ln:81} \nl
     Data ClientPush(\integer\ $cid$, Data $data$) $\lbrace$                \nl
  \n   send($sid, \langle\PUSH, data, cid, \bot\rangle$);     \label{ln:82} \nl
       $\langle status, sid\rangle$ = receive();              \label{ln:83} \nl
       \return\ $status$;                                     \label{ln:84} \ul
  \p $\rbrace$
\end{code}
\end{algorithm}	
\end{minipage}
\hfill
\begin{minipage}{.45\textwidth}
\begin{algorithm}[H]
	\caption{Pop operation for a client of the token-based stack.}
	\label{alg10}
	\begin{code}
		\firstline
		$sid=0$;                                           \label{ln:85}    \nl
		Data ClientPop(\integer\ $cid$) $\lbrace$                        \nl
		\n   send($sid, \langle\POP, \bot, cid, \bot\rangle$);\label{ln:86} \nl
		$\langle status, sid\rangle$ = receive();        \label{ln:87} \nl
		\return\ $status$;                                             \ul
		\p $\rbrace$
	\end{code}
\end{algorithm}
\end{minipage}
\end{figure}

The clients execute the operations push and pop, by calling the functions 
\texttt{ClientPush()} and \texttt{ClientPop()}, respectively. Each of these 
functions sends a message to the server. Initially, the clients forward their 
requests to $s_0$. Because the server that maintains the top element might 
change though, the clients update the $sid$ variable through a lazy 
mechanism. When a client $c$ wants to perform an operation, it sends a request 
to the server with id equal to the value of $sid$ (lines \lref{ln:81} 
and \lref{ln:85}). If the message was sent to an incorrect server, it is forwarded 
by the servers till it reaches the server that holds the token. That server is 
going to respond with the status value of the operation and with the its id. 
This way, $c$ updates the variable $sid$. 

During the execution of the \texttt{ClientPush()} function, described in Algorithm 
\ref{alg9}, the client sends a \PUSH\ message to the server with id $sid$ 
(line \lref{ln:82}). It then, waits for its response (line \lref{ln:83}). When the 
client receives the response, it updates the $sid$ variable (line \lref{ln:83}) 
and returns the $status$. The $status$ is either \ACK\ for a successful 
push, or \NACK\ for a full stack.

The \texttt{ClientPop()} function operates in a similar fashion. 
The client sends a \POP\ message to the server with id $sid$ 
(line \lref{ln:86}). It then, waits for its response (line \lref{ln:87}). 
The server responds with a \NACK\ (for empty queue), or with the value 
of the top element (otherwise). The server also forwards its id, 
which is stored in client's variable $sid$. The client finally, 
returns the $status$ value and terminates.

\subsubsection{Proof of Correctness}
Let $\alpha$ be an execution of the token-based stack algorithm presented 
in Algorithms \ref{alg8}, \ref{alg9}, and \ref{alg10}. Let $op$ be any 
operation in $\alpha$. We assign a linearization point to $op$ 
by considering the following cases: 
\begin{compactitem}
\item {\em $op$ is a push operation.} 
Let $s_t$ be the token server that responds to the client that initiated $op$ 
(i.e. the {\tt receive} of line \lref{ln:83} in the execution of $op$ receives 
a message from $s_t$). If $op$ returns \ACK, the linearization point is placed at the 
configuration resulting from the execution of line \lref{ln:lp1} by $s_t$ for $op$.
Otherwise, the linearization point of $op$ is placed at the configuration resulting 
from the execution of line \lref{ln:71} by $s_t$ for $op$.
\item {\em $op$ is a pop operation.} Let $s_t$ be the token server that responds to 
the client that initiated $op$ (line \lref{ln:87}).
If the operation returns \NACK, the linearization point of $op$ is placed at the 
configuration resulting from the execution of line \lref{ln:80} by $s_t$ for $op$.
Otherwise, the linearization point of $op$ is placed at the configuration resulting 
from the execution of line \lref{ln:lp2} by $s_t$ for $op$.
\end{compactitem}
Denote by $L$ the sequence of operations (which have been assigned linearization points)
in the order determined by their linearization points.

\begin{lemma}
The linearization point of a push (pop) operation $op$ is placed in its execution interval.
\end{lemma}

\begin{proof}[Proof Sketch]
Assume that $op$ is a push operation and let $c$ be the client that invokes it. 
After the invocation of $op$, $c$ sends a message to some server $s$ and awaits 
a response. Recall that routine {\tt receive()} (line \lref{ln:83}) blocks until 
a message is received. The linearization point of $op$ is placed either in the 
configuration resulting from the execution of line \lref{ln:lp1} by $s_t$ for $op$, 
where $s_t$ is the token server in this configuration, or in the configuration 
resulting from the execution of line \lref{ln:71} by $s_t$ for $op$.

Either of these lines is executed after the request by $c$ is received, i.e. after 
$c$ invokes {\tt ClientPush}. Furthermore, they are executed before $c$ receives 
the response by $s_t$ and thus, before {\tt ClientPush} returns. Therefore, the 
linearization point is inside the execution interval of push.

The argumentation regarding pop operations is analogous.
\end{proof}

\ignore{
Let $\alpha$ be any execution of the token-based stack algorithm. 
Each server participating in the algorithm maintains a local variable 
$lstack$. By inspection of the pseudocode (lines \lref{ln:lp3}, \lref{ln:75.1}), 
the semantics of routines {\tt push()} and {\tt pop()}, and the fact that 
messages that arrive at a server are served in a FIFO manner, we have the 
following.

\begin{observation}
\label{obs:lstack}
The local variable $lstack$ of a server $s$ that takes steps in $\alpha$ 
implements a sequential LIFO stack.
\end{observation}
}

Each server maintains a local variable $token$ with initial value $0$
(initially, the server with id equal to $0$ is the token server). 
Whenever some server $s_{i}$ receives a \TOKEN\ message, i.e. a message 
with its $tk$ field equal to \TOKEN\ (line \lref{ln:60}), 
the value of $token$ is set to $i$. By inspection of the pseudocode, 
it follows that the value of $token$ is set to 
the id of the next server if the local stack
of $s_{i}$ is full (line \lref{ln:68}); then, a \TOKEN\ message is sent 
to the next server (line \lref{ln:70}). Moreover, the value of $token$ 
is set to the id of the previous server if the local 
stack $lstack$  of $s_{i}$ is empty (line \lref{ln:76}); then, a \TOKEN\ 
message is sent to the previous server (lines \lref{ln:77}-\lref{ln:79}). 
(Unless the server is $s_0$ in which case a \NACK\ is sent to the client 
(line \lref{ln:80} but no \TOKEN\ message to any server.) 
Thus, the following observation holds.

\begin{observation}
\label{obs:unique}
At each configuration in $\alpha$, there is at most one server $s_{i}$ for 
which the local variable $token$ has the value $i$.
\end{observation}

At each configuration $C$, the server $s_i$ whose $token$ variable 
is equal to $i$ is referred to as the {\em token server} at $C$.

\begin{observation}
\label{obs:direction}
A \TOKEN\ message is sent from a server with id $i$, $0 \leq i < \maxser -1$, to a server 
with id $i+1$ only if the local stack of server $i$ is full. A \TOKEN\ message is sent from 
a server with id $i$, $0 < i \leq \maxser -1$, to a server with id $i-1$ only when the local 
stack of server $i$ is empty. 
\end{observation}

By the pseudocode, namely the \If\ clause of line \lref{ln:64} and 
the \If\ clause of line \lref{ln:73}, the following observation holds.

\begin{observation}
\label{obs:tokenserver}
Whenever a server $s_{i}$ performs push and pop operations on its local stack (lines \lref{ln:lp3} 
and \lref{ln:75.1}), it holds that its local variable $token$ is equal to $i$.
\end{observation}

Let $C_i$ be the configuration at which the $i$-th operation $op_i$ of $L$ is linearized. 
Denote by $\alpha_i$, the prefix of $\alpha$ which ends with $C_i$ and let $L_i$ be the 
prefix of $L$ up until the operation that is linearized at $C_i$. Denote by $S_i$ the 
sequence of values that a sequential stack contains after applying the sequence of operations 
in $L_i$, in order, starting from an empty stack; let $S_0 = \epsilon$, i.e. $S_0$ is the empty 
sequence.

\begin{lemma}
For each  $i$, $i \geq 0$, if $s_{k_i}$ is the token server at $C_i$ and $ls_i^j$ are the 
contents of the local stack of server $j$, $0 \leq j \leq k_i$, at $C_i$, then it holds that 
$S_i = ls_i^0 \cdot ls_i^1 \cdot \ldots \cdot ls_i^{k_i}$ at $C_i$.
\end{lemma}

\begin{proof}
We prove the claim by induction on $i$. 
The claim holds trivially for $i=0$.
Fix any $i \geq 0$ and assume that at $C_i$, it holds that 
$S_i$ = $ls_i^0 \cdot ls_i^1 \cdot \ldots \cdot ls_i^{k_i}$. 
We show that the claim holds for $i+1$.

We first assume that $op_{i+1}$ is a push operation initiated by some client $c$.
Assume first that $s_{k_i} = s_{k_{i+1}}$. Then, by induction hypothesis, 
$S_i$ = $ls_i^0 \cdot \ldots \cdot ls_i^{k_i}$. 
In case the local stack of $s_{k_i}$ is not full,  $s_{k_i}$ pushes the value $v_{i+1}$ 
of field $data$ of the request onto its local stack and responds to $c$. 
Since no other change occurs to the local stacks of $s_0, \ldots, s_{k_i}$ from $C_i$ to $C_{i+1}$, 
at $C_{i+1}$, it holds that 
$S_{i+1} = ls_i^0 \cdot \ldots \cdot ls_i^k \cdot \{v_{i+1}\} = ls_i^0 \cdot \ldots \cdot ls_{i+1}^{k_i}$. 
\ignore{In case that the local stack of $s_{k_i}$ is full, since $s_{k_i} = s_{k_{i+1}}$, by inspection 
of the pseudocode, it follows that $s_{k_i} = s_{\maxser - }$. In this case, $s_{k_i}$ responds 
with a \NACK\ to $c$ and the local stack remains unchanged. Thus, it holds that 
$S_{i+1} = ls_i^0 \cdot \ldots \cdot ls_i^k = S_i$. }
In case that the local stack of $s_{k_i}$ is full, since $s_{k_i} = s_{k_{i+1}}$ and it is 
the token server, it follows that $s_{k_i} = s_{\maxser - 1}$. In this case, $s_{k_i}$ responds 
with a \NACK\ to $c$ and the local stack remains unchanged. 
Thus, it holds that $S_{i+1} = ls_i^0 \cdot \ldots \cdot ls_i^k = S_i$.


Assume now that $s_{k_i} \neq s_{k_{i+1}}$. 
This implies that the local stack of $s_{k_i}$ is full just after $C_i$.
Observation \ref{obs:direction} implies that $s_{k_i}$ forwarded the token to $s_{k_i+1}$ 
in some configuration between $C_i$ and $C_{i+1}$. Notice that then, $s_{k_i+1} = s_{k_{i+1}}$.
Observation \ref{obs:tokenserver} implies that the local stack of $s_{k_i+1}$ is empty. 
Thus, the \If\ condition of line \lref{ln:66} evaluates to \true\ for server $s_{k_i+1}$ 
and therefore, it pushes the value $v_{i+1}$ of $op_{i+1}$ onto its local stack. 
Thus, at $C_{i+1}$, $ls_{i+1}^{k_i+1}$ = $\{v_{i+1}\}$. 
By definition, $S_{i+1}$ = $S_i \cdot \{v_{i+1}\}$. 
Therefore, $S_{i+1}$ = $ls_i^0 \cdot \ldots \cdot ls_{i+1}^{k_i+1}$. 
And since by Observations \ref{obs:unique} and \ref{obs:tokenserver}, 
the contents of the local stacks of servers other than $k_i+1$ do not change, 
it holds that $S_{i+1} = ls_{i+1}^0 \cdot \ldots \cdot ls_{i+1}^{k_i+1} = ls_{i+1}^0 \cdot \ldots \cdot ls_{i+1}^{k_{i+1}}$. 

\ignore{
Assume now that $s_k = $ be the last server. In this case, the \Else\ case of line \lref{ln:71} is 
executed, i.e. server $k$ sends \NACK\ to $c$ and does not modify its local stack nor is $op_{i+1}$ 
successful, thus $S_{i+1} = S_i$. Βy Observations \ref{obs:unique} and \ref{obs:tokenserver}, 
the contents of the local stacks of servers other than $k$ do not change between $C_i$ and $C_{i+1}$. 
Therefore, we can say that $S_{i+1} = ls_{i+1}^0 \cdot \ldots \cdot ls_{i+1}^k$ and the claim holds.
}

The reasoning for the case where $op_{i+1}$ is an instance of a pop operation is symmetrical.
\end{proof}

From the above lemmas and observations, we have the following.

\begin{theorem}
The token-based distributed stack implementation is linearizable. 
The time complexity and the communication complexity of each operation $op$ is $O(\maxser)$.
\end{theorem}

\subsection{Token-Based Queue}
\label{app:token-queue}

To implement a queue, two tokens are employed: 
at each point in time, there is a head token server $s_h$ and 
a tail token server $s_t$. Initially, server $0$ plays the role of 
both $s_h$ and $s_t$. 
Each server $s_r$, other than $s_t$ ($s_h$), that receives a request (directly) from a client $c$, 
it forwards the request to the next server to ensure that it will either reach the appropriate token server
or return back to $s_r$ (after traversing all servers). 
Servers $s_t$ and $s_h$ work in a way similar as server $s_t$ in stacks.

To prevent a request from being forwarded forever 
due to the completion of concurrent requests which may cause the token(s) to keep advancing, 
each server keeps track of the request that each client $c$ (directly) sends to it, 
in a {\em client table} (there can be only one such request per client). 
Server $s_t$ (and/or $s_h$) now reports the response to $s_r$ which forwards it
to $c$.  If $s_r$ receives a response for a request recorded in its client table, it deletes the request
from the client table. 
If $s_r$ receives the token (stack, tail, or head), it serves each request (push and pop, enqueue, or dequeue, respectively) 
in its client array and records its response. If a request, from those included in $s_r$'s client array,
reaches $s_r$ again, $s_r$ sends the response it has calculated for it to the client
and removes it from its client array. 
Since the communication channels are FIFO, the implementations ensures that all requests, 
their responses, and the appropriate tokens, move from one server to the next,
based on the servers' ring order, until they reach their destination. 
This is necessary to argue that the technique ensures termination for each request.

\subsubsection{Algorithm Description}
The queue implementation follows similar ideas to those of the 
token-based distributed stack, presented in Section \ref{app:stack-token}. However,
the queue implementation employs two tokens, one for the queue's tail and one for 
the queue's head, called {\em head token} and {\em tail token}, respectively. 
The tokens for the global head and tail are initially held by $s_0$.
However, they can be reassigned to other servers during the execution. If the local 
queue of the server that has the tail token becomes full, 
the token is forwarded to the next server. Similarly, if the local 
queue of the server that has the head token becomes empty, the head token is 
forwarded to the next server. If the appropriate token server receives the request
and serves it, it sends an \ACK\ message back to the server that initiated the forwarding. 
Then, the initial server responds to the client with an \ACK\ message, 
which also includes the id of the server that currently holds the token. 

The clients in their initial state store the id of $s_0$, which is the first server
to hold the head and tail tokens. The clients keep track of the servers that hold either token
in a lazy way. Specifically, a client updates its local variable (either $enq\_sid$ or $deq\_sid$
depending on whether its current active operation is an enqueue or a dequeue, respectively) 
with the id of the token server when it receives a server response. 

In this scheme, a client request may be transmitted indefinitely from a server to the next
without ever reaching the appropriate token server.
This occurs if both the head and the tail tokens are forwarded indefinitely
along the ring. Then, a continuous, never-ending race between a forwarded message and the 
appropriate token server may occur.
To avoid this scenario, we do the following actions. When a server $s$
receives a client's request $r$, if it does not have the appropriate token to serve it, 
it stores information about $r$ in a local array before it forwards it.
Next time that the server receives the tail (head) token, 
it will serve all enqueue (dequeue) requests. Notice that since channels preserve the FIFO order
and servers process messages in the order they arrive, the appropriate token will reach
$s$ earlier than the $r$. When $s$ receives $r$, it has already processed it; however,
it is then that $s$ sends the response for $r$ to the client.

\begin{algorithm}[!ht]
\small
\caption{Token-based queue server's local variables.}
\label{alg30-vars}
\begin{code}
  \lreset
  \firstline
     \integer\ $my\_sid$;                                                        \nl
     LocalQueue $lqueue = \varnothing$;                                          \nl
     LocalArray $clients = \varnothing$; 
                          \cm{Array of three values: $<$op, data, isServed$>$ }  \nl
     \bool\ $fullQueue$ = \false;  \ocm{True when tail and head are in the same} \ul
                                   \ccm{server and tail is before head }         \nl
     \bool\ $hasHead$;    \cm{Initially hasHead and hasTail are \true\ in server 0, 
     and false in the rest}                                                      \nl
     \bool\ $hasTail$; 
\end{code}
\end{algorithm}

\begin{algorithm}[!b]
\small
\caption{Events triggered in a server of the token-based queue.}
\label{alg30}
\begin{code}
  \firstline
     a message $\langle op, data, cid, sid, tk\rangle$ is received:           \nl
  \n   \If\ (!clients[$cid$] AND clients[$cid$].$isServed$) $\lbrace$
                                        \cm{If message has been served earlier.} \nl
  \n     send($cid, \langle\ACK,$ clients[$cid$].$data,\ my\_sid\rangle$);
                                                                \label{ln:a30-1} \nl
         clients[$cid$] = $\bot$;                               \label{ln:a30-2} \nl
  \p   $\rbrace$ \Else\ $\lbrace$                                                \nl
  \n     \Switch\ ($op$) $\lbrace$                              \label{ln:a30-3} \nl
  \n 	   \Case\ \ENQ:                                         \label{ln:a30-4} \nl
  \n         \If\ ($tk$ == \TAIL) $\lbrace$                    \label{ln:a30-36} \nl
  \n           $hasTail$ = \true;                              \label{ln:a30-37} \nl
               \If\ ($hasHead$)  $fullQueue$ = \true;          \label{ln:a30-38} \nl
               ServeOldEnqueues();                             \label{ln:a30-39} \ul
  \p         $\rbrace$                                                           \nl
             \If\ (!$hasTail$) $\lbrace$                        \label{ln:a30-9} 
                                                 \cm{Server does not have token} \nl
  \n           $nsid$ = find\_next\_server(($my\_sid$);        \label{ln:a30-41} \nl
               \If\ ($sid == -1$) $\lbrace$  \label{ln:a30-11} \cm{From client.} \nl
  \n             clients[$cid$] = $\langle\ENQ, data, \false\rangle$;
                                                               \label{ln:a30-12} \nl
                 send($nsid,\langle\ENQ, data, cid, my\_sid, \bot\rangle$);
                                                               \label{ln:a30-13} \nl
  \p           $\rbrace$ \Else\ $\lbrace$                      \cm{From server.} \nl
  \n             send($nsid,\langle\ENQ, data, cid, sid, \bot\rangle$); \label{ln:a30-14} \ul
  \p           $\rbrace$                                                         \nl
  \p         $\rbrace$ \Elseif\ (!IsFull($lqueue$)) $\lbrace$   \label{ln:a30-6}
                                                        \cm{Server can enqueue.} \nl
  \n          enqueue($lqueue, data$);                  \label{ln:a30-61}        \nl
              \If\ ($sid == -1$)                               \cm{From client.} \nl
  \n            send($cid, \langle\ACK, \bot, my\_sid\rangle$);\label{ln:a30-62} \nl
  \p          \Else                                            \cm{From server.} \nl
  \n            send($sid, \langle\ACK, \bot, cid, my\_sid, \bot\rangle$); \label{ln:a30-5.1}\nl
  \p\p       $\rbrace$ \Elseif\ ($fullQueue$) $\lbrace$         \label{ln:a30-5}
                                                          \cm{Global Queue full} \nl
  \n           \If\ ($sid$ == -1)                              \cm{From client.} \nl
  \n             send($cid, \langle\NACK, \bot, my\_sid\rangle$);\label{ln:a30-51}\nl
  \p           \Else\                                           \cm{From server} \nl
  \n             send($sid, \langle\NACK, \bot, cid, my\_sid, \bot\rangle$);\label{ln:a30-7.1}\nl
  \p\p       $\rbrace$ \Else\ $\lbrace$                         \label{ln:a30-7} 
                             \cm{Server moves the tail token to the next server} \nl                   
  \n           $nsid$ = find\_next\_server($my\_sid$);          \label{ln:a30-8} \nl 
                $fullQueue$ = \false;                                             \nl
                 $hasTail$ = \false;                            \label{ln:a30-10.1}  \nl
               send($nsid, \langle op, data, cid, my\_sid, \TAIL\rangle$);
                                                               \label{ln:a30-10} \ul
  \p        $\rbrace$                                                            \nl
            \Break;
\end{code}
\end{algorithm}

\begin{algorithm}[!t]
\small
\begin{code}
  \firstline
  \p~~~~~~~~~~\Case\ \DEQ:                                       \label{ln:a30-15} \nl
  \n        \If\ ($tk$ == \HEAD) $\lbrace$                       \label{ln:a30-33} \nl
  \n          $hasHead$ = \true;                                 \label{ln:a30-34} \nl
              ServeOldDequeues();                                \label{ln:a30-35} \ul
  \p        $\rbrace$                                                              \nl
            \If\ (!$hasHead$) $\lbrace$                          \label{ln:a30-23} \nl
  \n          $nsid$ = find\_next\_server($my\_sid$);            \label{ln:a30-19} \nl
              \If\ ($sid == -1$) $\lbrace$                        \cm{From client} \nl
  \n            clients[$cid$] = $\langle\DEQ, \bot, \false\rangle$;\label{ln:a30-24}\nl
                send($nsid, \langle\DEQ, \bot, cid, my\_sid\rangle$);\label{ln:a30-25}\nl
  \p          $\rbrace$ \Else\ $\lbrace$                          \cm{From server} \nl
  \n            send($nsid, \langle\DEQ, \bot, cid, sid, \bot\rangle$);\label{ln:a30-26} \ul
  \p          $\rbrace$                                                            \nl
  \p        $\rbrace$ \Elseif\ (!IsEmpty($lqueue$)) $\lbrace$
                                        \label{ln:a30-17} \cm{Server can dequeue.} \nl
  \n          $data$ = dequeue($lqueue$);                         \label{ln:a30-26.1} \nl
              \If\ ($sid == -1$)                                  \cm{From client} \nl
  \n            send($cid, \langle\ACK, data, my\_sid\rangle$); \label{ln:a30-16.1}  \nl
  \p          \Else\                                              \cm{From server} \nl
  \n            send($sid, \langle\ACK, data, cid, my\_sid, \bot\rangle$);         \nl
  \p\p      $\rbrace$ \Elseif\ ($hasTail$ AND !$fullQueue$) $\lbrace$ 
                                             \label{ln:a30-16} \cm{Queue is empty} \nl
  \n          \If\ ($sid == -1$)                                  \cm{From client} \nl
  \n            send($cid, \langle\NACK, \bot, my\_sid\rangle$);                   \nl
  \p          \Else \                                             \cm{From server} \nl
  \n            send($sid, \langle\NACK, \bot, cid, my\_sid, \bot\rangle$ );       \nl
  \p\p      $\rbrace$ \Else\ $\lbrace$                           \label{ln:a30-18}   
                               \cm{Server moves the head token to the next server} \nl
  \n          $nsid$ = find\_next\_server($my\_sid$);            \label{ln:a30-20} \nl
              $hasHead$ = \false;                                \label{ln:a30-22} \nl
              send($nsid, \langle op, \bot, cid, my\_sid, \HEAD\rangle$);\label{ln:a30-21}\ul
  \p        $\rbrace$                                                              \nl
            \Break;                                                                \nl
  \p      \Case\ \ACK:                                           \label{ln:a30-27} \nl
  \n        clients[$cid$] = $\bot$;                             \label{ln:a30-28} \nl
            send($cid, \langle\ACK, data, sid\rangle$);          \label{ln:a30-29} \nl
            \Break;                                                                \nl
  \p      \Case\ \NACK:                                          \label{ln:a30-30} \nl
  \n        clients[$cid$] = $\bot$;                             \label{ln:a30-31} \nl
            send($cid, \langle\NACK, \bot, sid\rangle$);         \label{ln:a30-32} \nl
            \Break;                                                                \ul
  \p\p  $\rbrace$                                                                  \ul
  \p  $\rbrace$
  \p 
\end{code}
\end{algorithm}

In Algorithm \ref{alg30-vars}, we present the local variables of a server. 
Each server $s$ holds its unique id $my\_sid$ and a local queue $lqueue$ that stores its 
part of the queue. Also keeps two \bool\ flag variables, ($hasHead$ and $hasTail$), 
indicating whether $s$ has the head token or the tail token, and one more bit flag 
($fullQueue$) indicating whether the queue is full. Finally, $s$ has a local array of 
size $n$, where $n$ is the maximum number of clients, used for storing all direct requests 
from clients (called $s$'s $clients$ array). In their initial state, all servers have 
$fullQueue$ set to \false\ and their $clients$ array and local queue empty. Also, all 
servers apart from server 0, have both their flags $hasHead$ and $hasTail$ set to \false, 
whereas in server 0, they are set to \true, as described above. 

The messages sent to a server $s_i$ are of type \ENQ\ or \DEQ, describing requests 
for enqueue or dequeue operations, respectively, sent by either a server or a client, and \ACK\ 
or \NACK\ sent by another server $s_j$ which executed a forwarded request, whose forwarding
was initiated by $s_i$. The token transition is encapsulated in a message of type 
\ENQ\ or \DEQ. The messages have five fields: (1) $op$, which describes the 
type of the request (\ENQ, \DEQ, \ACK\ or \NACK) , (2) $data$, which 
stores an element in case of \ENQ, and $\bot$ otherwise, (3) $cid$, which 
stores the id of the client that issued the request, (4) $sid$ which contains 
the id of the server if the message was sent by a server, and -1 otherwise, and (5)
$tk$, which contains \TAIL\ or \HEAD\ in forwarded messages of type \ENQ\ 
or \DEQ, respectively, to indicate if an additional tail (or head, respectively) 
token transition occurs, and it is equal to $\bot$ otherwise. 

Event-driven pseudocode for the server is presented in Algorithm \ref{alg30}. When 
a server $s$ receives a message of type \ENQ\ (line \lref{ln:a30-4}), it first checks 
if it contains a token transfer from another server (line \lref{ln:a30-36}). If it does, 
the server sets its token $hasTail$ to \true\ (line \lref{ln:a30-37}) and if 
it also had the head token from a previous round, it changes 
$fullQueue$ flag to \true\ 
as well (line \lref{ln:a30-38}). Then, $s$ serves all pending \ENQ\ messages stored in 
its $clients$ array (line \lref{ln:a30-39}). 

Then, $s$ continues to execute the \ENQ\ request. 
It checks first whether it has the tail token. If it does not (line 
\lref{ln:a30-9}), it finds the next server $s_{next}$ (line 
\lref{ln:a30-41}), to whom $s$ is going to forward the request. Afterwards, $s$ sends 
the received request to $s_{next}$ (lines \lref{ln:a30-13} and \lref{ln:a30-14}) and if 
that request came directly from a client (line \lref{ln:a30-11}), $s$ updates its 
$clients$  array storing in it information about this message 
(line \lref{ln:a30-12}).

If $s$ has the token, then it attempts to serve the request. If $s$ has remaining space 
in its local queue ($lqueue$), it enqueues the given data and informs the 
appropriate server with an \ACK\ message (lines \lref{ln:a30-6}-\lref{ln:a30-5.1}). 
If the implemented queue is full, $s$ sends a \NACK\ message to the client (line 
\lref{ln:a30-5}-\lref{ln:a30-7.1}). In any remaining case, the server $s$ must give the 
tail token to the next server (line \lref{ln:a30-7}). So, $s$ forwards the \ENQ\ message 
to $s_{next}$, after encapsulating in the message the tail token (line \lref{ln:a30-10}). 
After releasing the tail token, $s$ changes the values of its local variables 
($hasTail$ and $fullQueue$) to \false.

In case a \DEQ\ message is received (line \lref{ln:a30-15}), the actions performed 
by $s$ are similar to those for \ENQ. Server $s$ checks whether the
request message contains a token transition (line \lref{ln:a30-33}). 
If it does, the server sets its token $hasHead$ to \true\ (line \lref{ln:a30-34}). 
Then, $s$ serves all pending \DEQ\ messages stored in its $clients$ array (line 
\lref{ln:a30-35}),
and then attempts to serve the request. If $s$ does not 
hold the head token (line \lref{ln:a30-23}), it finds the next server $s_{next}$ in 
the ring (line \lref{ln:a30-19}), to whom $s$ is going to forward the request. 
Afterwards, $s$ sends the received request to $s_{next}$ (lines \lref{ln:a30-25}, 
\lref{ln:a30-26}) and if that request came directly from a client, 
$s$ updates its $clients$ array storing in it information about this message (line 
\lref{ln:a30-24}). 

If $s$ has the head token, it does the following actions. If its local queue ($lqueue$) 
is not empty (line \lref{ln:a30-17}), 
$s$ performs a dequeue on its local queue and sends an \ACK\ along with the dequeued 
data to the appropriate server. If $s$ holds both head and tail tokens, but no other server has 
a queue element stored, and $s$'s $lqueue$ is empty, it means that the global 
queue is empty (line \lref{ln:a30-16}). Thus, $s$ sends a \NACK\ message to the appropriate server. 
In the remaining cases, $s$ must forward its head token (line \lref{ln:a30-18}). 
Server $s$ finds the next server $s_{next}$ (line \lref{ln:a30-20}), which is going 
to receive the forwarded message and the head token transition. Server $s$ sets the message 
field $tk$ to \HEAD\ and sends the message (line \lref{ln:a30-21}). After 
releasing the head token, $s$ sets the value of its local variable $hasHead$ to 
\false\ (line \lref{ln:a30-22}). 

	\begin{minipage}{.5\textwidth}
		\begin{algorithm}[H]
\small
\caption{Auxiliary functions for a server of the token-based queue.}
\label{alg30-help}
\begin{code}
\firstline
     \void\ ServeOldEnqueues(\void) $\lbrace$                  \label{ln:a30h-1} \nl
  \n   \If\ (!$fullQueue$) $\lbrace$                           \label{ln:a30h-2} \nl
  \n     \Foreach\ $cid$ {\sf such that} clients[$cid$].$op == \ENQ\ \lbrace$    \nl
  \n       \If\ (!IsFull($lqueue$)) $\lbrace$                  \label{ln:a30h-3} \nl
  \n         enqueue($lqueue$, clients[$cid$].$data$);         \label{ln:a30h-31}\nl
             clients[$cid$].$isServed$ = \true;                \label{ln:a30h-32}\nl
  \p\p\p\p $\rbrace$ $\rbrace$ $\rbrace$ $\rbrace$                               \ul
                                                                                 \nl
     \void\ ServeOldDequeues(\void) $\lbrace$                  \label{ln:a30h-4} \nl
  \n   \Foreach\ $cid$ {\sf such that} clients[$cid$].$op == \DEQ\ \lbrace$      \nl
  \n     \If\ (!IsEmpty($lqueue$)) $\lbrace$                   \label{ln:a30h-5} \nl
  \n       clients[$cid$].$data$ = dequeue($lqueue$);          \label{ln:a30h-51}\nl
           clients[$cid$].$isServed$ = \true;                                    \ul
  \p\p\p $\rbrace$ $\rbrace$ $\rbrace$
\end{code}
\end{algorithm}
\end{minipage}

If $s$ received a message of type \ACK\ (line \lref{ln:a30-27}) or \NACK\ (line 
\lref{ln:a30-30}), then $s$ sets the entry $cid$ of its $clients$ array 
to $\bot$ (lines \lref{ln:a30-28}, \lref{ln:a30-31}) and sends an \ACK\ (line 
\lref{ln:a30-29}) or a \NACK\ (line \lref{ln:a30-32}) to that client. The \ACK\ and 
\NACK\ messages a server $s$ receives, are only sent by other servers and signify 
the result of the execution of a forwarded message sent by $s$.

On lines \lref{ln:a30-12} and \lref{ln:a30-24}, $s$ stores the client 
request in its $clients$ array when it does not hold the appropriate token. 
A request recorded in
the $clients$ array is removed from the array either when an \ACK\ or \NACK\ message 
is received for it (lines \lref{ln:a30-28} and \lref{ln:a30-31}) or when the server 
receives again the request (after a round-trip on the ring) (lines~\lref{ln:a30-1}-\lref{ln:a30-2}). 
Thus, the server, upon any message receipt, first checks 
whether the message exists in its $clients$ array and has already been served. 
In case of \ENQ\, $s$ answers with an \ACK\ message, whereas in case of \DEQ\ $s$ answers
with \ACK\ and the dequeued data. Then, server $s$ proceeds with the deletion of the entry 
in its clients array (lines \lref{ln:a30-1}, \lref{ln:a30-2}).

Functions \texttt{ServeOldEnqueues()} and \texttt{ServeOldDequeues()} are described 
in more detail in Algorithm \ref{alg30-help}. \texttt{ServeOldEnqueues()} (line 
\lref{ln:a30h-1}) processes all \ENQ\ requests stored in the $clients$ array, 
if the local queue has space (line \lref{ln:a30h-3}). Similarly, 
\texttt{ServeOldDequeues()} (line \lref{ln:a30h-4}) processes all \DEQ\ requests 
stored in the $clients$ array, if the local queue is not empty (line 
\lref{ln:a30h-5}).

The clients call the functions \texttt{ClientEnqueue()} and \texttt{ClientDequeue()}, 
presented in Algorithm \ref{alg31}, in order to perform one of these operations. In more 
detail, during enqueue, the client sends an \ENQ\ message to the $enq\_sid$ server, 
and waits for a response. When the client receives the response, it returns it. Likewise, 
in \texttt{ClientDequeue()} the client sends a \DEQ\ message to server $deq\_sid$ 
and blocks waiting for a response. When it receives the response, it returns it.

	\begin{minipage}{.45\textwidth}
		\begin{algorithm}[H]
\small
\caption{Enqueue and Dequeue operations for a client of the token-based queue.}
\label{alg31}
\begin{code}
	\firstline
	\integer\ $enq\_sid = 0$;                                            \nl
	\integer\ $deq\_sid = 0$;                                            \ul
	                                                                     \nl
	   Data ClientEnqueue(\integer\ $cid$, Data $data$) $\lbrace$        \nl
	\n   send($enq\_sid, \langle \ENQ, data, cid, -1\rangle $);          \label{ln:cl-enq:01}\nl
         $\langle status, \bot, enq\_sid\rangle$ = receive($enq\_sid$);  \label{ln:cl-enq:02}\nl
	     \return\ $status$;                                              \label{ln:cl-enq:03}\ul
	\p $\rbrace$                                                         \ul
	                                                                     \nl 
	   Data ClientDequeue(\integer\ $cid$) $\lbrace$                     \nl
	\n   send($deq\_sid, \langle \DEQ, \bot, cid\rangle$);               \label{ln:cl-deq:01}\nl
	     $\langle status, data, deq\_sid\rangle$ = receive($deq\_sid$);  \label{ln:cl-deq:02}\nl
	     \return\ $data$;                                                \label{ln:cl-deq:03}\ul
	\p $\rbrace$
\end{code}
\end{algorithm}
\end{minipage}

%
\subsubsection{Proof of Correctness}
Let $\alpha$ be an execution of the token-based queue algorithm presented 
in Algorithms \ref{alg30}, \ref{alg30-help}, and \ref{alg31}. 
Each server maintains local boolean variables $hasHead$ and $hasTail$, 
with initial values \false. Whenever some server $s_{i}$ receives a \TAIL\ 
message, i.e. a message with its $tk$ field equal to \TAIL\ (line \lref{ln:a30-36}), 
the value of $hasTail$ is set to \true\ (line \lref{ln:a30-37}). By inspection 
of the pseudocode, it follows that the value of $hasTail$ is set to \false\ 
if the local queue of $s_{i}$ is full (line \lref{ln:a30-6}, \lref{ln:a30-7}-
\lref{ln:a30-10.1}); then, a \TAIL\ message is sent to the next server (line 
\lref{ln:a30-10}). The same holds for $hasHead$ and \HEAD\ messages, i.e. 
messages with their $tk$ field equal to \HEAD. Thus, the following observations 
holds.

\begin{observation}
\label{obs:tq:unique}
At each configuration in $\alpha$, there is at most one server 
for which the local variable $hasHead$ ($hasTail$) has the value \true.
\end{observation}

\begin{observation}
\label{obs:tq:direction}
In some configuration $C$ of $\alpha$,  \TAIL\ message is sent from a server $s_j$, $0 \leq j < \maxser -1$, to a server 
$s_k$, where $k = (j+1) \mod{\maxser}$ only if the local queue of $s_j$ is full in $C$. Similary, a 
\HEAD\ message is sent from $s_j$ to $s_k$ only if the local queue of $s_j$ is empty in $C$. 
\end{observation}

By inspection of the pseudocode, we see that a server performs an enqueue (dequeue) 
operation on its local queue $lqueue$ either when executing line \lref{ln:a30-61} 
(line \lref{ln:a30-35}) or when executing {\tt ServeOldEnqueues} ({\tt ServeOldDequeues}). 
Further inspection of the pseudocode (lines \lref{ln:a30-36}-\lref{ln:a30-39}, lines 
\lref{ln:a30-6}-\lref{ln:a30-5}, as well as lines \lref{ln:a30-23}-\lref{ln:a30-26}, 
lines \lref{ln:a30-17}-\lref{ln:a30-16}), shows that these lines are executed when 
$hasTail = \true$. Then, the following observation holds.

\begin{observation}
\label{obs:tq:tokenserver}
Whenever a server $s_j$ performs an enqueue (dequeue) operation on its local 
queue, it holds that its local variable $hasTail$ ($hasHead$) is equal to \true. 
\end{observation}

By a straight-forward induction, the following lemma can be shown. 

\begin{lemma}
\label{tq:client-msgs}
The mailbox of a client in any configuration of $\alpha$ contains at most one 
incoming message.
\end{lemma}

If $hasTail = \true$ ($hasHead = \true$) for some server $s$ in some configuration $C$, 
then we say that $s$ has the tail (head) token. The server that has the tail token is 
referred to as {\em tail token server}. The server that has the head token is referred 
to as {\em head token server}.

Let $op$ be any operation in $\alpha$. We assign a linearization point to $op$ 
by considering the following cases: 
\begin{compactitem}
\item If  $op$ is an enqueue operation for which a tail token server executes an instance  
of Algorithm \ref{alg30}, then
it is linearized in the configuration resulting from the execution of either 
line \lref{ln:a30-61}, or line \lref{ln:a30h-31}, or line \lref{ln:a30-51}, whichever 
is executed for $op$ in that instance of Algorithm \ref{alg30} by the tail token server. 
\item If $op$ is a dequeue operation for which a head token server executes an instance 
of Algorithm \ref{alg30}, then
it is linearized in the configuration resulting from the execution of either 
line \lref{ln:a30-26.1}, or line \lref{ln:a30h-51}, or line \lref{ln:a30-16.1}, whichever 
is executed for $op$ in that instance of Algorithm \ref{alg30} by the head token server.
\end{compactitem}

\begin{lemma}
The linearization point of an enqueue (dequeue) operation $op$ is placed in its execution interval.
\end{lemma}

\begin{proof}
Assume that $op$ is an enqueue operation and let $c$ be the client that invokes it. 
After the invocation of $op$, $c$ sends a message to some server $s$ (line \lref{ln:cl-enq:01}) 
and awaits a response. Recall that routine {\tt receive()} (line \lref{ln:cl-enq:02}) 
blocks until a message is received. The linearization point of $op$ is placed either in the 
configuration resulting from the execution of line \lref{ln:a30-61} by  $s_t$ for $op$, 
in the configuration resulting from the execution of line \lref{ln:a30-51} by $s_t$ 
for $op$, or in the configuration resulting from the execution of line \lref{ln:a30h-31} by 
$s_t$ for $op$. Notice that either of these lines is executed after the request by $c$ is 
received, i.e. after $c$ invokes {\tt ClientEnqueue}, and thus, after the execution interval 
of $op$ starts. 

By definition, the execution interval of $op$ terminates in the configuration resulting 
from the execution of line \lref{ln:cl-enq:03}. By inspection of the pseudocode, this line 
is executed after line \lref{ln:cl-enq:02}, i.e. after $c$ receives a response by some server. 
In the following, we show that the linearization point of $op$ occurs before this response 
is sent to $c$. 


Let $s_j$ be the server that $c$ initially sends the request for $op$ to. 
By observation of the pseudocode, we see that $c$ may either receive a response 
from $s_j$ if $s_j$ executes lines \lref{ln:a30-62} or \lref{ln:a30-51}, 
or if $s_j$ executes lines \lref{ln:a30-28}-\lref{ln:a30-29} or lines 
\lref{ln:a30-31}-\lref{ln:a30-32}, or if $s_j$ executes line \lref{ln:a30-1}.
To arrive at a contradiction, assume that either of these lines is executed in $\alpha$ 
before the configuration in which the linearization point of $op$ is placed. Thus, a tail 
token server $s_t$ executes lines \lref{ln:a30-61}, \lref{ln:a30h-31}, or \lref{ln:a30-51} 
in a configuration following the execution of lines \lref{ln:a30-62}, or \lref{ln:a30-51}, 
or \lref{ln:a30-28}-\lref{ln:a30-29} or \lref{ln:a30-31}-\lref{ln:a30-32}, or line \lref{ln:a30-1}
by $s_j$. Since the algorithm is event-driven, inspection of the pseudocode shows that in 
order for a tail token server to execute these lines, it must receive a message containing 
he request for $op$ either from a client or from another server. 

Assume first that a tail token server executes the algorithm after receiving a message 
containing a request for $op$ from a client. This is a contradiction, since, on one hand, 
$c$ blocks until receiving a response, and thus, does not sent further messages requesting 
$op$ or any other operation, and since $op$ terminates after $c$ receives the response by 
$s_j$, and on the other hand, any other request from any other client concerns 
a different operation $op'$.

Assume next that a tail token server executes the algorithm after receiving a message 
containing the request for $op$ from some other server. This is also a contradiction since 
inspection of the pseudocode shows that after $s_j$ executes either of the lines that sends 
a response to $c$, it sends no further message to some other server and instead, 
terminates the execution of that instance of the algorithm.

The argumentation regarding dequeue operations is analogous.
\end{proof}

Denote by $L$ the sequence of operations which have been assigned linearization 
points in $\alpha$ in the order determined by their linearization points.
Let $C_i$ be the configuration at which the $i$-th operation $op_i$ of $L$ is 
linearized. Denote by $\alpha_i$, the prefix of $\alpha$ which ends with $C_i$ 
and let $L_i$ be the prefix of $L$ up until the operation that is linearized 
at $C_i$. Denote by $Q_i$ the sequence of values that a sequential queue contains 
after applying the sequence of operations in $L_i$, in order, starting from an 
empty queue; let $Q_0 = \epsilon$, i.e. $Q_0$ is the empty sequence.
In the following, we denote by $s_{t_i}$ the tail token server at $C_i$ 
and by $s_{h_i}$ the head token server at $C_i$.

\begin{lemma}
For each  $i$, $i \geq 0$, if 
$lq_i^j$ are the contents of the local queue of server $s_j$ at 
$C_i$, $h_i \leq j \leq t_i$, at $C_i$, then it holds that 
$Q_i = lq_i^{h_i} \cdot lq_i^{{h_i}+1} \cdot \ldots \cdot lq_i^{t_i}$ at $C_i$.
\end{lemma}

\begin{proof}
We prove the claim by induction on $i$.
The claim holds trivially at $i = 0$.

Fix any $i \geq 0$ and assume that at $C_i$, it holds that 
$Q_i$ = $lq_i^{h_i} \cdot lq_i^{{h_i}+1} \cdot \ldots \cdot lq_i^{t_i}$. 
We show that the claim holds for $i+1$.

First, assume that $op_{i+1}$ is an enqueue operation by client $c$. 
Furthermore, distinguish the following two cases: 
\begin{itemize}
\item Assume that $t_i = t_{i+1}$. Then, by the induction hypothesis, 
$Q_i = lq_i^{h_i} \cdot lq_i^{{h_i}+1} \cdot \ldots \cdot lq_i^{t_i}$. In case the local 
queue of $s_{t_i}$ is not full, $s_{t_i}$ enqueues the value $v_{i+1}$ of the 
$data$ field of the request for $op_{i+1}$ in the local queue (line \lref{ln:a30-61} 
or line \lref{ln:a30h-31}). Notice that, by Observation \ref{obs:tq:tokenserver} changes on the local queues of servers occur only 
on token servers. Notice also that those changes occur only in a step that immediately 
precedes a configuration in which an operation is linearized. Thus, no further 
change occurs on the local queues of $s_{h_i}, s_{{h_i}+1}, \ldots, s_{t_i}$ 
between $C_i$ and $C_{i+1}$, other than the enqueue on $lq_i^t$. Then, it holds 
that $Q_{i+1} = Q_i \cdot v_{i+1} = lq_i^{h_i} \cdot lq_i^{{h_i}+1} \cdot \ldots \cdot lq_i^{t_i} \cdot v_{i+1} = 
lq_i^{h_i} \cdot lq_i^{{h_i}+1} \cdot \ldots \cdot lq_{i+1}^{t_i} = 
lq_{i+1}^{h_i} \cdot lq_{i+1}^{{h_i}+1} \cdot \ldots \cdot lq_{i+1}^{t_i}$, and
if the head token server does not change between $C_i$ and $C_{i+1}$, then $h_{i+1} = h_i$ and 
$Q_{i+1} = lq_{i+1}^{h_{i+1}} \cdot lq_{i+1}^{{h_{i+1}}+1} \cdot \ldots \cdot lq_{i+1}^{t_{i+1}}$ 
and the claim holds. If the head token server changes, i.e., if $h_{i+1} \neq h_i$, then by 
Observation \ref{obs:tq:direction}, $lq_{i+1}^{h_i} = \emptyset$ and the claim holds again.

In case the local queue of $s_{t_i}$ is full and since by assumption, $s_{t_i} = s_{t_{i+1}}$, 
it follows by inspection of the pseudocode (line \lref{ln:a30-5}) and the definition of linearization 
points, that $s_{t_{i+1}} = s_{h_{i+1}}$. In this case, $s_{t_{i+1}}$ responds with a \NACK\ 
to $c$ and the local queue remains unchanged. Since no token server changes between $C_i$ and $C_{i+1}$, 
$Q_{i+1} = Q_i = lq_i^{h_i} \cdot lq_i^{{h_i}+1} \cdot \ldots \cdot lq_i^{t_i} = 
lq_{i+1}^{h_{i+1}} \cdot lq_{i+1}^{{h_{i+1}}+1} \cdot \ldots \cdot lq_{i+1}^{t_{i+1}}$ 
and the claim holds.

\item Next, assume that $t_i \neq t_{i+1}$. This implies that the local queue of $s_{t_i}$ 
is full just after $C_i$. Observation \ref{obs:tq:direction} implies that $s_{t_i}$ 
forwarded the token to $s_{t_i+1}$ in some configuration between $C_i$ and $C_{i+1}$. 
Notice that then, $s_{t_i+1} = s_{t_{i+1}}$. If the local queue of $s_{t_{i+1}}$ is 
not full, then the condition of line \lref{ln:a30-6} evaluates to \true\ and therefore, line 
\lref{ln:a30-61} is executed, enqueueing value $v_{i+1}$ to it. Then at $C_{i+1}$, $lq_{i+1}^{t_{i+1}} = v_{i+1}$. 
By definition, $Q_{i+1} = Q_i \cdot v_{i+1}$, and therefore, 
$Q_{i+1} = lq_i^{h_i} \cdot lq_i^{{h_i}+1} \cdot \ldots \cdot lq_i^{t_i} \cdot v_{i+1} = 
lq_{i+1}^{h_{i+1}} \cdot lq_{i+1}^{{h_{i+1}}+1} \cdot \ldots \cdot lq_{i+1}^{t_i} \cdot v_{i+1} = 
lq_{i+1}^{h_{i+1}} \cdot lq_{i+1}^{{h_{i+1}}+1} \cdot \ldots \cdot lq_{i+1}^{t_i} \cdot lq_{i+1}^{t_{i+1}}$ 
and the claim holds.  If the local queue of $s_{t_{i+1}}$ is full, then the condition of 
line \lref{ln:a30-6} evaluates to \false\ and therefore, line \lref{ln:a30-7.1} is executed. The operation 
is linearized in the resulting configuration and \NACK\ is sent to $c$. Notice that in that 
case, the local queue of the server is not updated. Then, 
$Q_{i+1} = Q_i = lq_i^{h_i} \cdot lq_i^{{h_i}+1} \cdot \ldots \cdot lq_i^{t_i} \cdot lq_{i+1}^{t_{i+1}} = 
lq_{i+1}^{h_{i+1}} \cdot lq_{i+1}^{{h_{i+1}}+1} \cdot \ldots \cdot lq_{i+1}^{t_i} \cdot lq_{i+1}^{t_{i+1}}$, 
and the claim holds. 
\end{itemize}
The reasoning for the case where $op_{i+1}$ is an instance of a dequeue operation is symmetrical.
\end{proof}


%

From the above lemmas and observations we have the following theorem.

\begin{theorem}
The token-based distributed queue implementation is linearizable. 
The time complexity and the communication complexity of each operation $op$ is $O(\maxser)$.
\end{theorem}

\subsection{Token-Based Double Ended Queue (Deque) }
\label{app:token-deque}

The dequeue implementation is a natural generalization of the stack 
and queue implementations described previously. The deque implementation 
is analogous to the queue implementation described in 
section~\ref{app:token-queue}. To provide
a deque, we add actions to the queue's design to support the additional operations
supported by a deque. We 
retain the static ordering of the servers and the head and tail tokens. Again, each server uses 
a local data structure, this time a deque, on which the 
server is allowed to execute enqueues or dequeues to the appropriate end, only if it has either 
the {\em tail token} or the {\em head token}. The head and tail tokens are initially held 
by $s_0$, but can be reassigned to other servers during the execution.

\subsubsection{Algorithm Description}
Algorithm \ref{alg32} presents the events triggered in a server $s$
and $s$'s actions for each event. Each server, in addition to its id ($my\_sid$), 
maintains a local deque ($ldeque$) to store elements of the implemented deque. For the token 
management, each server has two \bool\ flags ($hasHead$ and $hasTail$), which are 
initialized \true\ for the server $s_0$, and \false\ for the rest. Furthermore, the servers maintain 
a $fullDeque$ flag, similar to the $fullQueue$ flag in Algorithm \ref{alg30-vars} 
of Section \ref{app:token-queue}.
Finally, the servers maintain a local array (called $s$'s $clients$ array) 
for storing the requests they receive directly from clients, 
which is used in a similar way as in Section \ref{app:token-queue}. 

\begin{algorithm}[!t]
\small
\caption{Events triggered in a server}
\label{alg32}
\begin{code}
  \lreset
\firstline
   \integer\ $my\_sid$;                                                              \nl
   LocalDeque $ldeque = \varnothing$;                                                \nl
   LocalArray $clients = \varnothing$;  
                                     \cm{Array of three values <op, data, isServed>} \nl
   \bool\ $fullDeque$ = \false;  \ul
          \cm{True when tail and head are in the same server and tail is before head}\nl
   \bool\ $hasHead$;  \cm{Initially hasHead and hasTail are \true\ in server 0, and 
   false in the rest}                                                      \nl
   \bool\ $hasTail$;                                                        \ul
                                                                                 \nl
   a message $\langle op, data, cid, sid, tk\rangle$ is received:\label{ln:a32-18}\nl
   \n \If\ (clients[$cid$] $\neq \bot$ AND clients[$cid$].$isServed$) $\lbrace$
                               \label{ln:a32-19} \cm{If request was served earlier.} \nl
   \n   send($cid$, $\langle\ACK,$ clients[$cid$].$data, my\_sid\rangle$);
                                                                    \label{ln:a32-1} \nl
        clients[$cid$] = $\bot$;                                    \label{ln:a32-2} \nl
   \p $\rbrace$ \Else\ $\lbrace$                                                     \nl
   \n   \Switch\ ($op$) $\lbrace$                                   \label{ln:a32-9} \nl
   \n     \Case\ \ENQT:                                            \label{ln:a32-10} \nl
   \n       ServerEnqueueTail($op, data, cid, sid, tk$);           \label{ln:a32-11} \nl
            \Break;                                                                  \nl
   \p     \Case\ \DEQT:                                            \label{ln:a32-12} \nl
   \n       ServerDequeueTail($op, cid, sid, tk$);                 \label{ln:a32-13} \nl
            \Break;                                                                  \nl
   \p     \Case\ \ENQH:                                            \label{ln:a32-14} \nl
   \n       ServerEnqueueHead($op, data, cid, sid, tk$);           \label{ln:a32-15} \nl
            \Break;                                                                  \nl
   \p     \Case\ \DEQH:                                            \label{ln:a32-16} \nl
   \n       ServerDequeueHead($op, cid, sid, tk$);                 \label{ln:a32-17} \nl
            \Break;                                                                  \nl
   \p     \Case\ \ACK:                                              \label{ln:a32-3} \nl
   \n       clients[$cid$] = $\bot$;                                \label{ln:a32-4} \nl
            send($cid, \langle \ACK, data, sid \rangle$);           \label{ln:a32-5} \nl
            \Break;                                                                  \nl
   \p     \Case\ \NACK:                                             \label{ln:a32-6} \nl
   \n       clients[$cid$] = $\bot$;                                \label{ln:a32-7} \nl
            send($cid, \langle \NACK, \bot, sid\rangle$);           \label{ln:a32-8} \nl
            \Break;                                                                  \ul
   \p\p $\rbrace$                                                                    \ul
   \p $\rbrace$
   \p
\end{code}
\end{algorithm}

The types of messages a server $s$ can receive are \ENQT, \DEQT\ for enqueuing at and dequeuing 
from the tail, \ENQH, \DEQH\ for enqueuing at and dequeuing from the head, and \ACK\ or \NACK\, 
sent by other servers as responses to $s$'s forwarded messages. Every message $m$ a 
server receives, contains five fields: (1) the $op$ field that represents the type of the 
request, (2) a $data$ field, that contains either the data to be enqueued or 
$\bot$, (3) a $cid$ field that contains the client id that requested the operation 
$op$, (4) a $sid$ field that contains the id of the server which started
forwarding the request, and -1 otherwise, and (5) a $tk$ field, a flag 
used to pass tokens from one server to another. The values that $tk$ can take are 
either \HEAD\ for the head token transition, \TAIL\ for the tail token transition, or $\bot$ 
for no token transition.

When server $s$ receives the 
head token, that means that $s$ can serve all operations in its client array regarding
the head endpoint (\texttt{EnqueueHead()}, \texttt{DequeueHead()}). In analogous way, when server 
$s$ receives the token for the tail, that means that $s$ can serve all operations in its client 
array for the tail (\texttt{EnqueueTail()}, \texttt{DequeueTail()}). For this purpose, we use 
two functions, \texttt{ServeOldHeadOps()} and \texttt{ServeOldTailOps()}. The clients whose 
requests are served, are not informed until server $s$ receives their requests completes
a round-trip on the ring and returns back to $s$. 

When a message is received (line \lref{ln:a32-18}), server $s$ checks if the message is 
stored in its $client$ array (line \lref{ln:a32-19}). If it is, $s$ sends  
an \ACK\ message to client with $cid$ (line \lref{ln:a32-1}) and removes the entry 
in the $clients$ array (line \lref{ln:a32-2}). If it is not, $s$ checks the message's 
operation code and acts accordingly (line \lref{ln:a32-9}). Messages with operation code \ENQT, 
\DEQT, \ENQH, \DEQH\ (lines \lref{ln:a32-10}, \lref{ln:a32-12}, \lref{ln:a32-14} and 
\lref{ln:a32-16}, respectively) are handled by functions \texttt{ServerEnqueueTail()} (line 
\lref{ln:a32-11}),  \texttt{ServerDequeueTail()} (line \lref{ln:a32-13}), 
\texttt{ServerEnqueueHead()} (line \lref{ln:a32-15}) and \texttt{ServerDequeueHead()} (line 
\lref{ln:a32-17}), respectively. These functions are presented in Algorithms 
\ref{alg33}-\ref{alg36}. If the message is of type \ACK\ (line \lref{ln:a32-3}) or 
\NACK\ (line \lref{ln:a32-6}), $s$ sets the $cid$ entry of the
clients array to $\bot$ (lines \lref{ln:a32-4}, \lref{ln:a32-7}), and sends an \ACK\ (line 
\lref{ln:a32-5}) or \NACK\ (line \lref{ln:a32-8}) to the client.

Algorithm \ref{alg33} presents pseudocode for function \texttt{ServerEnqueueTail()}, which is 
called by a server $s$ when an \ENQT\ message is received. First, $s$ checks whether the 
$tk$ field of the message contains \TAIL\ (line \lref{ln:a33-1}). In such a case, 
the message received by $s$ denotes a tail token transition. So, $s$ sets its $hasTail$ 
flag to \true\ (line \lref{ln:a33-2}) and if it also had the head token from a previous round, 
it sets its $fullDeque$ flag to \true\ (line \lref{ln:a33-3}). At this point, $s$ has 
just received the global deque's tail, so it must serve all old operations that clients have 
requested directly from this server to be performed in the deque's tail. For that purpose, the 
server calls the \texttt{ServeOldTailOps()} function (line \lref{ln:a33-4}). 
Then, the server 
checks whether the global deque is full and also the $ldeque$ is full and if it is full, 
means that there is no free space left for the operation, thus $s$ responds to the request with
with a \NACK\ 
(lines \lref{ln:a33-5}, \lref{ln:a33-6}). 
Otherwise, if the server can serve the request, it enqueues 
the received data to the tail of $ldeque$ (line \lref{ln:a33-7}) and responds 
with an \ACK\ (lines \lref{ln:a33-8}, \lref{ln:a33-9}). In any other case, $s$ cannot serve the 
request received, but some other server might be capable to do so, so the message must be 
forwarded. Server $s$ finds the next server (line \lref{ln:a33-10}) and if $s$ handles the 
global deque's tail, turns its flags ($hasTail$ and $fullDeque$) to \false\ 
(line \lref{ln:a33-11}) and marks the token as \TAIL\ for the token transition (line 
\lref{ln:a33-12}). If the message was send by a client, $s$ stores the request to the 
$clients$ array (line \lref{ln:a33-13}). Finally, $s$ forwards the message to the next 
server (lines \lref{ln:a33-14}, \lref{ln:a33-15}).

\begin{algorithm}[!t]
\small
\caption{Server helping function for handling an enqueue request to the global deque's tail}
\label{alg33}
\begin{code}
 \firstline
    \void\ ServerEnqueueTail(\integer\ $op$, Data $data$, \integer\ $cid$, 
                                 \integer\ $sid$, enum  $tk$) $\lbrace$    \nl
 \n   \If\ ($tk == \TAIL$) $\lbrace$                   \label{ln:a33-1}    \nl
 \n     $hasTail$ = \true;                                \label{ln:a33-2} \nl
        \If\ ($hasHead$)   $fullDeque$ = \true;           \label{ln:a33-3} \nl
        ServeOldTailOps($clients, fullDeque$);            \label{ln:a33-4} \ul
 \p   $\rbrace$                                                            \nl
      \If\ ($fullDeque$ AND IsFull($ldeque$)) $\lbrace$
   	                                     \cm{Deque is full, can't enqueue} \nl
 \n     \If\ ($sid == -1$)                                \cm{From client} \nl
 \n       send($id, \langle \NACK, \bot, my\_sid\rangle$);\label{ln:a33-5} \nl
 \p     \Else\                                            \cm{From server} \nl
 \n       send($sid, \langle \NACK, \bot, cid, my_sid, \bot\rangle$);
                                                          \label{ln:a33-6} \nl
 \p\p $\rbrace$ \Elseif\ ($hasTail$ AND !IsFull($ldeque$)) $\lbrace$ 
                                                   \cm{Server can enqueue}  \label{ln:a33-6.1} \nl
 \n     enqueue\_tail($deque, data$);                     \label{ln:a33-7} \nl
        \If\ ($sid == -1$)    \cm{From client}                             \nl
 \n       send($cid, \langle \ACK, \bot, my\_sid\rangle$);\label{ln:a33-8} \nl
 \p     \Else\                \cm{From server}                             \nl
 \n       send($sid, \langle \ACK, \bot, cid, my\_sid, \bot\rangle$);
                                                          \label{ln:a33-9} \nl
 \p\p $\rbrace$ \Else\ $\lbrace$  
                    \cm{Server can't enqueue and global deque is not full} \nl
 \n     $nsid$ = find\_next\_server($my\_sid$);          \label{ln:a33-10} \nl
        \If\ ($hasTail$) $\lbrace$                                         \nl
 \n       $hasTail$ = \false;                            \label{ln:a33-11} \nl
          $fullDeque$ = \false;                                            \nl
          $tk = \TAIL$;                               	 \label{ln:a33-12} \nl
 \p     $\rbrace$ \Else\ $\lbrace$                                         \nl
 \n       $tk = \bot$;                                                     \ul
 \p     $\rbrace$                                                          \nl
        \If\ ($sid == -1$) $\lbrace$                      \cm{From client} \nl
 \n       clients[$cid$] = $\langle\ENQT, data, false\rangle$;  
                                                         \label{ln:a33-13} \nl
          send($nsid, \langle\ENQT, data, cid, my\_sid, tk\rangle$);
                                                         \label{ln:a33-14} \nl
 \p     $\rbrace$ \Else\ $\lbrace$                        \cm{From server} \nl
 \n       send($nsid, \langle\ENQT, data, cid, sid, tk\rangle$);
                                                         \label{ln:a33-15} \ul
 \p     $\rbrace$                                                          \ul
 \p   $\rbrace$                                                            \ul
 \p $\rbrace$
\end{code}
\end{algorithm}

Algorithm \ref{alg34} presents pseudocode for function \texttt{ServerDequeueTail()}, 
which is called by a server $s$ when a \DEQT\ message is received.
First, $s$ checks whether the $tk$ field of the message contains \TAIL\ 
(line \lref{ln:a34-1}). In such a case, the message received by $s$, denotes a tail 
token transition. So, $s$ sets its $hasTail$ flag to \true\ (line \lref{ln:a34-2}). 
then, $s$ serves all old operations that clients have requested directly from this 
server to be performed in the deque's tail. For that purpose, the server calls 
\texttt{ServeOldTailOps()} (line \lref{ln:a34-3}). Afterwards, $s$ checks if the 
``global'' deque is empty and if it is empty, $s$ responds with a \NACK\ (lines 
\lref{ln:a34-4}, \lref{ln:a34-5}). Otherwise, if $s$ can serve the request, it 
dequeues the data from the tail of its local deque (line \lref{ln:a34-6}), and 
sends them to the appropriate server or client with an \ACK\ (lines \lref{ln:a34-7}, 
\lref{ln:a34-8}). In any other case, $s$ cannot serve the request received, but some 
other server might be capable to do so, thus the message must be forwarded. Server 
$s$ finds the previous server $s_{prev}$ (line \lref{ln:a34-9}) and if $s$ handles 
the global deque's tail, turns its flag for the tail to \false\ (line \lref{ln:a34-10}) 
and marks $tk$ as \TAIL\ for the token transition (line \lref{ln:a34-11}). If the 
message received was send by a client, $s$ stores the request to its $clients$ 
array (line \lref{ln:a34-12}). Finally, $s$ forwards the message to the previous 
server (lines \lref{ln:a33-13}, \lref{ln:a33-14}).

\begin{algorithm}[!t]
\small
\caption{Server helping function for handling an dequeue request to the global deque's tail}
\label{alg34}
\begin{code}
\firstline
   \void\ ServerDequeueTail(\integer\ $op$, \integer\ $cid$, \integer\ $sid$, enum $tk$) $\lbrace$ \nl
   \n   \If\ ($tk == \TAIL$) $\lbrace$                       	\label{ln:a34-1} \nl
   \n     $hasTail$ = \true;                                    \label{ln:a34-2} \nl
          ServeOldTailOps($clients, fullDeque$);                \label{ln:a34-3} \ul
   \p   $\rbrace$                                                                \nl
        \If\ ($hasHead$ AND $hasTail$ AND $!fullDeque$ AND IsEmpty($ldeque$)) $\lbrace$ \ul
         \cm{ Deque is empty, can't dequeue} \nl
   \n     \If\ ($sid == -1$)                                    \cm{From client} \nl
   \n       send($cid, \langle\NACK, \bot, my\_sid\rangle$);    \label{ln:a34-4} \nl
   \p     \Else\                                                \cm{from server} \nl
   \n       send($sid, \langle\NACK, \bot, cid, my\_sid, \bot\rangle$); 
                                                                \label{ln:a34-5} \nl
   \p\p $\rbrace$ \Elseif\ ($hasTail$ AND !IsFull($ldeque$)) $\lbrace$
                                        \label{ln:a34-6} \cm{Server can dequeue} \nl
   \n       $data$ = dequeue\_tail($ldeque$);              \label{ln:a34-6.1}    \nl
            \If\ ($sid == -1$)                                  \cm{From client} \nl
   \n         send($cid, \langle \ACK, data, my\_sid\rangle$);  \label{ln:a34-7} \nl
   \p       \Else\                                              \cm{from server} \nl
   \n         send($sid, \langle \ACK, data, cid, my\_sid, \bot\rangle$);
                                                                \label{ln:a34-8} \nl
   \p\p   $\rbrace$ \Else\ $\lbrace$ 
                        \cm{Server can't dequeue and global deque is not empty.} \nl
   \n       $psid$ = find\_previous\_server($my\_sid$);         \label{ln:a34-9} \nl
            \If\ ($hasTail$) $\lbrace$                                           \nl
   \n          $hasTail$ = \false;                             \label{ln:a34-10} \nl
               $tk$ = $\TAIL$;                                 \label{ln:a34-11} \nl
   \p       $\rbrace$ \Else\ $\lbrace$                                           \nl
   \n         $tk = \bot$;                                                    	 \ul
   \p       $\rbrace$                                                            \nl
            \If\ ($sid == -1$) $\lbrace$                        \cm{From client} \nl
   \n         clients[$cid$] = $\langle \DEQT, \bot, \false\rangle$; 
                                                               \label{ln:a34-12} \nl
              send($psid, \langle \DEQT, \bot, cid, my\_sid, tk\rangle$);
                                                               \label{ln:a34-13} \nl
   \p       $\rbrace$ \Else\ $\lbrace$                         \cm{from server}  \nl
   \n         send($psid, \langle \DEQT, \bot, cid, sid, tk\rangle$); 
                                                               \label{ln:a34-14} \ul
   \p       $\rbrace$                                                            \ul
   \p     $\rbrace$                                                              \ul
   \p   $\rbrace$
\end{code}
\end{algorithm}

Algorithm \ref{alg35} presents pseudocode for function \texttt{ServerEnqueueHead()}, which is 
called by a server $s$ when an \ENQH\ message is received. First, $s$ checks if 
it has the token for global deque's head (line \lref{ln:a35-1}) and if this is the case, the 
server sets its $hasHead$ flag to \true\ (line \lref{ln:a35-2}). If $s$ already has the 
token for the global deque's tail, it sets its $fullDeque$ flag to \true\ (line 
\lref{ln:a35-3}). At this point, $s$ has just received the global deque's head, so it must serve 
all old operations that clients have requested directly from this server to be performed in the 
deque's head. For that purpose, server $s$ calls \texttt{ServeOldHeadOps()} (line 
\lref{ln:a35-4}), which iterates $s$'s $clients$ array and serves all operations for the 
deque's head. Then, the server 
checks whether the global deque is full and also the $ldeque$ is full and if it is full, 
means that there is no free space left for the operation, thus $s$ responds to the request with
with a \NACK\ (lines \lref{ln:a35-5}, \lref{ln:a35-6}). Otherwise, if the server can serve the 
request, it enqueues the received data to the deque's head (line \lref{ln:a35-7}) and responds 
with an \ACK\ (lines \lref{ln:a35-8}, \lref{ln:a35-9}). In any other case, $s$ cannot 
serve the request received, but some other server might be capable to do so, so the message must 
be forwarded. Server $s$ finds the previous server (line \lref{ln:a35-10}) and if $s$ handles 
the global deque's head, turns its flags ($hasHead$ and $fullDeque$) to \false\ 
(line \lref{ln:a35-11}) and marks the $tk$ as \HEAD\ for the token transition (line 
\lref{ln:a35-12}). If the message received was send by a client, $s$ stores the request to its 
$clients$ array (line \lref{ln:a35-13}). Finally, $s$ forwards the message to the previous 
server (lines \lref{ln:a35-14}, \lref{ln:a35-15}).

\begin{algorithm}[!t]
\small
\caption{Server helping function for handling an enqueue request to the global deque's head}
\label{alg35}
\begin{code}
	\firstline
	  \void\ ServerEnqueueHead(\integer\ $op$, Data $data$, \integer\ $cid$, \integer\ $sid$, enum $tk$)
	   $\lbrace$ \nl
	\n   \If\ ($tk == \HEAD$) $\lbrace$                           \label{ln:a35-1} \nl
	\n     $hasHead$ = \true;                                     \label{ln:a35-2} \nl
	       \If\ ($hasTail$) $fullDeque$ = \true;                  \label{ln:a35-3} \nl
	       ServeOldHeadOps(clients);                              \label{ln:a35-4} \ul
	\p   $\rbrace$                                                                 \nl
	     \If\ ($fullDeque$ AND IsFull($ldeque$))$\lbrace$                        
	                                             \cm{Deque is full, can't enqueue} \nl
	\n     \If\ ($sid == -1$)                                     \cm{From client} \nl
	\n       send($cid, \langle\NACK, \bot, my\_sid\rangle$);     \label{ln:a35-5} \nl
	\p     \Else\                                                 \cm{from server} \nl
	\n       send($sid, \langle\NACK, \bot, cid, my\_sid, \bot\rangle$);
	                                                              \label{ln:a35-6} \nl
	\p\p $\rbrace$ \Elseif\ ($hasHead$ AND !IsFull($ldeque$)) $\lbrace$
                                               \cm{Server can Server can enqueue.} \nl
	\n       enqueue\_head($ldeque, data$);                       \label{ln:a35-7} \nl
	         \If\ ($sid == -1$)                                   \cm{From client} \nl
	\n         send($cid, \langle\ACK, \bot, my\_sid\rangle$);    \label{ln:a35-8} \nl
	\p       \Else\                                               \cm{from server} \nl
	\n         send($sid, \langle\ACK, \bot, cid, my\_sid, \bot\rangle$);
	                                                              \label{ln:a35-9} \nl
	\p\p   $\rbrace$ \Else\ $\lbrace$ 
	                       \cm{Server can't dequeue and global deque is not full.} \nl
	\n       $psid$ = find\_previous\_server($my\_sid$);         \label{ln:a35-10} \nl
	         \If\ ($hasHead$) $\lbrace$                                            \nl
	\n         $hasHead$ = \false;                               \label{ln:a35-11} \nl
	           $fullDeque$ = \false;                                               \nl
	           $tk$ = \HEAD;                                     \label{ln:a35-12} \nl
	\p       $\rbrace$ \Else\ $\lbrace$                                            \nl
	\n         $tk = \bot$;                                                        \ul
	\p       $\rbrace$                                                             \nl
	         \If\ ($sid == -1$) $\lbrace$                         \cm{From client} \nl
	\n         clients[$cid$] = $\langle \ENQH, data, \false\rangle$; 
	                                                             \label{ln:a35-13} \nl
	           send($psid, \langle \ENQH, data, cid, my\_sid, tk\rangle$);
	                                                             \label{ln:a35-14} \nl
	\p       $\rbrace$ \Else\ $\lbrace$                          \cm{from server}  \nl
	\n         send($psid, \langle \ENQH, data, cid, sid, tk\rangle$);
	                                                             \label{ln:a35-15} \ul
	\p       $\rbrace$                                                             \ul
	\p     $\rbrace$                                                               \ul
	\p   $\rbrace$
\end{code}
\end{algorithm}

Algorithm \ref{alg36} presents pseudocode for function \texttt{ServerDequeueHead()}, which is 
called by a server when a \DEQH\ message is received by some server $s$. First, $s$ checks if it 
has the token for global deque's head (line \lref{ln:a36-1}) and if this is the case, the server 
sets its $hasHead$ flag to \true\ (line \lref{ln:a36-2}). At this point, server $s$ has 
just received the global deque's head, so it must serve all old operations that clients have 
requested directly from this server to be performed in the deque's head. For that purpose, the 
server calls the \texttt{ServeOldHeadOps()} routine (line \lref{ln:a36-4}), which iterates $s$'s 
$clients$ array and serves all operations for the deque's head. Then, the server checks 
if the global deque is empty and if it is, $s$ responds to the request with a \NACK\ (lines 
\lref{ln:a36-5}, \lref{ln:a36-6}). Otherwise, if the server can serve the request, it dequeues 
the data from the head (line \lref{ln:a36-7}) and sends them to the sender or client together with 
an \ACK\ (lines \lref{ln:a36-8}, \lref{ln:a36-9}). In any other case, $s$ cannot serve the request 
received, but some other server might be capable to do so, so the message must be forwarded. 
Server $s$ finds the next server (line \lref{ln:a36-10}) and if $s$ has the token for the global 
head, turns its flag for the head to \false\ (line \lref{ln:a36-11}) and and marks the $tk$ as 
\HEAD\ for the token transition (line \lref{ln:a36-12}). If the message received was send by a 
client, $s$ stores the request to its $clients$ array (line \lref{ln:a36-13}). Finally, $s$ 
forwards the message to the next server (lines \lref{ln:a36-14}, \lref{ln:a36-15}).

\begin{algorithm}[!t]
\small
\caption{Server helping function for handling an dequeue request to the global deque's head}
\label{alg36}
\begin{code}
	\firstline
	\void\ ServerDequeueHead(\integer\ $op$, \integer\ $cid$, \integer\ $sid$, enum $tk$) $\lbrace$ \nl
	\n   \If\ ($tk == \HEAD$) $\lbrace$                        \label{ln:a36-1} \nl
	\n     $hasHead$ = \true;                                  \label{ln:a36-2} \nl
	       ServeOldHeadOps(clients);                           \label{ln:a36-4} \ul
	\p   $\rbrace$                                                              \nl
	     \If\ ($hasHead$ AND $hasTail$ AND $!fullDeque$ AND IsEmpty($ldeque$))$\lbrace$ \ul
	     \cm{ Deque is empty, can't dequeue} \nl
	\n     \If\ ($sid == -1$)                                  \cm{From client} \nl
	\n       send($cid, \langle \NACK, \bot, my\_sid\rangle$); \label{ln:a36-5} \nl
	\p     \Else\                                              \cm{from server} \nl
	\n       send($sid, \langle \NACK, \bot, cid, my\_sid, \bot\rangle$);
	                                                           \label{ln:a36-6} \nl
	\p\p $\rbrace$ \Elseif\ ($hasHead$ AND !IsFull($ldeque$)) $\lbrace$
	                                        \cm{Server can Server can enqueue.} \nl
	\n       $data$ = dequeue\_head($ldeque$);                 \label{ln:a36-7} \nl
           \If\ ($sid == -1$)                                  \cm{From client} \nl
	\n         send($cid, \langle\ACK, data, my\_sid\rangle$); \label{ln:a36-8} \nl
	\p       \Else\                                            \cm{from server} \nl
	\n         send($sid, \langle\ACK, data, cid, my\_sid, \bot\rangle$);
                                                               \label{ln:a36-9} \nl
	\p\p   $\rbrace$ \Else\ $\lbrace$ 
	                   \cm{Server can't dequeue and global deque is not empty.} \nl
	\n       $nsid$ = find\_next\_server($my\_sid$);          \label{ln:a36-10} \nl
	\n          $hasHead$ = \false;                           \label{ln:a36-11} \nl
	            $tk$ = $\HEAD$;                               \label{ln:a36-12} \nl
	\p       $\rbrace$ \Else\ $\lbrace$                                         \nl
	\n         $tk = \bot$;                                                     \ul
	\p       $\rbrace$                                                          \nl
             \If\ ($sid == -1$) $\lbrace$                      \cm{From client} \nl
	\n         clients[$cid$] = $\langle \DEQH, \bot, \false\rangle$; 
	                                                          \label{ln:a36-13} \nl
	           send($nsid, \langle \DEQH, \bot, cid, my\_sid, tk\rangle$);
	                                                          \label{ln:a36-14} \nl
	\p       $\rbrace$ \Else\ $\lbrace$                        \cm{from server} \nl
	\n         send($nsid, \langle \DEQH, \bot, cid, sid, tk\rangle$);
	                                                          \label{ln:a36-15} \ul
	\p       $\rbrace$                                                          \ul
	\p     $\rbrace$                                                            \ul
	\p   $\rbrace$                                                               
\end{code}
\end{algorithm}

Algorithm \ref{alg37} presents pseudocode for the client functions. These are  
\texttt{EnqueueTail(), DequeueTail(), EnqueueHead(), DequeueHead()}. All clients 
have two local variables, $head\_sid$ and $tail\_sid$, to store 
the last known server to have the head token and the tail token, respectively. These 
variables initially store the id of the server zero, but may change values during runtime.
The messages received by clients contain two fields: $status$ which contains either
the data from a dequeue operation, or $\bot$ in case of an enqueue, and $tail\_sid$ 
or $head\_sid$, depending on the function, which contains the id of the server that 
currently holds the token. 

%

	\begin{minipage}{.45\textwidth}
		\begin{algorithm}[H]
\small
\caption{Enqueue and dequeue operations for a client of the token-based deque. }
\label{alg37}
\begin{code}

    \firstline
    \integer\ $tail\_sid$ = 0,  $head\_sid$ = 0;                                      \ul 
                                                                                      \nl
    Data EnqueueTail(\integer\ $cid$, Data $data$) $\lbrace$         \label{ln:a37-1} \nl
 \n   send($tail\_sid, \langle \ENQT, data, cid, -1, \bot\rangle $); \label{ln:a37-2} \nl
      $\langle status, tail\_sid\rangle$ = receive();                \label{ln:a37-3} \nl
      \return\ $status$;                                             \label{ln:a37-4} \ul
 \p $\rbrace$                                                                         \ul
                                                                                      \nl
    Data DequeueTail(\integer\ $cid$) $\lbrace$                      \label{ln:a37-5} \nl
 \n   send($tail\_sid, \langle \DEQT, \bot, cid, -1, \bot\rangle$);  \label{ln:a37-6} \nl
      $\langle status, tail\_sid\rangle$ = receive();                \label{ln:a37-7} \nl
      \return\ $status$;                                             \label{ln:a37-8} \ul
 \p $\rbrace$                                                                         \ul
                                                                                      \nl
    Data EnqueueHead(\integer\ $cid$, Data $data$) $\lbrace$         \label{ln:a37-9} \nl
 \n   send($head\_sid, \langle\ENQH, data, cid, -1, \bot\rangle $); \label{ln:a37-10} \nl
      $\langle status, head\_sid\rangle$ = receive();                                 \nl
      \return\ $status$;                                                              \ul	
 \p $\rbrace$                                                                         \ul
                                                                                      \nl
    Data DequeueHead(\integer\ $cid$) $\lbrace$                     \label{ln:a37-11} \nl
 \n   send($head\_sid, \langle\DEQH, \bot, cid, -1, \bot\rangle$);  \label{ln:a37-12} \nl
      $\langle status, head\_sid\rangle$ = receive();                                 \nl
      \return\ $status$;                                                              \ul
 \p $\rbrace$
\end{code}
\end{algorithm}
\end{minipage}

For enqueuing at the deque's tail, clients calls \texttt{EnqueueTail()} (line \lref{ln:a37-1}). 
This function sends an \ENQT\ message to the last known server, which has the token for the 
global tail (line \lref{ln:a37-2}). The server with the tail token may have changed, but the 
client is still unaware of the change. In that case, server $s$, which received the message, 
stores this message in its $clients$ array and then forwards the message to the next 
server in the order, as described in Algorithm \ref{alg33}. During this time, the client 
blocks waiting for a server's response. Once the client receives the response (line 
\lref{ln:a37-3}), it returns the contents of the $status$ variable (line \lref{ln:a37-4}). 


%
\begin{minipage}{.5\textwidth}
\begin{algorithm}[H]
\caption{Auxiliary functions for a server of the token-based deque.}
\label{alg38-help}
\small
\begin{code}
\firstline
     \void\ ServeOldTailOps(\void) $\lbrace$                   \label{ln:a38h-4} \nl
\n  	LocalSet $eliminated = \emptyset$    \bl\nl
  	 
	    \Foreach\ $cid1 \not\in eliminated$ {\sf such that} clients[$cid1$].$op == \ENQT\ \lbrace$  
	    																\label{ln:a38h-21} \nl
\n       \If\  {\sf there is} $cid2 \not\in eliminated$ {\sf such that} 
				clients[$cid2$].$op == \DEQT\ \lbrace$                  \label{ln:a38h-22} \nl
\n 			clients[$cid2$].$data$ = clients[$cid1$].$data$;            \label{ln:a38h-23} \nl
			clients[$cid1$].$isServed$ = \true;                			\label{ln:a38h-24}\nl
			clients[$cid2$].$isServed$ = \true;                			\label{ln:a38h-25}\nl
			$eliminated = eliminated \cup \{cid1, cid2\}$;     			\label{ln:a38h-26}\ul
  \p $\rbrace$                              \ul
  \p $\rbrace$							
  \p\nl     
  \n         \If\ (!$fullDeque$) $\lbrace$                           \label{ln:a38h-7} \nl
          \n     \Foreach\ $cid$ {\sf such that} clients[$cid$].$op == \ENQH\ \lbrace$  \nl
          \n       \If\ (!IsFull($ldeque$)) $\lbrace$                  \label{ln:a38h-8} \nl
          \n         enqueue\_tail($ldeque$, clients[$cid$].$data$);         \label{ln:a38h-81}\nl
                     clients[$cid$].$isServed$ = \true;                \label{ln:a38h-82}\ul
          \p $\rbrace$                              \ul
          \p $\rbrace$								\ul
          \p $\rbrace$								\nl
          
          \Foreach\ $cid$ {\sf such that} clients[$cid$].$op == \DEQT\ \lbrace$      \nl
  \n     \If\ (!IsEmpty($ldeque$)) $\lbrace$                   \label{ln:a38h-5} \nl
  \n       clients[$cid$].$data$ = dequeue\_tail($ldeque$);          \label{ln:a38h-51}\nl
           clients[$cid$].$isServed$ = \true;                                    \ul
\p \p \p \} \} \} \bl\bl\nl     
      
      \void\ ServeOldHeadOps(\void) $\lbrace$                   \label{ln:a38h-1} \nl
\n  	LocalSet $eliminated = \emptyset$    \bl\nl
  	 
	    \Foreach\ $cid1 \not\in eliminated$ {\sf such that} clients[$cid1$].$op == \ENQH\ \lbrace$  
	    																\label{ln:a38h-11} \nl
\n       \If\  {\sf there is} $cid2 \not\in eliminated$ {\sf such that} 
				clients[$cid2$].$op == \DEQH\ \lbrace$                  \label{ln:a38h-12} \nl
\n 			clients[$cid2$].$data$ = clients[$cid1$].$data$;            \label{ln:a38h-13} \nl
			clients[$cid1$].$isServed$ = \true;                			\label{ln:a38h-14}\nl
			clients[$cid2$].$isServed$ = \true;                			\label{ln:a38h-15}\nl
			$eliminated = eliminated \cup \{cid1, cid2\}$;     			\label{ln:a38h-16}\ul
  \p $\rbrace$                              \ul
  \p $\rbrace$							
  \p\nl     
  \n   \If\ (!$fullDeque$) $\lbrace$                           \label{ln:a38h-2} \nl
  \n     \Foreach\ $cid$ {\sf such that} clients[$cid$].$op == \ENQH\ \lbrace$  \nl
  \n       \If\ (!IsFull($ldeque$)) $\lbrace$                  \label{ln:a38h-3} \nl
  \n         enqueue\_head($ldeque$, clients[$cid$].$data$);         \label{ln:a38h-31}\nl
             clients[$cid$].$isServed$ = \true;                \label{ln:a38h-32}\ul
  \p $\rbrace$                              \ul
  \p $\rbrace$								\ul
  \p $\rbrace$								\nl
    \Foreach\ $cid$ {\sf such that} clients[$cid$].$op == \DEQH\ \lbrace$      \nl
    \n     \If\ (!IsEmpty($ldeque$)) $\lbrace$                   \label{ln:a38h-6} \nl
    \n       clients[$cid$].$data$ = dequeue\_head($ldeque$);          \label{ln:a38h-61}\nl
             clients[$cid$].$isServed$ = \true;                                    \ul
\p \p \p \} \} \}
\end{code}
\end{algorithm}
\end{minipage}

The enqueuing at the deque's head is symmetrical to this approach and is achieved with the 
function \texttt{EnqueueHead()} (line \lref{ln:a37-9}). The only difference is that the server 
which is send the message to, is the server with id $head\_sid$ (line \lref{ln:a37-10}) 
and the message's operation code is \ENQH\ instead of \ENQT.

For dequeuing at the deque's tail, clients call \texttt{DequeueTail()} (line \lref{ln:a37-5}). 
This function sends 
a \DEQT\ message to the last known server which has the token for the global tail 
(line \lref{ln:a37-6}). Then, the client waits for server to respond. 
Once the client receives the response (line \lref{ln:a37-7}), it returns the contents of the 
$status$ variable (line \lref{ln:a37-8}). 

The dequeuing at the deque's head is symmetrical to this approach and is achieved with the 
function \texttt{DequeueHead()} (line \lref{ln:a37-11}). The only difference is that the 
server, to whom the message was sent, is the server with id $head\_sid$ (line 
\lref{ln:a37-12}) and the message's operation code is \DEQH\ instead of \DEQT.

\subsubsection{Proof of Correctness}
Let $\alpha$ be an execution of the token-based deque algorithm presented 
in Algorithms \ref{alg32}, \ref{alg33}, \ref{alg34}, \ref{alg35}, \ref{alg36}, 
\ref{alg37}, and \ref{alg38-help}. 

Each server maintains local boolean variables $hasHead$ and $hasTail$, 
with initial values \false. Whenever some server $s_{i}$ receives a \TAIL\ 
message, i.e. a message with its $tk$ field equal to \TAIL\ (line \lref{ln:a33-1}, line \lref{ln:a34-1}), 
the value of $hasTail$ is set to \true\ (line \lref{ln:a33-2}, line \lref{ln:a34-2}). By inspection 
of the pseudocode, it follows that the value of $hasTail$ is set to \false\ 
if the local deque of $s_{i}$ is full (line \lref{ln:a33-11}, line \lref{ln:a34-10}); 
then, a \TAIL\ message is sent to the next or previous server (line \lref{ln:a33-15}, line \lref{ln:a34-14}). 
The same holds for $hasHead$ and \HEAD\ messages, i.e. messages with their 
$tk$ field equal to \HEAD. Thus, the following observations holds.

\begin{observation}
\label{obs:tdq:unique}
At each configuration in $\alpha$, there is at most one server 
for which the local variable $hasHead$ ($hasTail$) has the value \true.
\end{observation}

\begin{observation}
\label{obs:tdq:direction}
In some configuration $C$ of $\alpha$, a \TAIL\ message is sent from a server $s_j$, $0 \leq j < \maxser -1$, to a server 
$s_k$, where $k = (j+1) \mod{\maxser}$, only if the local deque of $s_j$ is full in $C$. A \TAIL\ message is sent from a server $s_j$, $0 \leq j < \maxser -1$, to a server 
$s_k$, where $k = (j-1) \mod{\maxser}$, only if the local deque of $s_j$ is empty in $C$. 

Similary, a \HEAD\ message is sent from $s_j$ to $s_k$, where $k = (j+1) \mod{\maxser}$, only if the local deque of $s_j$ is empty in $C$. \HEAD\ message is sent from $s_j$ to $s_k$, where $k = (j-1) \mod{\maxser}$, only if the local deque of $s_j$ is full in $C$. 
\end{observation}

By inspection of the pseudocode, we see that a server performs an enqueue (dequeue) back 
operation on its local deque  $ldeque$ either when executing line \lref{ln:a33-7} 
(line \lref{ln:a34-6.1}) or when it executes {\tt ServeOldTailOps}. 
Further inspection of the pseudocode (lines \lref{ln:a33-1}-\lref{ln:a33-3}, line 
\lref{ln:a33-6.1}, as well as lines \lref{ln:a34-1}-\lref{ln:a34-3}, 
line \lref{ln:a34-6}), shows that these lines are executed when 
$hasTail = \true$. By inspection of the pseudocode, 
the same can be shown for $hasHead$. 
Then, the following observation holds.

\begin{observation}
\label{obs:tdq:tokenserver}
Whenever a server $s_j$ performs an enqueue or dequeue back (front) operation on its local 
deque, it holds that its local variable $hasTail$ ($hasHead$) is equal to \true. 
\end{observation}

If $hasTail = \true$ ($hasHead = \true$) for some server $s$ in some configuration $C$, 
then we say that $s$ has the tail (head) token. The server that has the tail token is 
referred to as {\em tail token server}. The server that has the head token is referred 
to as {\em head token server}.

By a straight-forward induction, the following lemma can be shown. 

\begin{lemma}
\label{tdq:client-msgs}
The mailbox of a client in any configuration of $\alpha$ contains at most one 
incoming message.
\end{lemma}

Let $op$ be any operation in $\alpha$. We assign a linearization point to $op$ 
by considering the following cases: 
\begin{compactitem}
\item If  $op$ is an enqueue back operation for which a tail token server 
executes an instance  of Algorithm \ref{alg32}, then it is linearized in 
the configuration resulting from the execution of either line \lref{ln:a33-5}, 
or line \lref{ln:a33-7}, or line \lref{ln:a38h-23}, or line \lref{ln:a38h-81}, 
whichever is executed for $op$ in that instance of Algorithm \ref{alg32} by 
the tail token server. 
\item If $op$ is a dequeue back operation for which a head token server 
executes an instance of Algorithm \ref{alg32}, then it is linearized in 
the configuration resulting from the execution of either line \lref{ln:a34-4}, 
or line \lref{ln:a34-6.1}, or line \lref{ln:a38h-23}, or line \lref{ln:a38h-51}, 
whichever is executed for $op$ in that instance of Algorithm \ref{alg32} by 
the tail token server.
\item If  $op$ is an enqueue front operation for which a tail token server 
executes an instance  of Algorithm \ref{alg32}, then it is linearized in 
the configuration resulting from the execution of either line \lref{ln:a35-5}, 
or line \lref{ln:a35-7}, or line \lref{ln:a38h-13},  or line \lref{ln:a38h-31}, 
whichever is executed for $op$ in that instance of Algorithm \ref{alg32} by 
the head token server. 
\item If $op$ is a dequeue front operation for which a head token server 
executes an instance of Algorithm \ref{alg32}, then it is linearized in 
the configuration resulting from the execution of either line \lref{ln:a36-5}, 
or line \lref{ln:a36-7}, or line \lref{ln:a38h-13}, or line \lref{ln:a38h-61}, 
whichever is executed for $op$ in that instance of Algorithm \ref{alg32} by 
the head token server.
\end{compactitem}

\begin{lemma}
The linearization point of an enqueue (dequeue) operation $op$ is placed in its execution interval.
\end{lemma}

\begin{proof}
Assume that $op$ is an enqueue back operation and let $c$ be the client that invokes it. 
After the invocation of $op$, $c$ sends a message to some server $s$ (line \lref{ln:a37-2}) 
and awaits a response. Recall that routine {\tt receive()} (line \lref{ln:a37-3}) 
blocks until a message is received. The linearization point of $op$ is placed in the 
configuration resulting from the execution of either  line \lref{ln:a33-5}, or line 
\lref{ln:a33-7}, or line \lref{ln:a38h-23}, or line \lref{ln:a38h-81}  by  $s_t$ for $op$. Notice that since the 
execution of Algorithm \ref{alg32} by $s_t$ is triggered by a message that contains 
the request for $op$, either of these lines is executed after the request by $c$ is 
received, i.e. after $c$ invokes {\tt EnqueueTail}, and thus, after the execution interval 
of $op$ starts. 

By definition, the execution interval of $op$ terminates in the configuration resulting 
from the execution of line \lref{ln:a37-4}. By inspection of the pseudocode, this line 
is executed after line \lref{ln:a37-3}, i.e. after $c$ receives a response by some server. 
In the following, we show that the linearization point of $op$ occurs before this response 
is sent to $c$. 

Let $s_j$ be the server that $c$ initially sends the request for $op$ to. 
By observation of the pseudocode, we see that $c$ may either receive a response 
from $s_j$ if $s_j$ executes lines \lref{ln:a32-1}, or \lref{ln:a33-5}, 
or \lref{ln:a33-8}.

To arrive at a contradiction, assume that either of these lines is executed in $\alpha$ 
before the configuration in which the linearization point of $op$ is placed. Thus, a tail 
token server $s_t$ executes lines line \lref{ln:a33-5}, or line \lref{ln:a33-7}, or line \lref{ln:a38h-23}, or line \lref{ln:a38h-81},  
in a configuration following the execution of lines  \lref{ln:a32-1}, or \lref{ln:a33-5}, 
or \lref{ln:a33-8} by $s_j$. Since the algorithm is event-driven, inspection of the pseudocode 
shows that in order for a tail token server to execute these lines, it must receive a message 
containing the request for $op$ either from a client or from another server. 

Assume first that a tail token server executes the algorithm after receiving a message 
containing a request for $op$ from a client. This is a contradiction, since, on one hand, 
$c$ blocks until receiving a response, and thus, does not sent further messages requesting 
$op$ or any other operation, and since $op$ terminates after $c$ receives the response by 
$s_j$, and on the other hand, any other request from any other client concerns 
a different operation $op'$.

Assume next that a tail token server executes the algorithm after receiving a message 
containing the request for $op$ from some other server. This is also a contradiction since 
inspection of the pseudocode shows that after $s_j$ executes either of the lines that sends 
a response to $c$, it sends no further message to some other server and instead, 
terminates the execution of that instance of the algorithm.

The argumentation regarding dequeue back, enqueue front, and dequeue front operations is analogous.
\end{proof}

Denote by $L$ the sequence of operations which have been assigned linearization 
points in $\alpha$ in the order determined by their linearization points.
Let $C_i$ be the configuration at which the $i$-th operation $op_i$ of $L$ is 
linearized. Denote by $\alpha_i$, the prefix of $\alpha$ which ends with $C_i$ 
and let $L_i$ be the prefix of $L$ up until the operation that is linearized 
at $C_i$. Denote by $D_i$ the sequence of values that a sequential deque contains 
after applying the sequence of operations in $L_i$, in order, starting from an 
empty deque; let $D_0 = \epsilon$, i.e. $D_0$ is the empty sequence.
In the following, we denote by $s_{t_i}$ the tail token server at $C_i$ 
and by $s_{h_i}$ the head token server at $C_i$.

\begin{lemma}
For each  $i$, $i \geq 0$, if 
$ld_i^j$ are the contents of the local deque of server $s_j$ at 
$C_i$, $h_i \leq j \leq t_i$, at $C_i$, then it holds that 
$D_i = ld_i^{h_i} \cdot ld_i^{{h_i}+1} \cdot \ldots \cdot ld_i^{t_i}$ at $C_i$.
\end{lemma}

\begin{proof}
We prove the claim by induction on $i$.
The claim holds trivially at $i = 0$.

Fix any $i \geq 0$ and assume that at $C_i$, it holds that 
$D_i$ = $ld_i^{h_i} \cdot ld_i^{{h_i}+1} \cdot \ldots \cdot ld_i^{t_i}$. 
We show that the claim holds for $i+1$.

Assume that $op_{i+1}$ is an enqueue back operation by client $c$. 
Furthermore, distinguish the following two cases: 
\begin{itemize}
\item Assume that $t_i = t_{i+1}$. Then, by the induction hypothesis, 
$D_i = ld_i^{h_i} \cdot ld_i^{{h_i}+1} \cdot \ldots \cdot ld_i^{t_i}$. In case the local 
queue of $s_{t_i}$ is not full, $s_{t_i}$ enqueues the value $v_{i+1}$ of the 
$data$ field of the request for $op_{i+1}$ in the local deque (line \lref{ln:a33-7} 
or line \lref{ln:a38h-81}). Notice that, by Observation \ref{obs:tdq:tokenserver} changes on the local deques of servers occur only 
on token servers. Notice also that those changes occur only in a step that immediately 
precedes a configuration in which an operation is linearized. Thus, no further 
change occurs on the local deques of $s_{h_i}, s_{{h_i}+1}, \ldots, s_{t_i}$ 
between $C_i$ and $C_{i+1}$, other than the enqueue on $ld_i^t$. Then, it holds 
that $D_{i+1} = D_i \cdot v_{i+1} = ld_i^{h_i} \cdot ld_i^{{h_i}+1} \cdot \ldots \cdot ld_i^{t_i} \cdot v_{i+1} = 
ld_i^{h_i} \cdot ld_i^{{h_i}+1} \cdot \ldots \cdot ld_{i+1}^{t_i} = 
ld_{i+1}^{h_i} \cdot ld_{i+1}^{{h_i}+1} \cdot \ldots \cdot ld_{i+1}^{t_i}$, and
if the head token server does not change between $C_i$ and $C_{i+1}$, then $h_{i+1} = h_i$ and 
$D_{i+1} = ld_{i+1}^{h_{i+1}} \cdot ld_{i+1}^{{h_{i+1}}+1} \cdot \ldots \cdot ld_{i+1}^{t_{i+1}}$ 
and the claim holds. If the head token server changes, i.e., if $h_{i+1} \neq h_i$, then by 
Observation \ref{obs:tdq:direction}, $ld_{i+1}^{h_i} = \emptyset$ and the claim holds again.

In case the local deque of $s_{t_i}$ is full and since by assumption, $s_{t_i} = s_{t_{i+1}}$, 
it follows by inspection of the pseudocode (line \lref{ln:a30-5}) and the definition of linearization 
points, that $s_{t_{i+1}} = s_{h_{i+1}}$. In this case, $s_{t_{i+1}}$ responds with a \NACK\ 
to $c$ and the local deque remains unchanged. Since no token server changes between $C_i$ and $C_{i+1}$, 
$D_{i+1} = D_i = ld_i^{h_i} \cdot ld_i^{{h_i}+1} \cdot \ldots \cdot ld_i^{t_i} = 
ld_{i+1}^{h_{i+1}} \cdot ld_{i+1}^{{h_{i+1}}+1} \cdot \ldots \cdot ld_{i+1}^{t_{i+1}}$ 
and the claim holds.

\item Next, assume that $t_i \neq t_{i+1}$. This implies that the local deque of $s_{t_i}$ 
is full just after $C_i$. Observation \ref{obs:tdq:direction} implies that $s_{t_i}$ 
forwarded the token to $s_{t_i+1}$ in some configuration between $C_i$ and $C_{i+1}$. 
Notice that then, $s_{t_i+1} = s_{t_{i+1}}$. If the local deque of $s_{t_{i+1}}$ is 
not full, then the condition of line \lref{ln:a33-6.1} evaluates to \true\ and therefore, line 
\lref{ln:a33-7} is executed, enqueueing value $v_{i+1}$ to it. Then at $C_{i+1}$, $ld_{i+1}^{t_{i+1}} = v_{i+1}$. 
By definition, $D_{i+1} = D_i \cdot v_{i+1}$, and therefore, 
$D_{i+1} = ld_i^{h_i} \cdot ld_i^{{h_i}+1} \cdot \ldots \cdot ld_i^{t_i} \cdot v_{i+1} = 
ld_{i+1}^{h_{i+1}} \cdot ld_{i+1}^{{h_{i+1}}+1} \cdot \ldots \cdot ld_{i+1}^{t_i} \cdot v_{i+1} = 
ld_{i+1}^{h_{i+1}} \cdot ld_{i+1}^{{h_{i+1}}+1} \cdot \ldots \cdot ld_{i+1}^{t_i} \cdot ld_{i+1}^{t_{i+1}}$ 
and the claim holds.  If the local deque of $s_{t_{i+1}}$ is full, then the condition of 
line \lref{ln:a33-6.1} evaluates to \false\ and therefore, line \lref{ln:a33-5} is executed. The operation 
is linearized in the resulting configuration and \NACK\ is sent to $c$. Notice that in that 
case, the local deque of the server is not updated. Then, 
$D_{i+1} = D_i = ld_i^{h_i} \cdot ld_i^{{h_i}+1} \cdot \ldots \cdot ld_i^{t_i} \cdot ld_{i+1}^{t_{i+1}} = 
ld_{i+1}^{h_{i+1}} \cdot ld_{i+1}^{{h_{i+1}}+1} \cdot \ldots \cdot ld_{i+1}^{t_i} \cdot ld_{i+1}^{t_{i+1}}$, 
and the claim holds. 
\end{itemize}
The reasoning for the case where $op_{i+1}$ is an instance of a dequeue back, enqueue front, or enqueue back operation is symmetrical.
\end{proof}

From the above lemmas and observations we have the following theorem.

\begin{theorem}
The token-based distributed deque implementation is linearizable. 
The time complexity and the communication complexity of each operation $op$ is $O(\maxser)$.
\end{theorem}
\subsection{Hierarchical approach.}
In this section, we outline how the hierarchical approach, 
described in Section~\ref{sec:impl-paradigms}, is applied to the 
token-based designs.

Only the island masters play the role of clients to the algorithms described in this section.
So, it is each island master $m_i$ that keeps track of the last server(s), which responded
to its batches of requests. In the stack and deque implementations, $m_i$ performs elimination before
contacting a server. In the queue implementation, batching is done by having each batch containing requests of the same type.
In the deque implementation, each batch contains requests of the same type that are to be applied to the same endpoint. 
A batch can be sent to a server using DMA; the same could be done for getting back the responses. 
A server that does not hold the appropriate token to serve a batch of requests, 
forwards the entire batch to the next (or previous) server. Since token-based algorithms exploit
locality, a batch of requests will be processed by at most two servers.


\subsection{Dynamic Versions of the Implementations}
\label{app:directory-dyn}

\begin{algorithm}[!b]
\caption{Events triggered in a server of a dynamic token-based deque.}
\label{alg30-d}
\begin{code}
\lreset
	\firstline
    a message $\langle op, data, cid, sid, tk\rangle$ is received:           \nl
 \n   \If\ (!clients[$cid$] AND clients[$cid$].$isServed$) $\lbrace$ \ul
                                       \cm{If message has been served earlier.} \nl
 \n     send($cid, \langle\ACK,$ clients[$cid$].$data,\ my\_sid\rangle$);       \nl
        clients[$cid$] = $\bot$;                                                \nl
 \p   $\rbrace$ \Else\ $\lbrace$                                                \nl
 \n     \Switch\ ($op$) $\lbrace$                                               \nl
 \n 	  \Case\ \ENQ:                                      \label{ln:alg30d-1} \nl
 \n         \If\ ($tk$ == \TAIL) $\lbrace$                                      \nl
 \n           $hasTail$ = \true;                                                \nl
              ServeOldEnqueues();                                               \ul
 \p         $\rbrace$                                                           \nl
            \If\ (!$hasTail$) $\lbrace$         \cm{Server does not have token} \nl
 \n           $nsid$ = find\_next\_server($my\_sid$);                           \nl
              \If\ ($sid == -1$) $\lbrace$                   \cm{From client.}  \nl
 \n             clients[$cid$] = $\langle\ENQ, data, \false\rangle$;            \nl
                send($nsid,\langle\ENQ, data, cid, my\_sid, \bot\rangle$);      \nl
 \p           $\rbrace$ \Else\ $\lbrace$                      \cm{From server.} \nl
 \n             send($nsid,\langle\ENQ, data, cid, sid, \bot\rangle$);          \ul
 \p           $\rbrace$                                                         \nl
 \p         $\rbrace$ \Elseif\ (!IsFull($lqueue$)) $\lbrace$                    \nl
 \n          enqueue($lqueue, data, tail\_round$);                              \nl
             \If\ ($sid == -1$)                               \cm{From client.} \nl
 \n            send($cid, \langle\ACK, \bot, my\_sid\rangle$);                  \nl
 \p          \Else                                            \cm{From server.} \nl
 \n            send($sid, \langle\ACK, \bot, cid, my\_sid, \bot\rangle$);       \nl
 \p\p       $\rbrace$ \Else\ $\lbrace$ \cm{Server moves the tail token} 
                                                            \label{ln:alg30d-2} \nl                   
 \n           $nsid$ = find\_next\_server($my\_sid$);                           \nl
              send($nsid, \langle op, data, cid, my\_sid, \TAIL\rangle$);       \nl
              $tail\_round++$;                              \label{ln:alg30d-3} \nl
              allocate\_new\_space($lqueue, tail\_round$);  \label{ln:alg30d-4} \nl
              $hasTail$ = \false;                                               \ul
 \p        $\rbrace$                                                            \nl
\Break;
\end{code}
\end{algorithm}

\begin{algorithm}[!t]
\begin{code}
     \firstline
 \p~~~~~~~~~~\Case\ \DEQ:                                                       \nl
 \n        \If\ ($tk$ == \HEAD) $\lbrace$                                       \nl
 \n          $hasHead$ = \true;                                                 \nl
             ServeOldDequeues();                                                \ul
 \p        $\rbrace$                                                            \nl
           \If\ (!$hasHead$) $\lbrace$                                          \nl
 \n          $nsid$ = find\_next\_server($my\_sid$);                            \nl
           \If\ ($sid == -1$) $\lbrace$                        \cm{From client} \nl
 \n            clients[$cid$] = $\langle\DEQ, \bot, \false\rangle$;             \nl
               send($nsid, \langle\DEQ, \bot, cid, my\_sid\rangle$);            \nl
 \p          $\rbrace$ \Else\ $\lbrace$                        \cm{From server} \nl
 \n            send($nsid, \langle\DEQ, \bot, cid, sid, \bot\rangle$);          \ul
 \p          $\rbrace$                                                          \nl
 \p        $\rbrace$ \Elseif\ (!IsEmpty($lqueue$)) $\lbrace$  \cm{can dequeue.} \nl
 \n          $data$ = dequeue($lqueue, head\_round$);                           \nl
             \If\ ($sid == -1$)                                \cm{From client} \nl
 \n            send($cid, \langle\ACK, data, my\_sid\rangle$);                  \nl
 \p          \Else\                                            \cm{From server} \nl
 \n            send($sid, \langle\ACK, data, cid, my\_sid, \bot\rangle$);       \nl
 \p\p      $\rbrace$ \Elseif\ ($tail\_round == head\_round$) $\lbrace$ 
                                                            \label{ln:alg30d-5} \nl
 \n          $tail\_round = head\_round = 0$;               \label{ln:alg30d-6} \nl
 \n          \If\ ($sid == -1$)                            \cm{empty to client} \nl
 \n            send($cid, \langle\NACK, \bot, my\_sid\rangle$);                 \nl
 \p          \Else \                                       \cm{empty to server} \nl
 \n            send($sid, \langle\NACK, \bot, cid, my\_sid, \bot\rangle$ );     \nl
 \p\p      $\rbrace$ \Else\ $\lbrace$          \cm{Move the head token to next} \nl
 \n          $nsid$ = find\_next\_server($my\_sid$);                            \nl
             send($nsid, \langle op, \bot, cid, my\_sid, \HEAD\rangle$);        \nl
             $head\_round++$;                               \label{ln:alg30d-7} \nl
             $hasHead$ = \false;                                                \ul
 \p        $\rbrace$                                                            \nl
           \Break;                                                              \nl
 \p      \Case\ \ACK:                                                           \nl
 \n        clients[$cid$] = $\bot$;                                             \nl
           send($cid, \langle\ACK, data, sid\rangle$);                          \nl
           \Break;                                                              \nl
 \p      \Case\ \NACK:                                                          \nl
 \n        clients[$cid$] = $\bot$;                                             \nl
           send($cid, \langle\NACK, \bot, sid\rangle$);                         \nl
           \Break;                                                              \ul
 \p\p  $\rbrace$                                                                \ul
 \p  $\rbrace$
 \p\p 
\end{code}
\end{algorithm}

The implementations presented above (in Section~\ref{app:token}) are static. 
Their dynamic versions retain the placement of servers in a logical ring, and the token that 
renders the server able to execute operations in its local partition. In the static 
versions of the algorithms, when the servers consume all their predefined space for 
the data structure, the global (implemented) data structure is considered full, and the token server 
was sending \NACK\ to clients to notify them of this event.

In the dynamic version, though, there is no upper bound to the number of elements 
that can be stored in the data structure. In order to modify the static version of
the structures of this section, we remove the mechanism that sends \NACK\ messages 
to clients. Instead, every time a server $s$ receives the token (regarding inserts), it allocates 
an additional chunk of memory for its local partition. Because of this circular 
movement of the token, the elements are stored along a spiral path, that spans over 
all servers. Each chunk is marked with 
a sequence number, associated with the coil of the spiral, to distinguish the order 
of allocation.

An example of the transformation of a static algorithm to a dynamic is the dynamic 
version of the queue algorithm, presented in Algorithm~\ref{alg30-d}. In this design
a server $s$ uses two tokens, the head and tail token. In analogy, $s$ deploys two 
variables ($tail\_round$ and $head\_round$) to count the times the 
tokens have come to its possession. When a server $s$ receives an \ENQ\ message (line 
\lref{ln:alg30d-1}) but has no space left to store the element (line \lref{ln:alg30d-2}), 
it forwards the request along with the token to the next server in the ring. 
Afterwards, $s$ increases by one the variable $tail\_round$ and allocates a new 
memory chunk, by calling \texttt{allocate\_new\_space()}, to be used during the next time 
the token comes to its possession.  

For the \DEQ\ operation, the server performs additional actions concerning the 
empty queue state (line \lref{ln:alg30d-5}), where after responding with a \NACK, 
it re-initializes $tail\_round$ and $head\_round$ to be equal to zero
(line \lref{ln:alg30d-6}). An empty queue implies that the allocated chunks for
\texttt{lqueue} are also empty, hence they can be recycled and be used again anew. 
During the head token transition, $s$ increases $head\_round$ by one
chunk (line \lref{ln:alg30d-7}), so that when the head token comes to its possession 
to dequeue from the next memory. 

The double ended queue (deque) algorithm is going to work verbatim after these 
modifications. For the stack implementation the modifications are analogous. In this 
design there is one token, hence each server associates one counter with the token 
rounds. Each time a server $s$ moves the token to another server because the local 
stack is full, it increases the counter and allocates a new chunk for future use, and 
$s$ moves the token due to an empty stack, the counter is decreased by one.  
However, the dynamic design for the stack would introduce the 
termination problem described for queues. Nevertheless, the problem can be solved by 
applying the same technique of using client arrays as we did to solve the problem in the 
queue implementation.

\section{Distributed Lists}
\label{app:list}
A {\em list} is an ordered collection of elements.
It can either be {\em sorted}, in which case 
the elements appear in the list in increasing (or decreasing) order of their keys, 
or {\em unsorted}, in which case the elements appear  in the list in some arbitrary order 
(e.g. in the order of their insertion).
A list $L$  supports the operations {\em Insert}, {\em Delete}, and {\em Search}. 
Operation {\em Insert(L, k, I)} inserts an element with key $k$ and associated info $I$ to 
$L$. Operation {\em Delete(L, k)} removes the element with key $k$ from $L$ (if it exists), 
while operation {\em Search(L, k)} detects whether an element with key $k$ is present in 
$L$ and returns the information $I$ that is associated with $k$. 


In this section, we first provide an implementation of an unsorted distributed list 
in which we follow a token-based approach for implementing $Insert$. In this implementation, 
Search and Delete are highly parallel.
We then build on this approach in order to get a distributed implementation of a sorted list. 

\subsection{Unsorted List}
\label{app:unsorted-list}
The list state is stored distributedly in the local memories of several of the available servers, 
potentially spreading among all of them, if its size is large enough. The proposed implementation  
follows a token-based approach for implementing insert. Thus, we assume that the servers are arranged 
on a logical ring, based on their ids.

\begin{wrapfigure}{L}{0.7\textwidth}
	\begin{minipage}{.65\textwidth}
		\begin{algorithm}[H]
			\small
			\caption{Events triggered in a server of the distributed unsorted list.}
			\label{fig:dl-server}
			\begin{code}
				\lreset
				\firstline
				List $llist$ = $\varnothing$; \nl
				\integer\ $my\_id$, $next\_id$, $token$ = $0$, $round$ = $0$; \bl\bl\nl
				
				a message $\langle op, cid, key, data, mloop, tk\rangle$ is received:       \nl
				\n   \Switch\ ($op$) $\lbrace$                                                 \nl
				\n     \Case\ \INSERT:                                   \label{code:ser:ins1} \nl
				\n       \If\ ($tk$ == \TOKEN) $\lbrace$                 \label{code:ser:tok1} \nl
				\n         $token$ = $my\_id$;                           \label{code:ser:tok2} \nl
				allocate\_new\_memory\_chunk($llist, round$); \label{code:ser:tok3} \ul
				\p       $\rbrace$                                                             \nl
				$status_1$ = search($llist$, $key$);    \label{code:ser:ins2} \nl 
				\If\ ($status_1$) send($cid$, \NACK);           \label{code:ser:ins3} \nl 
				\Else\ $\lbrace$                                \label{code:ser:ins4} \nl 
				\n         \If\ ($token \neq my\_id$) $\lbrace$          \label{code:ser:ins5} \nl 
				\n           $next\_id$ = get\_next($my\_id$);           \label{code:ser:ins7} \nl 
				\If\ ($my\_id  \neq \maxser-1$) $\lbrace$ 
				\label{code:ser:ins6} \nl 
				\n               send($next\_id$, $\langle op, cid, key, data, mloop, tk\rangle$);
				\label{code:ser:ins8} \nl 
				\p           $\rbrace$ \Else\ 
				\n             send($next\_id$, $\langle op, cid, key, data, \true, tk\rangle$);
				\label{code:ser:ins9} \nl 
				\p\p       $\rbrace$ \Else\ $\lbrace$                   \label{code:ser:ins10} \nl 
				\n           \If\ (($my\_id \neq \maxser - 1$) AND ($round>0$) AND !($mloop$)) 
				$\lbrace$ \label{code:ser:ins11} \nl
				\n             $next\_id$ = get\_next($my\_id$);        \label{code:ser:ins12} \nl
				send($next\_id$, $\langle op, cid, key, data, mloop, tk\rangle$);
				\label{code:ser:ins13} \nl 
				\p           $\rbrace$ \Else\ $\lbrace$                 \label{code:ser:ins14} \nl 
				\n             $status_2$ = insert($llist, round, key, data$); 
				\label{code:ser:ins15} \nl 
				\If\ ($status_2$ == \false) $\lbrace$    \label{code:ser:ins16} \nl
				\n               $round++$;                             \label{code:ser:ins17} \nl 
				$token$ = get\_next($my\_id$);      	 \label{code:ser:ins18} \nl
				send($token$, $\langle op, cid, key, data, mloop, \TOKEN\rangle$);
				\label{code:ser:ins21} \nl
				\p             $\rbrace$ \Else\ send($cid$, \ACK);      \label{code:ser:ins22} \ul
				\p           $\rbrace$                                                         \ul
				\p         $\rbrace$                                                           \ul
				\p       $\rbrace$                                                             \nl
				\Break;                                                               \nl
				\p     \Case\ \SEARCH:                                  \label{code:ser:srch1} \nl
				\n       $status_1$ = search($llist$, $key$);   \label{code:ser:srch2} \nl 
				\If\ ($status_1$) 
				send($cid$, $\langle \ACK, my\_id \rangle$); \label{code:ser:srch4} \nl 
				\Else\   
				send($cid$, $\langle \NACK, my\_id \rangle$);\label{code:ser:srch6} \nl			
				\Break;                                        \label{code:ser:srch7} \nl
				\p     \Case\ \DELETE:                                   \label{code:ser:del1} \nl
				\n       $status_1$ = delete($llist$, $key$);      \label{code:ser:del2} \nl 
				\If\ ($status_1$)
				send($cid$, \ACK);                            \label{code:ser:del4} \nl 
				\Else\  
				send($cid$, \NACK);                           \label{code:ser:del6} \nl			
				\Break;                                         \label{code:ser:del7} \ul
				\p\p $\rbrace$
				\p
			\end{code}
		\end{algorithm} 
	\end{minipage}
\end{wrapfigure}

At each point in time, there is a server (not necessarily always the same), denoted by $s_t$, 
which holds the insert token, and serves insert operations. Initially, server $s_0$ has the 
token, thus the first element to be inserted in the list is stored on server $s_0$. Further 
element insertions are also performed on it, as long as the space it has allocated for the 
list does not exceed a threshold.
In case server $s_0$ has to service an insertion but its space is filled up, it 
forwards the token by sending a message to the next server, 
i.e. server $s_1$. Thus, if server $s_i$, $0 \leq i < \maxser$, has the token, but cannot 
service an insertion request without exceeding the threshold, it forwards the token to server 
$s_{(i+1) \mod \maxser}$. When the next server receives the token, it allocates a memory chunk 
of size equal to threshold, to store list elements. When the token reaches $s_{\maxser -1}$, if 
$s_{\maxser -1}$ has filled all the local space up to a threshold, it sends the token again to 
$s_0$. Then, $s_0$ allocates more memory (in addition to the memory chunk it had initially allocated 
for storing list elements) for storing more list elements. The token might go through the server 
sequence again without having any upper-bound restrictions concerning the number of round-trips. 
In order for a server to know whether the token has performed a round-trip on the ring, and hence 
all servers have stored list elements, it deploys a variable to count the number of ring round-trips 
it knows that the token has performed.

\begin{wrapfigure}{L}{0.65\textwidth}
	\begin{minipage}{.6\textwidth}
		\begin{algorithm}[H]
			\small
			\caption{Insert, Search and Delete operation for a client of the distributed list.}
			\label{fig:dl-client}
			\begin{code}
				\firstline 
				\bool\ ClientInsert(int $cid$, int $key$, data $data$) \{ \label{code:ins:0} \nl
				\n  \bool\ $status$;                                                         \bl \nl
				
				send($0$, $\langle\INSERT, cid, key, data, \false, -1\rangle$); \label{code:ins:1} \nl
				$status$ = receive();                                   \label{code:ins:2} \nl
				\return\ $status$;                                      \label{code:ins:3} \ul 
				\p \}
				\bl \nl
				\bool\ ClientSearch(int $cid$, int $key$) \{             \label{code:srch:0} \nl
				\n  \integer\ $sid$;                                                             \nl 
				\integer\ $c$ = $0$;                                                         \nl 
				\bool\ $status$;                                           			      \nl
				\bool\ $found$ = \false;                                                 \bl \nl
				
				send\_to\_all\_servers($\langle\SEARCH, cid, key, \bot, \false, -1\rangle$); \label{code:srch:1} \nl
				\Do\ \{                                                  \label{code:srch:2} \nl
				\n    $\langle status, sid \rangle$ = receive();             \label{code:srch:3} \nl 
				\If\ ($status$ == \ACK) $found$ = \true;               \label{code:srch:4} \nl
				$c++$;                                                 \label{code:srch:5} \nl	
				\p \} \While\ ($c < \maxser$);                               \label{code:srch:6} \nl
				\return\ $found$;                                         \label{code:srch:7} \ul 
				\p\}                                                                         \bl \nl
				
				\bool\ ClientDelete(int $cid$, int $key$) \{              \label{code:del:0} \nl
				\n  \integer\ $sid$;                                                             \nl 
				\integer\ $c$ = $0$;                                                         \nl 
				\bool\ $status$;                                           			      \nl
				\bool\ $deleted$ = \false;                                               \bl \nl
				
				send\_to\_all\_servers($\langle\DELETE, cid, key, \bot, \false, -1\rangle$);  
				\label{code:del:1} \nl
				\Do\ \{                                                    \label{code:del:2} \nl
				\n    $\langle status, sid \rangle$ = receive();               \label{code:del:3} \nl 
				\If\ ($status$ == \ACK)   $deleted$ = \true;             \label{code:del:4} \nl
				$c++$;                                                   \label{code:del:5} \nl	
				\p  \} \While\ ($c < \maxser$);                                \label{code:del:6} \nl
				\return\ $deleted$;                                         \label{code:del:7} \ul 
				\p \}
			\end{code}
		\end{algorithm}
	\end{minipage}
\end{wrapfigure}

Event-driven code for the server is presented in Algorithm \ref{fig:dl-server}. 
Each server $s$ maintains a local list ($llist$ variable) allocated 
for storing list elements, a $token$ variable which indicates whether $s$ 
currently holds the token, and a variable $round$ to mark the ring round-trips 
the token has performed; $round$ is initially $0$, and is incremented after 
every transmission of the token to the next server. 

Each message a server receives 
has five fields: (1) $op$ that denotes the operation to be executed, (2) 
$cid$ that holds the id of the client that initiated a request, (3) 
$key$ that holds the value to be inserted, (4) $mloop$ stands for 
``message loop'', a \bool\ value that denotes if the message has traversed the 
whole server sequence and (5) $tk$ that is set when a forwarded message 
also denotes a token transition from one server to the other. 

When a message is received, the server $s$ first checks its type. If the message 
is of type \INSERT\ (line \lref{code:ser:ins1}), $s$ first checks whether the 
message has the $tk$ field marked. If it is marked (line \lref{code:ser:tok1}), 
$s$ sets a local variable $token$ equal to its own id (line \lref{code:ser:tok2}) 
and allocates additional space for its local part of the list (line \lref{code:ser:tok3}). 

Afterwards, $s$ searches the part of the list 
that it stores locally, for an element with the same key ($key$ variable in 
the algorithm) as the one to be inserted (line \lref{code:ser:ins2}). Searching
$llist$ for the element has to be performed independently of whether the server holds the
token or not. Since this design does not permit duplicate entries, if such an 
element is found, the server responds with \NACK\ to the client (line 
\lref{code:ser:ins3}). Otherwise (line \lref{code:ser:ins4}), $s$ checks whether 
the new element can be stored in $llist$.

In case $s$ does not hold the token (line \lref{code:ser:ins5}), it is not allowed 
to perform an insertion, therefore it must forward the message to the next server 
in the ring. If $s$ is not $s_{\maxser -1}$  (line 
\lref{code:ser:ins6}), it forwards to the next server the request (\lref{code:ser:ins8}). 
In case $s$ is $s_{\maxser -1}$, it means that all servers have 
been searched for the element and the element was not found. Server $s$ 
sends the message to the next server (in order to eventually reach the token server), 
after marking the $mloop$ field 
of the message as \true, to indicate that the message has completed a full round-trip on the ring 
(line \lref{code:ser:ins9}).

On the other hand, if $s$ holds the token (line \lref{code:ser:ins10}), it must 
first check whether there is room in $llist$ to insert the element in it. 
If there is room in $llist$ and the local variable $round$ of $s$
equals to \false\ (which means that the list does not expand to the next servers)
or the message has already performed a round-trip on the ring, 
then $s$ inserts the element and returns \ACK. 
If however, $round > 0$ and the message has not performed a round trip
on the ring ($mloop==\false$), $s$ continues forwarding the message.  

If the token server's local memory is out of sufficient space (line 
\lref{code:ser:ins16}) (i.e. the \texttt{insert()} function was unsuccessful), 
$s$ forwards the message to the next server 
the $tk$ field with \TOKEN\ (line \lref{code:ser:ins21}) to indicate
that this server will become the new token server after $s$. Also, $s$
increments $round$ by one to count the number of times the token has passed from it. 
The $round$ variable is also used by function \texttt{allocate\_new\_memory\_chunk()} 
that allocates additional space for the list (line \lref{code:ser:tok3}). 

Notice that, contrary to other token-based implementations presented in previous 
sections, the token server of the unsorted list does not need to rely on client tables in 
order to stop a message from being incessantly forwarded from one server to another, without 
ever being served. By virtue of having clients always sending their insert requests to $s_0$, 
an insert request $r_j$ that arrives at $s_0$ before some other insert request $r_k$, is 
necessarily served before $r_k$. The scenario where insert requests constantly arrive at the 
token server before $r_j$, making the token travel to the next server before $r_j$ can be served, 
is thus avoided.

Upon receiving a \SEARCH\ request from a client (line \lref{code:ser:srch1}), a server 
searches for the requested element in its local list (line \lref{code:ser:srch2}) and 
sends \ACK\ to the server if the element is found (line \lref{code:ser:srch4}) and \NACK\ 
otherwise (line \lref{code:ser:srch6}). 

Upon receiving a \DELETE\ request from a client (line \lref{code:ser:del1}), a server 
attempts to delete the requested element from its local list (line \lref{code:ser:del2}) 
and sends \ACK\ to the server if the deletion was successful (line \lref{code:ser:del4}). 
Otherwise it sends \NACK\ (line \lref{code:ser:del6}).

The pseudocode of the client is presented in Algorithm \ref{fig:dl-client}. Notice 
that insert operations in the proposed implementation are executed in sequence and 
must necessarily pass through server $0$ and be forwarded through the server ring, 
if necessary due to space constraints. Search and Delete operations, on the contrary, 
are executed in parallel. 

In order to execute an insertion, a client calls \texttt{ClientInsert()}
(line \lref{code:ins:0}) which sends an \INSERT\ message (line \lref{code:ins:1}) to 
server $0$, regardless of which server holds the token in any given configuration, and 
then blocks waiting for a response (line \lref{code:ins:2}). If the client receives 
\ACK\ from a server, then the element was inserted correctly. If the client receives 
\NACK, then the insertion failed, due to either limited space, or the existence of 
another element with the same key value.

For a search operation the client calls \texttt{ClientSearch()} (line \lref{code:srch:0}). 
The client sends a \SEARCH\ request to all servers (line \lref{code:srch:1}) and 
waits to receive a response message (line \lref{code:srch:3}) from each server (\Do\ 
\While\ loop of lines \lref{code:srch:2}-\lref{code:srch:6}). The requested element is 
in the list if the client receives \ACK\ from some server (line \lref{code:srch:4}). 
A delete operation proceeds similarly to \texttt{ClientSearch()}. It is initiated by 
a client by sending a \DELETE\ request to all servers (line \lref{code:del:1}). 
The client then waits to receive 
a response message (line \lref{code:del:3}) from each server (\Do\ \While\ loop of 
lines~\lref{code:del:2}-\lref{code:del:6}). The requested element has been found in 
the list of some client and deleted from there, if the client receives \ACK\ from 
some server $s$.

\subsubsection{Proof of Correctness}
We sketch the correctness argument for the proposed implementation by providing 
linearization points.  Let $\alpha$ be an execution of the distributed unsorted 
list algorithm presented in Algorithms \ref{fig:dl-server} and \ref{fig:dl-client}. 
We assign linearization points to insert, delete and search operations in $\alpha$ 
as follows: 
\begin{compactitem}
\item {\em Insert.} Let $op$ be any instance of {\tt ClienInsert} for which an \ACK\ or 
a \NACK\ message is sent by a token server. Then, if \ACK\ is sent by a token server for 
$op$ (line \lref{code:ser:ins22}), the linearization point is placed in the configuration 
resulting from the execution of line \lref{code:ser:ins15} that successfully inserted the 
required element into the server's local list. If \NACK\ is sent for $op$ (line \lref{code:ser:ins3}), 
then the linearization point is placed in the configuration resulting from the execution 
of line \lref{code:ser:ins2}, where the search operation on the local list of the server 
returned \true. 

\item Let $op$ be any instance of {\tt ClientDelete} for which an \ACK\ or a \NACK\ 
message is sent by a server. Then, if \ACK\ is sent by a server $s$ for $op$, the 
linearization point is  placed in the configuration resulting from the execution of 
line \lref{code:ser:del2} by the server that sent the \ACK. Otherwise, if the key $k$ 
that $op$ had to delete was not present in any of the local lists of the servers in the 
beginning of the execution interval of $op$, then the linearization point of $op$ is 
placed at the beginning of its execution interval. Otherwise, if $k$ was present but was 
deleted by a concurrent instance $op'$ of {\tt ClientDelete}, then the linearization 
point is placed right after the linearization point of $op'$.

\item Let $op$ be any instance of {\tt ClientSearch} for which an \ACK\ or a \NACK\ 
message is sent by a server. Then, if \ACK\ is sent by a server $s$ for $op$, the 
linearization point is  placed in the configuration resulting from the execution of 
line \lref{code:ser:srch2} by the server that sent the \ACK. Otherwise, if the key $k$ 
that $op$ had to find was not present in the list in the beginning of its execution 
interval, then the linearization point is placed there. Otherwise, if $k$ was present 
but was deleted by a concurrent instance $op'$ of {\tt ClientDelete}, then the 
linearization point is placed right after the linearization point of $op'$.
\end{compactitem}

\begin{lemma}
\label{lemma:list-lin}
Let $op$ be any instance of an insert, delete, or a search operation executed 
by some client $c$ in $\alpha$. Then, the linearization point of $op$ is placed 
in its execution interval.
\end{lemma}

\begin{proof}
Let $op$ be an instance of an insert operation invoked by client $c$. 
A message with the insert request is sent on line \lref{code:ins:1}, after 
the invocation of the operation. Recall that routine {\tt receive()} blocks 
until a message is received. 
Notice that both line \lref{code:ser:ins15} as well as line \lref{code:ser:ins2} 
are executed by a server before it sends a message to the client. 
Therefore, whether $op$ is linearized at the 
point some server sends it a message on line \lref{code:ser:ins22} or on line 
\lref{code:ser:ins3}, it terminates only after receiving it. 
Notice also that the operation terminates only after the client receives it.
Thus, the linearization point is included in its execution interval. 

By similar reasoning, if $op$ is an instance of a delete operation that 
is linearized in the configuration resulting from the execution of line 
\lref{code:ser:del2} or a search operation that is linearized in the configuration 
resulting from the execution of line \lref{code:ser:srch2}, then the linearization 
point is included in the execution interval of $op$. 

Let $op$ be an instance of a delete operation that deletes key $k$ and that 
terminates after receiving only \NACK\ messages on line \lref{code:del:3}. If 
$k$ is not present in the list in the beginning of the execution interval of 
$op$, then $op$ is linearized at that point and the claim holds. 

Consider the case where $k$ is included in the list when $op$ is invoked. 
By observation of the pseudocode (lines \lref{code:ser:del1}-\lref{code:ser:del7}), 
we have that when a server receives a delete request by a client, it traverses 
its local part of the list and deletes the element with key equal to $k$ (line 
\lref{code:ser:del2}), if it is included in it. By further observation of the 
pseudocode (lines \lref{code:del:1}-\lref{code:del:7}), we have that after $c$ 
invokes $op$, it sends a delete request to all servers (line \lref{code:del:1}) 
and then awaits for a response from all of them (\Do\ \While\ loop of lines 
\lref{code:del:2}-\lref{code:del:6}). 
By assumption, all servers responds with \NACK. Notice that this implies that between 
the execution of line \lref{code:del:3} and \lref{code:del:5} the element with key $k$ 
is removed from the local list of $s$ because of some other concurrent 
delete operation $op'$ invoked by some client $c'$. By scrutiny of the pseudocode, 
we have that a server that deletes an element from its local list, does so on line 
\lref{code:ser:del2}, which occurs before the server sends a response to the delete 
request. By definition, then, $op'$ is linearized at the point $s$ executes line 
\lref{code:ser:del2}, before it sends an \ACK\ message to $c'$. Since $op'$ causes 
the element with key $k$ to be removed from the local list of $s$ between the 
execution of lines \lref{code:del:3} and \lref{code:del:5} by $c$, its 
linearization point is included in the execution interval of $op$. Since 
we place the linearization point of $op$ right after the linearization point 
of $op'$, the claim holds. 

The argument is similar for when $op$ is an instance of a search operation 
for key $k$ that terminates after receiving a \NACK\ message from all the 
servers on lines \lref{code:srch:2}-\lref{code:srch:6}.
\end{proof}

Each server maintains a local variable $token$ with initial value $0$. 
Let some server $s$ receive a message $m$ in some configuration $C$. If the field 
{\tt tk} of $m$ is equal to \TOKEN, we say that {\em receives a token message}. 
Observe that when $s$ receives a token message (line \lref{ln:60}), the value of 
$token$ is set to $s$. Furthermore, when $s$ executes line \lref{code:ser:ins18}, 
where the value of $token$ changes from $s$ to $s+1$, $s$ also sends a token message 
to $s+1$ (line \lref{code:ser:ins21}). Notice that $s$ can only reach and execute 
this line if the condition of the \If\ clause of line \lref{code:ser:ins5} evaluates 
to \false, i.e. if $token$ = $s$. Then, the following holds:

\begin{observation}
\label{obs:ul-unique}
At each configuration in $\alpha$, there is at most one server $s$ for 
which the local variable $token$ has the value $s$.
\end{observation}

This server is referred to as {\em token server}. By the pseudocode, namely 
the \If\ \Else\ clause of lines \lref{code:ser:ins5}, \lref{code:ser:ins10}, 
and by line \lref{code:ser:ins15}, the following observations holds.

\begin{observation}
\label{obs:ul-tokenserver}
A server $s$ performs insert operations on its local list in $\alpha$ only 
during those subsequences of $\alpha$ in which it is the token server. 
\end{observation}

Each server maintains a local list collection, $llist$. By observation of the 
pseudocode, lines \lref{code:ser:ins2} and \lref{code:ser:ins3}, we have that 
if an insert operation attempts to insert key $k$ in either of the lists of 
a server $s$, but an element with that key already exists, then no second 
element for $k$ is inserted and the operation terminates. Thus, the following 
holds: 

\begin{observation}
\label{obs:ul-llset}
The keys contained in the list collection of $s$ in any configuration 
$C$ of $\alpha$ form a set. 
\end{observation}

We denote this set by $ll^s$. By scrutiny of the pseudocode, we see that 
a new list object is allocated in $llist$ each time a server receives a 
token message (lines \lref{code:ser:tok1}-\lref{code:ser:tok3}). The new 
object is identified by the value of local variable $round$. By observation 
of the pseudocode, we further have that each time a server inserts a key 
into $ll^s$, it does so on the list object identified by $round$ (line 
\lref{code:ser:ins15}). We refer to this object as {\em current list object}. 
Then, based on lines \lref{code:ser:ins16}-\lref{code:ser:ins21} we have the 
following:

\begin{observation}
\label{obs:ul-direction}
A token message is sent from a server $s$ to a server $((s+1) \mod \maxser)$
in some configuration $C$ only if the current local list object of server $s$ 
is full at $C$.
\end{observation}

Further inspection of the pseudocode shows that the local list object of 
a server is only accessed by the execution of line \lref{code:ser:ins2}, 
\lref{code:ser:ins15}, \lref{code:ser:srch2}, or \lref{code:ser:del2}. 
From this, we have the following observation.

\begin{observation}
\label{obs:ul-mod}
If an operation $op$ modifies the local list object of some server, 
then this occurs in the configuration in which $op$ is linearized. 
\end{observation}

Let $C_i$ be the configuration in which the $i$-th linearization point 
in $\alpha$ is placed. Denote by $\alpha_i$, the prefix of $\alpha$ which 
ends just after $C_i$ and let $L_i$ be the sequence of linearization points 
that is defined by $\alpha_i$. Denote by $S_i$ the set of keys that a 
sequential list contains after applying the sequence of operations that $L_i$ 
imposes. Denote by $S_i = \epsilon$ the empty sequence (the list is empty).


\begin{lemma}
\label{obs:not-in-list}
Let $k$ be the token server in some configuration $C$ in which it receives a 
message $m$ for an insert operation $op$ with key $k$ invoked by client $c$. 
Then at $C$, no element with key $k$ is contained in the local list set of 
any other server $s \neq k$.
\end{lemma}

\begin{proof}
By inspection of the pseudocode, when a client $c$ sends a message $m$ 
to some server either on line \lref{code:ins:1}, line \lref{code:srch:1}, 
line \lref{code:del:1}, or line \lref{code:del:5}, the $mloop$ field of 
$m$ is equal to \false. This field is set to \true\ when server $s_{\maxser-1}$ 
executes line \lref{code:ser:ins9}. Notice that in the configuration in 
which this line is executed by $s_{\maxser-1}$, it is not the token server 
(otherwise the condition of line \lref{code:ser:ins5} would not evaluate to 
\true\ and the line would not be executed).   


Consider the case where $m$ reaches a server $s$ at some configuration $C$ 
and let $ll^s$ contain an element with key $k$ in $C$. By inspection of the 
pseudocode (lines \lref{code:ser:ins2}-\lref{code:ser:ins3}) we have that in 
that case, $m$ is not forwarded to a subsequent server.

Furthermore, by lines \lref{code:ser:ins5}-\lref{code:ser:ins9}, we have that 
if $s$ is not the token server and not $s_{\maxser-1}$, and provided that $ll^s$ 
does not contain an element with key $k$, then $s$ forwards $m$ without modifying 
the $mloop$ field. This implies that the $mloop$ field of $m$ is changed at most 
once in $\alpha$ from \false\ to \true, and that by server $\maxser-1$, in a 
configuration $C'$ in which $k$ is not contained in  $ll^{\maxser-1}$.
\end{proof}

\begin{lemma}
Let $C_i$, $i \geq 0$, be a configuration in $\alpha$ in which server $s_{t_i}$ is the 
token server. Let $ll_i^j$ be the local list set of server $s_j$, $0 \leq j < \maxser$, 
in $C_i$. Then it holds that $S_i = \bigcup_{j=0}^{\maxser-1}ll_i^j$. 
\end{lemma}

\begin{proof}
We prove the claim by induction on $i$. 

{\bf Base case (i = 0).} 
The claim holds trivially at $C_0$.

{\bf Hypothesis.} Fix any $i > 0$ and assume that at $C_i$, it holds that 
$S_i$ = $\bigcup _{j=0}^{\maxser-1}ll_i^j$. We show that the claim holds 
for $i+1$.

{\bf Induction step.} 
Let $op_{i+1}$ be the operation that corresponds to the linearization 
point placed in $C_{i+1}$. We proceed by case study. 

Let $op_{i+1}$ be an insert operation for key $k$. Assume first that the linearization 
point of $op_{i+1}$ is placed at the execution of line \lref{code:ser:ins2} by $s_{t_{i+1}}$ 
for it. Notice that when this line is executed, $k$ is searched for in the local 
list of $s_{t_{i+1}}$. Recall that, by the way linearization points are assigned, 
the client $c$ that invoked $op_{i+1}$ receives \NACK\ as response. Notice also that 
$s_{t_{i+1}}$ sends \NACK\ as a response to $c$ if $k$ is present in the local list 
of $s_{t_{i+1}}$, and thus $status_1 = \true$. In that case, lines \lref{code:ser:ins5} 
to \lref{code:ser:ins22} are not executed, and therefore, no new element is inserted 
into the local list of $s_{t_{i+1}}$ (line \lref{code:ser:ins15}). 
Thus $ll_{i+1}^{s_{t_{i+1}}} = ll_{i}^{s_{t_{i+1}}}$.
By the induction hypothesis, $S_i$ = $\bigcup _{j=0}^{\maxser-1}ll_i^j$.
By Observation \ref{obs:ul-mod} it follows that for any other server $s_j$, 
where $j \neq t_{i+1}$, $ll_{i+1}^{s_j} = ll_{i}^{s_j}$ as well. Then, 
$\bigcup_{j=0}^{\maxser-1}ll_{i+1}^j = \bigcup_{j=0}^{\maxser-1}ll_i^j$. 
Notice that since the server responds with \NACK, $S_{i+1} = S_i$ by definition.
Thus, $S_{i+1} = \bigcup_{j=0}^{\maxser-1}ll_{i+1}^j$ and the claim holds. 

Now, assume that $op_{i+1}$ is linearized at the execution of line \lref{code:ser:ins15} 
by the token server for it. By the way linearization points are assigned, this implies 
that when this line is executed, $status_2 = \true$, and the insertion of an element 
with key $k$ into the local list of $s_t$ was successful. This in turn implies that 
at $C_{i+1}$, $ll_{i+1}^{s_t} = ll_{i}^{s_t}\cup \{k\}$. 
By Observation \ref{obs:ul-mod} it follows that for any other server $s_j$, 
where $j \neq t_{i+1}$, $ll_{i+1}^{s_j} = ll_{i}^{s_j}$ as well.
Notice that since the server responds with \ACK, by definition the insertion is 
successful and thus $S_{i+1} = S_i \cup \{k\}$. Since by the induction hypothesis, 
$S_i$ = $\bigcup _{j=0}^{\maxser-1}ll_i^j$, it holds that 
$S_{i+1} = \bigcup _{j=0}^{\maxser-1}ll_i^j \cup \{k\} = \bigcup _{j=0}^{\maxser-1}ll_{i+1}^j$, 
thus, the claim holds.

Now consider that $op_{i+1}$ is a delete operation for key $k$. Assume first that 
some server $s_d$ responds with \ACK, by executing line \lref{code:ser:del4}, to the 
client $c$ that invoked $op_{i+1}$. Then $op_{i+1}$ is linearized at the execution 
of this line by $s_d$. Notice that this line is executed by a server if $status_1 = \true$, 
i.e. if the server was successful in locating and deleting an element with key 
$k$ from its local list. Thus, $ll_{i+1}^{s_t} = ll_{i}^{s_t} \setminus \{k\}$. 
Furthermore, by definition, $S_{i+1} = S_i \setminus \{k\}$. By the induction 
hypothesis, $S_i$ = $\bigcup _{j=0}^{\maxser-1}ll_i^j$ and since by Observation 
\ref{obs:ul-mod} no other modification occurred on the local list of some other 
server between $C_i$ and $C_{i+1}$, it follows that 
$S_{i+1} = S_i \setminus\ \{k\} = \bigcup _{j=0}^{\maxser-1}ll_i^j \setminus \{k\} = \bigcup _{j=0}^{\maxser-1}ll_{i+1}^j$.

Assume now that $op_{i+1}$ is a delete operation for which no server responds with \ACK\ 
to the invoking client. Recall that in this case, by definition, $S_{i+1} = S_i$. 
By inspection of the pseudocode, it follows that no server finds 
an element with key $k$ in its local list when it is executing line \lref{code:ser:del2} 
for $op_{i+1}$. We examine two cases: (i) either no element with key $k$ is contained 
in any local list of any server in the beginning of the execution interval of $op_{i+1}$,
or (ii) an element with key $k$ is contained in the local list of some server $s_d$ in 
the beginning of $op_{i+1}$'s execution interval, but $s_d$ deletes it while serving a 
different delete operation $op'$, before it executes line \lref{code:ser:del2} for 
$op_{i+1}$.

Assume that case (i) holds. Then, the linearization point is placed in the beginning 
of the execution interval of $op_{i+1}$. Notice that in this case, the invocation (nor 
in fact the further execution) of $op_{i+1}$ has no effect on the local list of any 
server. Thus, between $C_i$ and $C_{i+1}$ no server local list is modified and, by the 
induction hypothesis, the claim holds. 


Assume now that case (ii) holds. By Lemma \ref{lemma:list-lin}, we have that a concurrent 
delete operation $op'$ removes the element with key $k$ from the local list of $s_d$ during 
the execution interval of $op_{i+1}$. By the assignment of linearization points, Observation 
\ref{obs:ul-mod} and Lemma \ref{lemma:list-lin}, it further follows that $op' = op_i$. Notice 
that in this case (ii) also, $op_{i+1}$ has no effect on the local list of any server. 
Thus, since by the induction hypothesis it holds that $S_i = \bigcup_{j=0}^{\maxser-1}ll_i^j$, 
it also holds that $S_i = \bigcup_{j=0}^{\maxser-1}ll_{i+1}^j$, and since $S_i = S_{i+1}$, 
the claim holds. 

Since a search operation does not modify the local list of any server, the argument is 
analogous as for the case of the delete operation. 
\end{proof}

From the above lemmas and observations, we have the following.

\begin{theorem}
The distributed unsorted list is linearizable. The insert operation has time and communication complexity $O(\maxser)$. The search and delete operations have communication complexity $O(1)$.
\end{theorem}

\ignore{
Using a proof technique similar to that followed for the token-based stack implementation 
(Section \ref{app:stack-token}), we can prove:
\begin{theorem}
The distributed unsorted list is linearizable. The insert operation has time and communication complexity $O(\maxser)$. The search and delete operations have communication complexity $O(1)$.
\end{theorem}
}

\subsubsection{Alternative Implementation}
\label{app:unsorted-list-alt}


At each point in time, there is a server (not necessarily always the same), 
denoted by $s_t$, which holds the insert token, and serves insert operations. 
Initially, server $s_0$ has the token, thus the first element to be inserted 
in the list is stored on server $s_0$. Further element insertions are also performed on it, 
as long as the space it has allocated for the 
list does not exceed a threshold.
In case server $s_0$ has to service an insertion but its space is filled up, it 
forwards the token by sending a message to the next server, 
i.e. server $s_1$. Thus, if server $s_i$, $0 \leq i < \maxser$, has the token, but cannot 
service an insertion request without exceeding the threshold, it forwards the token to server 
$s_{(i+1) \mod \maxser}$. When the next server receives the token, it allocates a memory chunk 
of size equal to threshold, to store list elements. When the token reaches $s_{\maxser -1}$, if 
$s_{\maxser -1}$ has filled all the local space up to a threshold, it sends the token again to 
$s_0$. Then, $s_0$ allocates more memory (in addition to the memory chunk it had initially allocated 
for storing list elements) for storing more list elements. The token might go through the server 
sequence again without having any upper-bound restrictions concerning the number of round-trips. 

Event-driven code for the server is presented in Algorithm \ref{fig:dl-server-alt}. 
Each server $s$ maintains a local list ($llist$ variable) allocated 
for storing list elements, a $token$ variable which indicates whether $s$ 
currently holds the token, and a variable $round$ to mark the ring round-trips 
the token has performed; $round$ is initially $0$, and is incremented after 
every transmission of the token to the next server. 
The pseudocode of the client is presented in Algorithm~\ref{fig:dl-client-alt}. 

A client $c$ sends an insert request for an element with key $k$ to all 
servers in parallel and awaits a response. If any of the servers contains 
$k$ in its local list, it sends \ACK\ to $c$ and the insert operation 
terminates. If no server finds $k$, then all reply \NACK\ to $c$. 
In addition, the token server $s_t$ encapsulates its id in the \NACK\ reply.
After that, $c$ sends an insert request for $k$ to $s_t$ only. If $s_t$ 
can insert it, it replies \ACK\ to $c$. If $k$ has in the meanwhile been 
inserted, $s_t$ replies \NACK\ to $c$. If $s_t$ is no longer the token 
server, it forwards the request along the server ring until it reaches 
the current token server. Servers along the ring should check whether 
they contain $k$ or not, and if some server does, then it replies \NACK\ 
to $c$. Let $s_t'$ be a token server that receives such a request. It 
also checks whether it contains $k$ or not. If not, it attempts to insert 
$k$ into its local list. Otherwise it replies \NACK. When attempting to 
insert the element in the local list, it may occur that the allocated 
space does not suffice. In this case, the server forwards the request 
as well as the token to the next server in the ring, and increments the 
value of $round$ variable. If the insertion at a token server is successful, 
the server then replies \ACK\ to $c$.

To perform a search for an element $e$, a client $c$ sends a search request 
to all servers and awaits their responses. A server $s$ that receives a search request, 
checks whether $e$ is present in its local part of the list and if so, it responds with \ACK\ to 
$c$. Otherwise, the response is \NACK. If all responses that $c$ receives are \NACK, $e$ is not 
present in the list. Notice that if $e$ is contained in the list, exactly one server responds with \ACK. 
Delete works similarly; if a server $s$ responds with \ACK, then $s$ has found and deleted $e$ from its local list. 
Given that communication is fast 
and the number of servers is much less than the total number of cores, 
forwarding a request to all servers does not flood the network. 

\begin{wrapfigure}{L}{0.55\textwidth}
	\begin{minipage}{.5\textwidth}
		\begin{algorithm}[H]
\small
\caption{Events triggered in a server of the distributed unsorted list.}
\label{fig:dl-server-alt}
\begin{code}
\lreset
\firstline
List $llist$ = $\varnothing$; \nl
\integer\ $my\_id$, $next\_id$, $token$ = $0$, $round$ = $0$; \bl\bl\nl

    a message $\langle op, cid, key, data,tk\rangle$ is received:       \nl
 \n   \Switch\ ($op$) $\lbrace$                                                 \nl
 \n     \Case\ \INSERT:                                   \label{code:ser:ins1-alt} \nl
 \n       \If\ ($tk$ == \TOKEN) $\lbrace$                 \label{code:ser:tok1-alt} \nl
 \n         $token$ = $my\_id$;                           \label{code:ser:tok2-alt} \nl
            allocate\_new\_memory\_chunk($llist, round$); \label{code:ser:tok3-alt} \ul
 \p       $\rbrace$                                                             \nl
          $status_1$ = search($llist$, $key$);    		  \label{code:ser:ins2-alt} \nl 
          \If\ ($tk$ == $-2$) $\lbrace$ 					 \label{code:ser:ins2.1}\nl 
\n          \If\ ($status_1$) $\lbrace$               \label{code:ser:tok1.1-alt} \nl
\n          	\If\ ($token == my\_id$) send($cid$, $\langle\ACK, \true\rangle$); \label{code:ser:ins3-alt} \nl 
				\Else\ send($cid$, $\langle\ACK, \false\rangle$);           		\label{code:ser:ins3.1} \nl 
\p			$\rbrace$ \Else\  $\lbrace$               \label{code:ser:tok1.1} \nl
\n          	\If\ ($token == my\_id$) send($cid$, $\langle\NACK, \true\rangle$); \label{code:ser:ins33} \nl 
				\Else\ send($cid$, $\langle\NACK, \false\rangle$);           		\label{code:ser:ins3.13} \ul 
\p			$\rbrace$ \nl
\p        $\rbrace$ 
		  \Else\ $\lbrace$                                \label{code:ser:ins4-alt} \nl 
\n 			\If\ ($status_1$) send($cid$, \NACK);      \label{code:ser:ins22.1} \nl 
			\Else\ $\lbrace$ \nl
 \n         \If\ ($token \neq my\_id$) $\lbrace$          \label{code:ser:ins5-alt} \nl 
 \n           $next\_id$ = get\_next($my\_id$);           \label{code:ser:ins7-alt} \nl 
  send($next\_id$, $\langle op, cid, key, data, tk\rangle$);
                                                          \label{code:ser:ins8-alt} \nl 
 \p       $\rbrace$ \Else\ $\lbrace$                   \label{code:ser:ins10-alt} \nl 
 \n             $status_2$ = insert($llist, round, key, data$); 
 														 \label{code:ser:ins15-alt} \nl 
                \If\ ($status_2$ == \false) $\lbrace$    \label{code:ser:ins16-alt} \nl
 \n               $round++$;                             \label{code:ser:ins17-alt} \nl 
                  $token$ = get\_next($my\_id$);      	 \label{code:ser:ins18-alt} \nl
                  send($token$, $\langle op, cid, key, data, \TOKEN\rangle$);
                                                         \label{code:ser:ins21-alt} \nl
 \p             $\rbrace$ \Else\ send($cid$, \ACK);      \label{code:ser:ins22-alt} \ul
 \p           $\rbrace$                                                         \ul
 \p         $\rbrace$                                                           \ul
 \p       $\rbrace$                                                             \nl
          \Break;                                                               \nl
 \p     \Case\ \SEARCH:                                  \label{code:ser:srch1-alt} \nl
 \n       $status_1$ = search($llist$, $key$);   \label{code:ser:srch2-alt} \nl 
          \If\ ($status_1$) 
            send($cid$, $\langle \ACK, my\_id \rangle$); \label{code:ser:srch4-alt} \nl 
          \Else\   
            send($cid$, $\langle \NACK, my\_id \rangle$);\label{code:ser:srch6-alt} \nl			
          \Break;                                        \label{code:ser:srch7-alt} \nl
 \p     \Case\ \DELETE:                                   \label{code:ser:del1-alt} \nl
 \n       $status_1$ = delete($llist$, $key$);      \label{code:ser:del2-alt} \nl 
          \If\ ($status_1$)
            send($cid$, \ACK);                            \label{code:ser:del4-alt} \nl 
          \Else\  
            send($cid$, \NACK);                           \label{code:ser:del6-alt} \nl			
          \Break;                                         \label{code:ser:del7-alt} \ul
 \p $\rbrace$
\p
\end{code}
\end{algorithm} 
\end{minipage}
\end{wrapfigure}

\begin{wrapfigure}{L}{0.55\textwidth}
	\begin{minipage}{.5\textwidth}
		\begin{algorithm}[H]
\small
\caption{Insert, Search and Delete operation for a client of the distributed list.}
\label{fig:dl-client-alt}
\begin{code}
\firstline 
    \bool\ ClientInsert(int $cid$, int $key$, data $data$) \{ \label{code:ins:0-alt} \nl
 \n  \bool\ $status$;												\nl
 	  \bool\ $found$ = \false;                                      \nl
 	 \integer\ $tid$; 												\bl \nl

       send\_to\_all\_servers($\langle\INSERT, cid, key, \bot, -2\rangle$); \label{code:ins:1-alt} \nl
       \Do\ \{                                                  \label{code:ins:2-alt} \nl
\n    		$\langle status, sid, is\_token \rangle$ = receive();   \label{code:ins:3-alt} \nl 
            \If\ ($status$ == \ACK) $found$ = \true;               \label{code:ins:4-alt} \nl
            \If\ ($is\_token$) $tid = sid$;						 \label{code:ins:4.1-alt} \nl
            $c++$;                                                 \label{code:ins:5} \nl	
\p 	   \} \While\ ($c < \maxser$);                               \label{code:ins:6} \nl
       \If\ ($found$ == \true) \return\ \false;						 \label{code:ins:4.1} \nl
       send($tid$, $\langle\INSERT, cid, key, data, -1\rangle$); \label{code:ins:7} \nl
       $status$ = receive();                                   \label{code:ins:8} \nl
       \If\ ($status$ == \NACK) \return\ \false;              \label{code:ins:9} \nl
       \Else\ \return\ \true;              					\label{code:ins:10} \ul  
 \p \}
                                                                              \bl \nl
     \bool\ ClientSearch(int $cid$, int $key$) \{             \label{code:srch:0-alt} \nl
 \n  \integer\ $sid$;                                                             \nl 
     \integer\ $c$ = $0$;                                                         \nl 
     \bool\ $status$;                                           			      \nl
     \bool\ $found$ = \false;                                                 \bl \nl
	
     send\_to\_all\_servers($\langle\SEARCH, cid, key, \bot, -1\rangle$); \label{code:srch:1-alt} \nl
     \Do\ \{                                                  \label{code:srch:2-alt} \nl
 \n    $\langle status, sid \rangle$ = receive();             \label{code:srch:3-alt} \nl 
       \If\ ($status$ == \ACK) $found$ = \true;               \label{code:srch:4-alt} \nl
       $c++$;                                                 \label{code:srch:5-alt} \nl	
 \p \} \While\ ($c < \maxser$);                               \label{code:srch:6-alt} \nl
	\return\ $found$;                                         \label{code:srch:7-alt} \ul 
 \p\}                                                                         \bl \nl

    \bool\ ClientDelete(int $cid$, int $key$) \{              \label{code:del:0-alt} \nl
 \n  \integer\ $sid$;                                                             \nl 
     \integer\ $c$ = $0$;                                                         \nl 
     \bool\ $status$;                                           			      \nl
     \bool\ $deleted$ = \false;                                               \bl \nl
	
     send\_to\_all\_servers($\langle\DELETE, cid, key, \bot, -1\rangle$);  
     															\label{code:del:1-alt} \nl
     \Do\ \{                                                    \label{code:del:2-alt} \nl
 \n    $\langle status, sid \rangle$ = receive();               \label{code:del:3-alt} \nl 
       \If\ ($status$ == \ACK)   $deleted$ = \true;             \label{code:del:4-alt} \nl
       $c++$;                                                   \label{code:del:5-alt} \nl	
 \p  \} \While\ ($c < \maxser$);                                \label{code:del:6-alt} \nl
	\return\ $deleted$;                                         \label{code:del:7-alt} \ul 
 \p \}
\end{code}
\end{algorithm}
\end{minipage}
\end{wrapfigure}
\subsection{Sorted List}
\label{app:sorted-list}
\begin{wrapfigure}{l}{0.7\textwidth}
	\begin{minipage}{.65\textwidth}
		\begin{algorithm}[H]
			\scriptsize
			\caption{Events triggered in a server of the distributed sorted list.}
			\label{fig:ol-server}
			\begin{code}
				\firstline
				List $llist = \varnothing$;  																		\nl
				\integer\ $my\_id$, $next\_id$, $k_{max}$, $cv[MC]$, $nbr\_cv[MC]$;									\nl
				data[$0 \ldots CHUNKSIZE$] $chunk1$, $chunk2$;  													\nl
				\bool\ $status$ = \false, $served$ = \false;												  \bl\bl\nl
				
				a message $\langle op, cid, key, data\rangle$ is received:                						    \nl
				\n  \Switch\ ($op$) $\lbrace$                                                  						\nl
				\n    \Case\ $REQC$:  \label{code:ser-o:mov1} 						       							\nl
				\n      send($cid$, $cv$);          							\label{code:ser-o:req1}  			\nl 
				$chunk2$ = receive($cid$);								\label{code:ser-o:req2}				\nl
				\If\ (not enough free space in local list to fit elements of $chunk2$ )
				$\lbrace$ 				\label{code:ser-o:req3}				\nl
				\n      	\If\ ($my\_sid$ == $\maxser-1$)	$status$ = \false;	\label{code:req-o:0.1} 				\nl
				\Else\ $\lbrace$ 									\label{code:req-o:0.2} 				\nl
				\n 				$chunk1$ = getChunkOfElementsFromLocalList($llist$);
				\label{code:ser-o:req4}				\nl
				$status$ = ServerMove($next\_id$, $chunk1$);	\label{code:ser-o:req5}				\ul
				\p			$\rbrace$																				\nl
				\p 		$\rbrace$ \Else\ $status$ = \true;						\label{code:ser-o:req5.0}			\nl
				
				\If\ ($status$ == \true) $\lbrace$ 			 			\label{code:ser-o:req5.1}			\nl
				\n			insertChunkOfElementsInLocalList($llist, chunk2$);\label{code:ser-o:req9}				\nl
				send($cid$, \ACK);									\label{code:ser-o:req10}			\nl															
				\p		$\rbrace$ \Else\ send($cid$, \NACK);					\label{code:ser-o:req7}				\nl
				\Break; \p                                                             						\nl
				\Case\ $INSERT$: 											\label{code:ser-o:ins1}             \nl
				\n      \While\ ($served \neq \true$) \{						\label{code:ser-o:ins1.1}           \nl
				\n			$k_{max}$ = find\_max($llist$); 					\label{code:ser-o:ins2}				\nl 
				\If\ ($k_{max} > key$ and isFull($llist$) $\neq$ \true) \{    							\nl
				\n        		status = insert($llist, key, data$); 			\label{code:ser-o:ins3}	    		\nl
				send($cid$, $status$);    						\label{code:ser-o:ins4} 			\nl
				$served$ = \true;	    						\label{code:ser-o:ins4.3} 			\nl
				\p      	\} \Else\ \If\ ($k_{max} > key$) \{ 				\label{code:ser-o:ins4.1}			\nl
				\n        		$chunk1$ = getChunkOfElementsFromLocalList($llist$);\label{code:ser-o:ins4.2}		\nl	
				$status$ = ServerMove($next\_id$, $chunk1$);	\label{code:ser-o:ins5}				\nl
				\If\ ($status$ == \true) \{ 					\label{code:ser-o:ins6}				\nl
				\n					removeChunkOfElementsFromLocalList($llist$, $chunk1$);\label{code:ser-o:ins6.1} \nl	
				\p		 		\} \Else\  \{	    							\label{code:ser-o:ins6.2} 			\nl
				\n					send($cid$, \NACK); 						\label{code:ser-o:ins10.1}			\nl
				$served$ = \true;							\label{code:ser-o:ins10.1.1}		\ul
				\p				\}																					\nl
				\p     		\} \Else\  \{  										\label{code:ser-o:ins10.2}			\nl
				\n				\If\ ($my\_id \neq \maxser-1$) send($next\_id$, $\langle\INSERT, cid, key, data\rangle$);		\label{code:ser-o:ins10.4}			\nl
				\Else\  send($cid$, \NACK); 					\label{code:ser-o:ins10.5}			\nl
				$served$ = \true;								\label{code:ser-o:ins10.6}			\ul
				\p			\}																						\ul
				\p 		\}                                                                       \nl
				\Break;   \nl
				%
				\p \Case\ $SEARCH$: 										\label{code:ser-o:srch1}			\nl
				\n      $cv[cid]++$;                                    		\label{code:ser-o:s1}				\nl
				$status$ = search($llist,key$);             			\label{code:ser-o:s2}				\nl
				\If\ ($status$ == $\false$) send($cid$, \NACK));         \label{code:ser-o:s3}				\nl
				\Else\ send($cid$, \ACK);                                \label{code:ser-o:s4}				\nl
				\Break;                                                               						\nl
				\p    \Case\ $DELETE$:                                                       						\nl
				\n      $cv[cid]++$;		                                       \label{code:ser-o:del1}			\nl
				status = search($llist,key$); 							  \label{code:ser-o:del2}		    \nl
				\If\ ($status == \true$) \{ 								\label{code:ser-o:del3}			\nl
				\n        delete($llist$, $key$);                                \label{code:ser-o:del4}			\nl
				send($cid$, \ACK);                                     \label{code:ser-o:del5}			\nl
				\p      \} \Else\ send($cid$, \NACK);            				 \label{code:ser-o:del7}			\nl
				\Break;                                                        								\ul 
				\p\p \}         
				\p
			\end{code}
		\end{algorithm} 
	\end{minipage}
\end{wrapfigure}

The proposed implementation is based on the distributed unsorted list, 
presented in Section~\ref{app:unsorted-list}. 
%
%
Each server $s$ has a memory chunk of predetermined size where it maintains a part 
of the implemented list so that all elements stored on server $s_i$ have smaller keys 
than those stored on server $s_{i+1}$, $0 \leq i < \maxser-1$. Because of this sorting 
property, an element with key $k$ is not appended to the end of the list, so a token 
server is useless in this case. This is an essential difference with the unsorted list 
implementation.

Similarly to the unsorted case, a client sends an insert request for key $k$ to server 
$s_0$. The server searches its local part of the list for a key that is greater than or 
equal to $k$. In case that it finds such an element that is not equal to $k$, it can try 
to insert $k$ to its local list, $llist$. More specifically, if the server has sufficient 
storage space for a new element, it simply creates a new node with key $k$ and inserts 
it to the list. However, in case that the server does not have enough storage space, it 
tries to free it by forwarding a chunk of elements of $llist$ to the next server. If 
this is possible, it serves the request. In case $s_0$ does not find a key that is greater 
than or equal to $k$ in its $llist$, if forwards the message with the insert request 
to the next server, which in turn tries to serve the request accordingly. Notice that this 
way, a request may be forwarded from one server to the next, as in the case of the unsorted 
list. However, for ease of presentation, in the following we present a static algorithm 
where this forwarding stops at $s_{\maxser-1}$. In case that an element with $k$ is already 
present in the $llist$ of some server $s$ of the resulting sequence, then $s$ sends an 
\NACK\ message to the client that requested the insert.
 


As in the case of the unsorted list, a client performs a search or delete operation for key 
$k$ by sending the request to all servers. If not handled correctly, then the interleaving 
of the arrival of requests to servers may cause a search operation to ``miss'' the key $k$ 
that it is searching, because the corresponding element may be in the process to be moved 
from one server to a neighboring one. In order to avoid this, servers maintain a sequence 
number for each client that is incremented at every search and delete operation. Neighboring 
servers that have to move a chunk of elements among them, first verify that the latest (search 
or delete) requests that they have served for each client have compatible sequence numbers and 
perform the move only in this case.

Event-driven code for the server is presented in Algorithms \ref{fig:ol-server} 
and \ref{fig:ol-server-aux}. The clients access the sorted list using the same 
routines as they do in the case of the unsorted list (see Algorithm \ref{fig:dl-client}).

When an insert request for key $k$ reaches a server $s$, $s$ compares the maximal 
key stored in its local list to $k$. If $k$ is greater than the maximal key and $s$ is 
not $s_{\maxser-1}$, the request must be forwarded to the next server (line~\lref{code:ser-o:ins10.4}). 
Otherwise, if $k$ is to be stored on $s$, $s$ checks if $llist$ has enough space 
to serve the insert. If it does, $s$ inserts the element and sends an \ACK\ to the 
client (line~\lref{code:ser-o:ins3}-\lref{code:ser-o:ins4}). If $s$ does not have 
space for inserts, the operation cannot be executed, hence $s$ must check whether 
a chunk of its elements can be forwarded to the next server to make room for further 
inserts. To move a chunk, $s$ calls {\tt ServerMove()} (presented in Algorithm~\ref{fig:ol-server-aux})
(line \lref{code:ser-o:ins5}). If {\tt ServerMove()} succeeds in making room in $s$'s 
$llist$, the insert can be accommodated (line~\lref{code:ser-o:ins6}). In any other 
case, $s$ responds to the client with \NACK\ (line~\lref{code:ser-o:ins10.1}).

A server process a search request as described for the unsorted list, but it now pairs
each such request with a sequence number (line~\lref{code:ser-o:s1}). Delete 
is processed by a server in a way analogous to search.

In order to move a chunk of $llist$ to the next server, a server $s_i$ 
invokes the auxiliary routine {\tt ServerMove()} (line \lref{code:ser-o:ins5}).  
{\tt ServerMove()} sends a {\sf REQC} message to server $s_{i+1}$ (line \lref{code:sm:1}). 
When $s_{i+1}$ receives this request, it sends its client vector to $s_i$ (line 
\lref{code:ser-o:req1}). Upon reception (line \lref{code:sm:2}), $s_i$ compares 
its own client vector to that of $s_{i+1}$ and as long as it lags behind $s_{i+1}$ for 
any client, it services search and delete requests until it catches up to $s_{i+1}$ 
(lines \lref{code:sm:3}-\lref{code:sm:5}). Notice that during this time, $s_{i+1}$ 
does not serve further client request, in order allow $s_i$ to catch up with it. 
As soon as $s_i$ and $s_{i+1}$ are compatible in the client delete and search requests 
that they have served, $s_i$ sends to $s_{i+1}$ a chunk of the elements in its local 
list (lines \lref{code:sm:6}-\lref{code:sm:7}) and awaits the response of $s_{i+1}$. 
We remark that in order to perform this kind of bulk transfer, as the one carried out 
between a server executing line \lref{code:sm:8} and another server executing line 
\lref{code:ser-o:req2}, we consider that remote DMA transfers are employed. 
This is omitted from the pseudocode for ease of presentation.
 
If $s_{i+1}$ can store the chunk of elements, then it does so and sends \ACK\ to $s_i$. 
Upon reception, $s_i$ may now remove this chunk from its local list (line \lref{code:ser-o:ins6}) 
and attempt to serve the insert request. Notice that if $s_{i+1}$ cannot store the 
chunk of elements of $s_i$, then it itself initiates the same chunk moving procedure 
with its next neighbor (lines \lref{code:req-o:0.2}-\lref{code:ser-o:req5}), and if 
it is successful in moving a chunk of its own, then it can accommodate the chunk 
received by $s_i$.
Notice that in the static sorted list that is presented here, this protocol may 
potentially spread up to server $s_{\maxser-1}$ (line \lref{code:req-o:0.1}). 
If $s_{\maxser-1}$  does not have available space, then the moving of the chunk 
fails (line \lref{code:ser-o:ins6.2}). The client then receives a \NACK\ response, 
corresponding to a full list. 

We remark that this implementation can become dynamic by appropriately exploiting 
the placement of the servers on the logical ring, in a way similar to what we do in 
the unsorted version.

\ignore{The semantics here are relaxed, as concurrent delete 
operations may have freed up space in other servers. Such issues may be overcome 
in a dynamic version of the protocol, where $s_{\maxser-1}$ is allowed to request 
the chunk move from $s_0$. }

\begin{wrapfigure}{l}{0.7\textwidth}
	\begin{minipage}{.65\textwidth}
\begin{algorithm}[H]
\caption{Auxiliary routine ServerMove for the servers of the distributed sorted list.}
\label{fig:ol-server-aux}
\footnotesize
\begin{code}
\firstline
 \bf \bool\ ServerMove(int $cid$, data $chunk1$)  \{             	\label{code:sm:0} 				\nl
 \n  \bool\ $status$;                                                         					 	\nl
 	 data $chunk2$;																				 \bl\nl

       	send($next\_id$, $\langle REQC, cid, 0, \bot\rangle$);     	\label{code:sm:1} 				\nl
       	$nbr\_cv$ = receive($next\_id$);                        	\label{code:sm:2} 				\nl
      	\While\ (for any element $i$, $cv[i] < nbr\_cv[i]$) \{  	\label{code:sm:3}               \nl
\n      	receiveMessageOfType($SEARCH$ or $DELETE$);     		\label{code:sm:4}    			\nl
			$service\ request$                                  	\label{code:sm:5}   			\ul
\p      \}                                                                							\nl
       	$chunk2$ = getChunkOfElementsFromLocalList($llist$);   		\label{code:sm:6}   			\nl
        send($next\_id$, $chunk2$);                       			\label{code:sm:7}				\nl
       	$status$ = receive($next\_id$);                        		\label{code:sm:8} 				\nl
		\If\ ($status$ == \true) $\lbrace$ 			 			\label{code:ser-o:sm5.1}			\nl
\n			removeChunkOfElementsFromLocalList($llist$, $chunk2$); \label{code:ser-o:sm6}			\nl
			insertChunkOfElementsInLocalList($llist, chunk1$);	\label{code:ser-o:sm9}				\nl
			\return\ \true;										\label{code:ser-o:sm10}				\ul
\p		$\rbrace$ \Else\ \return\ \false;						\label{code:ser-o:sm7}				\ul
 \p \}
\end{code}
\end{algorithm} 
\end{minipage}
\end{wrapfigure}


\section{Details on Hierarchical Approach}
\label{app:hierarchical}
We interpolate one or more communication layers 
by using intermediate servers between the servers that maintain parts of the 
data structure and the clients. The number of intermediate servers and the number of
layers of intermediate servers between clients and servers can be tuned for achieving
better performance. 

For simplicity of presentation, we focus on the case where there is a single layer of 
intermediate servers. We present first the details for a fully non cache-coherent architecture.

For each island $i$, we appoint one process executing on a core of this island,
called the island master (and denoted by $m_i$), as the intermediate server. 
Process $m_i$ is responsible to gather messages from all the 
other cores of the island, and batch them together before 
forwarding them to the appropriate server. This way, we exploit the fast communication between the cores 
of the same island and we minimize the number of messages sent to the servers by 
putting many small messages to one batch.
We remark that each batch can be sent to a server by performing \DMA. In this case,
$m_i$ initiates the \DMA\ to the server's memory, and once the  \DMA\ is completed,
it sends a small message to the server to notify it about the new data that it 
has to process.

Algorithm \ref{alg:hier-master} presents the events triggered in an island master 
$m_i$ and its actions in order to handle them. $m_i$ receives messages from clients
that have type OUT (outgoing messages) and from servers that have type IN (incoming 
messages). The outbatch messages are stored in the $outbuf$ array. Each time a 
client from island $i$ wants to execute an operation, it sends a message to $m_i$;
$m_i$ checks the destination server id ($sid$), recorded in the message, 
and packs this message together with other messages directed to 
$sid$ (lines~\lref{ln:hier-outgoing1}-\lref{ln:hier-outgoing2}). Server $m_i$
has set a timer, and it will submit this batch
of messages to server $sid$ (as well as other batches of messages to other servers), 
when the timer expires. When $m_i$ receives an incoming message from 
a server, it unpacks it, and sends each message to the appropriate client on its 
island. When the timer is triggered, $m_i$ places each batch of messages in an
$outbuf$ array that $m_i$ maintains for the appropriate server. The
transfer of all these messages to the server may occur using \DMA. For simplicity,
we use two auxiliary functions: \texttt{add\_message} to add a message to an $outbuf$ buffer of $m_i$,
and split to split a batch of messages $msg$ that have arrived to 
the particular messages of the batch which are then placed in the $inbuf$ buffer of $m_i$.
 
The code of the client does not change much. Instead of sending messages directly to the 
server, it sends them to the board master.

\begin{algorithm}[!t]
\caption{Events triggered in an island master - Case of fully non cache-coherent architectures.}
\label{alg:hier-master}
\begin{code}
\lreset
\firstline
      LocalArray $outbuf= \varnothing$;	    \cm{stores outgoing messages}  \nl
      LocalArray $inbuf = \varnothing$;    \cm{stores incoming messages}  \ul
	  \nl
    a message $\langle type, msg\rangle$ is received:                      \nl
 \n	  \If ($type$ == OUT) $\lbrace$              \label{ln:hier-outgoing1} \nl 
 \n     $sid$ = read field $sid$ from $msg$;                                               \nl
        add\_message($outbuf$, $sid$, $msg$);    \label{ln:hier-outgoing2} \nl
 \p   $\rbrace$ \Elseif\ ($op$ == IN) $\lbrace$  \label{ln:hier-incoming1}
                          \cm{$m_i$ received a batch of messages from a server } \nl
 \n     $inbuf$ = split($msg$);		        \cm{unbundle the message to many small ones }  \nl
        \Foreach\ message $m$ in $inbuf$ $\lbrace$  \cm{the batch is for all cores in the island }  \nl
 \n       $cid$ = read field $client$ from $msg$;                                 \nl
          send($cid$, $msg$);                    \ul
 \p     $\rbrace$    \ul
 \p   $\rbrace$    \ul
    \nl
 \p timer is triggered:  \cm{Every timeout, $m_i$ sends outgoing messages} \nl
 \n   \Foreach\ batch of messages in $outbuf$ $\lbrace$ 
                                    \cm{send a bach to a server} \nl
 \n     send($sid$, batch); \cm{this send can be done using DMA}                         \nl
       	delete($outbuf$, $sid$, batch); \ul
 \p   $\rbrace$
 \p
\end{code}
\end{algorithm}

\ignore{
\begin{algorithm}
\caption{Client Algorithm.}
\label{alg101}
\begin{code}
 \n \void\ send(msg) $\lbrace$               \nl
 \n   mid = get the island master id;        \nl
      send(mid, $\langle OUT, msg\rangle$);  \nl
 \p $\rbrace$                                \nl
	                                         \nl
    msg receive(\integer\ mid) $\lbrace$     \nl
 \n   $\langle OUT, msg\rangle$              \nl
      \return\ msg;                          \ul
 \p $\rbrace$
\end{code}
\end{algorithm} 
}

We remark that in order to improve the scalability 
of a directory-based algorithm, a locality-sensitive hash functions
could be a preferable choice. For instance, the simple currently employed $\mod$ hash function,  
can be replaced by a  hash function that divides by some integer $k$.
Then, elements of up to $k$ subsequent insert (i.e. push or enqueue) operations may be sent to the same directory server.
This approach suites better to bulk transfers since it allows for exploiting locality. 
Specifically, consider that $m_i$ sends a batch of elements to be inserted to the 
synchronizer $s_s$ of a directory-based data structure. Server $s_s$ 
unpacks the batch and processes each of the requests contained therein separately.
Thus, if $\mod$ is used, each of these elements will be stored in different buckets (of different directory servers).  
As a sample alternative, if the ${\tt div}$ hash function is used, 
more than one elements may end up to be stored in the same bucket. 
When later on a batch of remove (i.e. pops or dequeues) operations arrives to $s_s$,
it can request from the directory server that stores the first of the elements
to be removed, to additionally remove and send back further elements with subsequent keys 
that are located in the same bucket. 
Notice that in this case, the use of 
DMA can optimize these transfers.

We now turn attention to partially non cache-coherent architectures where the cores of an island communicate via 
cache-coherent shared memory.
Algorithm \ref{alg:ccsynch} presents code for the hierarchical approach in this case. 
For each island, we use an instance of the \CCSYNCH\ combining synchronization algorithm, presented in~\cite{FK12}. 
All clients of island $i$ participate to the instance of \CCSYNCH\ for island $i$,
i.e. each such client  calls \CCSYNCH\ (see Algorithm \ref{alg:ccsynch}) to execute an operation.

\begin{algorithm}[!t]
\caption{Pseudocode of hierarchical approach - Case of partially non cache-coherent architectures.}
\footnotesize
\label{alg:ccsynch}
\begin{code}
~~~~~struct Node \{                                                                          \ul
\n  Request req;                                                                        \ul
    RetVal ret;                                                                         \ul
    \integer\ id;                                                                       \ul
    boolean wait;                                                                       \ul
    boolean completed;                                                                  \ul
    int sid;                                                                  \ul
    Node *next;                                                                         \ul
\p\};                                                                                   \bl
                                                                                        \ul
shared Node *Tail;                                                                      \ul
private Node *$node_i$;                                                                 \ul
LocalArray $outbuf= \varnothing$;	    \cm{stores outgoing messages}                   \ul
                                                                                        \ul
RetVal CC-Synch(Request req) \{                \cm{ Pseudocode for thread $p_i$ }       \ul
\n  Node *nextNode, *tmpNode, *tmpNodeNext;                                             \ul
    \integer\ counter = 0;                                                               \bl
                                                                                        \nl
    $node_i\rightarrow wait$ = true; \label{alg:ccsynch:init_wait}                      \nl
    $node_i\rightarrow next$ = null;\label{alg:ccsynch:init_next}                       \nl
    $node_i\rightarrow completed$ = false;\label{alg:ccsynch:init_completed}            \nl
    $nextNode = node_i$;\label{alg:ccsynch:next_node}                                   \nl
    $node_i = \textit{Swap}(Tail, node_i)$;  \label{alg:ccsynch:swap}           \nl
    $node_i\rightarrow req$ = req;            \cm{ $p_i$ announces its request\label{alg:ccsynch:init_reg} }      \nl
    $node_i\rightarrow sid$ = destination server;                  \nl
    $node_i\rightarrow next$ = nextNode;\label{alg:ccsynch:connect_to_next}                                       \nl
    \While\ ($node_i\rightarrow wait$ == true)\cm{ $p_i$ spins until it is unlocked\label{alg:ccsynch:spinning} }  \ul 
\n      nop;                                                                            \nl
\p  \If\ ($node_i\rightarrow completed$==true)\cm{ if $p_i$'s req is already applied\label{alg:ccsynch:operation_completed} }\nl
\n      \return\ $node_i\rightarrow ret$;      \cm{ $p_i$ returns its return value\label{alg:ccsynch:return} }     \nl
\p  tmpNode = $node_i$;                     \cm{ otherwise $p_i$ is the combiner\label{alg:ccsynch:tmpnode} }     \nl
    \While\ (tmpNode $\rightarrow$ next $\neq$ $null$ AND counter $<$ $h$)\label{alg:ccsynch:combine}\{            \nl 
\n      counter = counter + 1;\label{alg:ccsynch:inc_counter}                          \nl
        tmpNodeNext=tmpNode$\rightarrow$next;\label{alg:ccsynch:go_next}               \nl
        add\_message($outbuf$, tmpNode$\rightarrow$sid, tmpNode$\rightarrow$req);       \nl
        tmpNode = tmpNodeNext;             \cm{ proceed to the next node\label{alg:ccsynch:combine_next} }    \ul
\p  \}                                                                                 \nl
    \Foreach\ batch of messages in $outbuf$  $\lbrace$  \cm{send a batch of messages to  servers}\label{alg:ccsynch:send_fat}       \nl
\n      send($sid$, batch of message);       \cm{where $sid$ is the destination server for this batch}                               \nl
       	delete($outbuf$, $sid$, $fatm$);                                      \ul
\p  $\rbrace$                                                                          \nl
    $inbuf$ = split(receive());		                                                       \nl
    tmpNode = $node_i$;                                                                \nl
    \While\ (counter $\geq$ 0) \{ \label{alg:ccsynch:while2}                           \nl 
\n      counter = counter - 1;                                                         \nl
        tmpNodeNext=tmpNode$\rightarrow$next;                                          \nl
        tmpNode$\rightarrow$ret = find in $inbuf$ the response for tmpNode$\rightarrow$id;                     \nl
        tmpNode$\rightarrow$completed = true;     \nl
        tmpNode$\rightarrow$wait = false;  \cm{ unlock the spinning thread\label{alg:ccsynch:set_wait} }          \nl
        tmpNode = tmpNodeNext; \label{alg:ccsynch:combine_next2}                       \ul
\p  \}                                                                                 \nl
    tmpNode$\rightarrow$wait = false;         \cm{ unlock next node's owner\label{alg:ccsynch:unlock_next} }      \nl
    \return\ $node_i\rightarrow$ret;\label{alg:ccsynch:combiner_returns}                \ul
\p\}                                                                                  
\end{code}
\end{algorithm}

\CCSYNCH\ employs a list which contains
one node for each client that has initiated an operation; 
the last node of the list is a dummy node.
After announcing its request by 
by recording it
in the last node of the list (i.e., in the 
dummy node) and by inserting a new node as the last node of the list
(which  will comprise the new dummy node),
a client tries to acquire 
a global lock (line~\lref{alg:ccsynch:swap}) which is implemented as a queue lock.
The client that manages to acquire the lock,
called the {\em combiner}, batches those active requests, recorded in the list,
that target the same server 
(lines~\lref{alg:ccsynch:combine}-\lref{alg:ccsynch:combine_next}) 
and forwards them to this server (line~\lref{alg:ccsynch:send_fat}).
Thus, at each point in time, the combiner plays the role of the island master.
When the island master receives (a batch of) responses from a server, it records 
each of them in the appropriate element of the request list to inform active clients 
of the island about the completion of their requests 
(lines\lref{alg:ccsynch:while2}-\lref{alg:ccsynch:combine_next2}). 
In the meantime, each such client performs spinning (on the element
in which it recorded its request) until either the
response for its request has been fulled by the island master or the global lock has 
been released (line~\lref{alg:ccsynch:spinning}).

We use the list of requests to implement the global lock as a 
{\em queue lock}~\cite{C93, MLH94}.
The process that has recorded its request in the head node of the list plays the role of the combiner.

\section{Experimental Evaluation}
\label{perf}
We run our experiments on the Formic-Cube~\cite{formic},
which is a hardware prototype of a $512$ core,  
non-cache-coherent machine.
It consists of $64$ boards with $8$ cores each (for a total of 512 cores).
Each core owns $8$ KB of private L1 cache, and $256$ KB of private L2 cache. 
None of these caches is hardware coherent.
The boards are connected with a fast, lossless 
packet-based network forming a 3D-mesh with a diameter of $6$ hops.
Each core is equipped with its own local {\em hardware mailbox},
an incoming hardware FIFO queue, whose size is $4$ KB. 
It can be written by any core and read by the core that owns it. 
One core per board plays the role of the island master
(and could be one of the algorithm's servers), 
whereas the remaining $7$ cores of the board serve as clients.

Our experiments are similar to those presented in~\cite{fatourou:spaa11,FK12,MS96}.
More specifically, $10^7$ pairs of requests (\PUSH\ and \POP\ or \ENQUEUE\ and \DEQUEUE) 
 are executed in total, as the number of cores increases.
To make the experiment more realistic, a random local work (up to $512$ dummy loop iterations) is 
simulated between the execution of two consecutive requests by the same thread
as in~\cite{fatourou:spaa11,FK12,MS96}. 
To reduce the overheads for the memory allocation of the stack nodes, we allocate a pool of nodes
(instead of allocating one node each time). 

In Figure~\ref{fig:perf-queue}.a, we experimentally compare the performance of the
centralized queue (\CQUEUE)
to the performance of its hierarchical version (\HQUEUE), 
and to those of the hierarchical versions of the directory-based queue (\DQUEUE)
and the token-based queue (\TQUEUE).
We measure the average throughput achieved by each algorithm.
As expected, \CQUEUE\ does not scale well. 
Specifically, the experiment shows that for more 
than $16$ cores in the system,
the throughput of the algorithm remains almost the same.
We remark that the clients running on these $16$ cores 
do not send enough messages to fill up the mailbox of the server. 
This allows us to conclude that 
when the server receives about $16$ messages or more,
for reading these messages, processing the requests
they contain, and sending back the responses to clients, 
the server ends up to be always busy.
We remark that reading each message from the mailbox causes
a cache miss to the server. So, 
the dominant factor at the server side in this case 
is to perform the reading of these messages from the mailbox.


These remarks are further supported by studying the experiment for the \HQUEUE\ implementation.
Specifically, the throughput of \HQUEUE\ does not further increase when the number of cores
becomes $64$ or more. 
Remarkably, the $64$ active cores are located in $8$ boards,
so there exist $8$ island masters in the system. Each island master sends 
two messages to the server -- one for batched enqueue and 
one for batched dequeue requests. So, again, the server becomes saturated 
when it receives about $16$ messages. These messages
are read from the mailbox in about the same time  as 
in the setting of \CQUEUE\ with $16$ running cores. 
However, in the case of \HQUEUE, each message
contains more requests to be processed by the server.
Therefore, the average time needed to process a request is now smaller
since the overhead 
of reading and processing a message is divided over the number of 
requests it contains.
Notice that now, the server has more work to do in terms of processing
requests.
The experiment shows that the time required for this 
is evened out by the time saved for processing each request. 

Similarly to \CQUEUE\ and \HQUEUE, the \DQUEUE\ 
implementation uses a centralized component, namely 
the synchronizer. However, contrary to the \HQUEUE\ case, 
in the \DQUEUE, each island master sends one 
message instead of two, which contains the number of 
both the enqueues and dequeues requests. 
Since we do not see the throughput
stop increasing at $128$ cores, 
we conclude that now the dominant factor is not 
the time that the server requires to read
the messages from its mailbox.
The \DQUEUE\ graph of Figure~\ref{fig:perf-queue}.a 
shows that \DQUEUE\ scales well
for up to $512$ cores.
Therefore, in the \DQUEUE\ approach, the synchronizer does not 
pose a scalability problem. The reason for this 
is, not only that the synchronizer receives a smaller number of messages, 
but also that it has to do a simple 
arithmetic addition or subtraction for each batch of 
requests that it receives. This computational effort 
is significantly smaller than those 
carried out by the centralized component in the \CQUEUE\ 
and \HQUEUE\ implementations.  Therefore, it is important 
that the local computation done by a server be small. 

\begin{figure}[!t]
\centering
\subfigure{\includegraphics[width=0.99\linewidth]{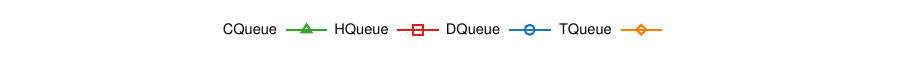}}
\subfigure{\includegraphics[width=0.99\linewidth]{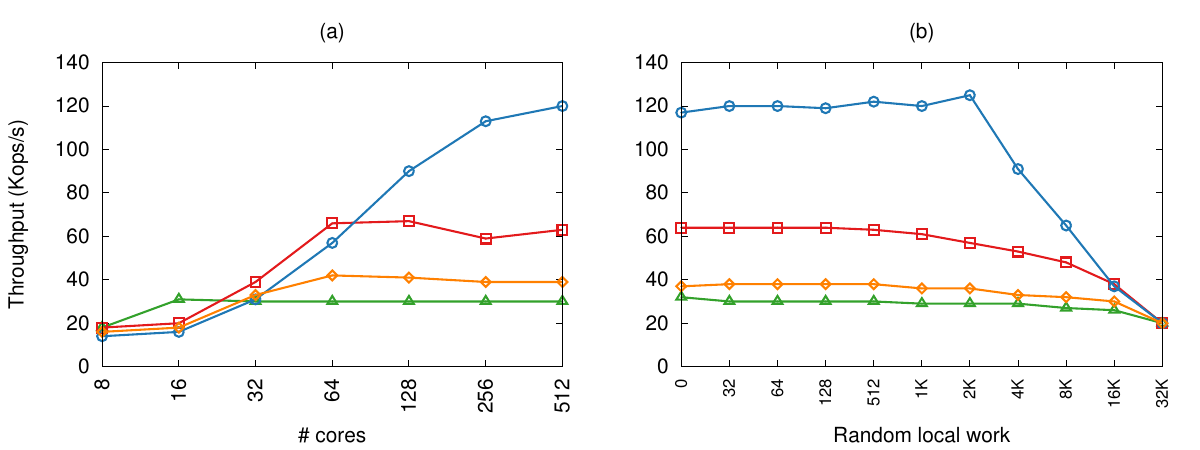}}
\caption{Performance evaluation of (a) distributed queue implementations,
(b) distributed queue implementations while executing different 
amounts of local work ($512$ cores).
}
\label{fig:perf-queue}
\end{figure}

\begin{figure}[!t]
	\centering
	\subfigure{\includegraphics[width=0.99\linewidth]{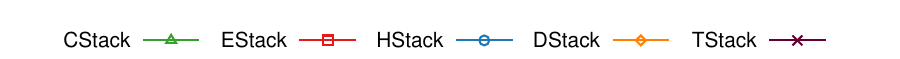}}
	\subfigure{\includegraphics[width=0.99\linewidth]{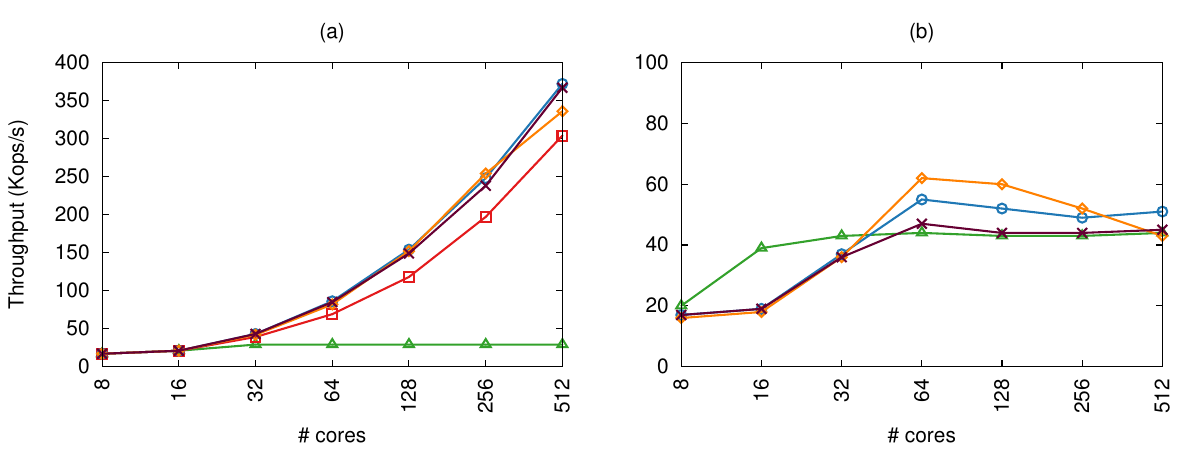}}
	\caption{Performance evaluation of (a) distributed stack implementations with elimination,
		(b) distributed stack implementations without elimination.
	}
	\label{fig:perf-stack}
\end{figure}

Notice that in the \DQUEUE\ implementation, 
the actual request processing takes place on the hash 
table servers. 
So, clients do not initiate requests 
as frequently as in the previous algorithms, 
since they also have to communicate with the hash table servers. 
Moreover, the processing in this 
case is shared among the hash-table servers and therefore, 
this processing does not 
cause scalability problems. It follows that load 
balancing is also an important factor affecting 
scalability. In the case of the \DQUEUE\, the local 
work on the synchronizer is a linear function of the 
amount of island masters. On the contrary, in the other 
two implementations, it is a function of the amount of 
clients. It follows that this is another reason 
for the good scalability observed on the \DQUEUE\ 
implementation.

We remark that when the amount of clients, and therefore, of island 
masters, is so large as to cause saturation on the 
synchronizer, a tree-like hierarchy of island masters 
would solve the scalability problem. There, the \DQUEUE\ 
algorithm can offer a trade-off between overloaded 
activity of the centralized component and the latency 
that is caused by the height of the tree hierarchy. 

\ignore{
It seems that 
the hash table computation plus the work of the clients 
leads to non-linear speed-up. The non-linear speed-up may 
also be due to the fact that more messages have to pass 
through the same number of paths, leading the server 
underutilized. 
}

Figure~\ref{fig:perf-queue}.a further shows the observed 
throughput of \TQUEUE. The behavior of \TQUEUE\ in terms 
of scalability follows that of \HQUEUE. Notice that 
\TQUEUE\ works in a similar way as \HQUEUE, with the 
difference that the identity of the centralized component 
may change, as the token moves from server to server. 
This makes \TQUEUE\ more complex than \HQUEUE. 
As a result, 
its throughput is lower than that of \HQUEUE. 
However, \TQUEUE\ is a nice 
generalization of \HQUEUE, which can be used in cases where 
the expected size of the required queue is not known in 
advance and where a moderate number of cores are active. 
\ignore{
Although in the \TQUEUE\ case, each message received by 
the token server batches $7$ requests, the local computation 
that is necessary for the handling of the token creates a 
significant overhead, that starts to dominate when more 
than $64$ cores are active, i.e. when merely $8$ messages 
are in the server's mailbox. 
}

In Figure~\ref{fig:perf-queue}.b, we fix the number of 
cores to $512$ and perform the experiment for several 
different random work values ($0$ - $32$K). It is shown 
that, for a wide range of values ($0$-$512$), we see no 
big difference on the performance of each algorithm. 
This is so because, for this range of values, 
the cost to perform the requests dominates 
the cost introduced by the random work. When the random work becomes too
high (greater than $32$K dummy loop iterations), the 
throughput of all algorithms degrades and the performance 
differences among them become minimal, since the amount 
of random work becomes then the dominant performance 
factor.


In Figure~\ref{fig:perf-stack}.a, we experimentally compare the 
performance of the centralized stack (\CSTACK) with its hierarchical 
version where the island master performs elimination (\ESTACK), an 
improved version of \ESTACK\ where the island master performs batching 
(\BSTACK), and the hierarchical versions of the directory-based stack (\DSTACK)
and the token based-stack (\TSTACK). As expected, the centralized 
implementation does not scale for more than $16$ cores. All other 
algorithms scale well for up to $512$ cores. This shows that the 
elimination technique is highly-scalable. It is so efficient that 
it results in no significant performance differences between the 
algorithms that apply it. 

\ignore{
Remarkably, this experiment additionally 
shows that \TSTACK\ does not have any significant overheads in 
comparison to \BSTACK\ if the number of servers that store stack 
data is not too big. Thus, \TSTACK\ is a nice generalization of 
\BSTACK\ since \TSTACK\ fairly distributes the stack data in a 
locality-sensitive way and ensures load balancing when necessary.
}

In order to get a better estimation of the effect that the 
elimination technique has on the algorithms, we experimentally 
compared the performance that \CSTACK, \ESTACK, \BSTACK, \TSTACK, 
and \DSTACK\ can achieve, when they are not performing elimination. 
Figure~\ref{fig:perf-stack}.b shows the obtained throughput. 
Notice that the scalability characteristics of \CSTACK, \BSTACK,
and \TSTACK\ are similar to those of \CQUEUE, \HQUEUE, and \TQUEUE.
This is not the case for \DSTACK. 
The reason for this is the following. In our experiment, each client performs pairs of 
push and pops. By the way the synchronizer works, the push request and the pop request 
of each pair are often assigned the same key. This results in contention at the hash table.
In an experiment where the number of pushes is not equal to the number of pops,
the observed performance would be much better than that of other algorithms for all numbers of cores.

Figure~\ref{fig:msgs}.b shows the total number 
of messages sent in each experiment presented in Figure~\ref{fig:perf-queue}.a. 
Remarkably, there is an inverse relationship between these results and 
those of Figure~\ref{fig:perf-queue}.a. 
For instance, \DQUEUE\  circulates the largest number of messages
in the system. However, \DQUEUE\ is the implementation that scales best. 
Thus, the total number of messages is not necessarily 
an indicative factor of the actual performance.
\ignore{sheds more 
light on this issue, as it presents the maximum number of messages 
that a single server receives for a given implementation. }
This is so because a server may easily become saturated even by a moderate 
amount of messages, if the required processing is important. This 
is for example the case in the \CQUEUE\ implementation. 
Therefore, in order to be scalable, an implementation should avoid 
overloading the servers in this manner. 
Backed up by Figure~\ref{fig:msgs}.b,
becomes aparent that good scalability is offered when the implementation 
ensures good load balancing between the messages that the servers 
have to process.

\ignore{
\begin{figure}[!t]
\centering
\includegraphics[width=0.99\linewidth]{queues.eps}
\caption{Performance evaluation of (a) distributed queue implementations,
(b) distributed queue implementations while executing different 
amounts of local work ($512$ cores),
and (c) total number of messages.}
\label{fig:msgs-queue}
\end{figure}
}

\begin{figure}[!t]
	\centering
	\includegraphics[width=0.99\linewidth]{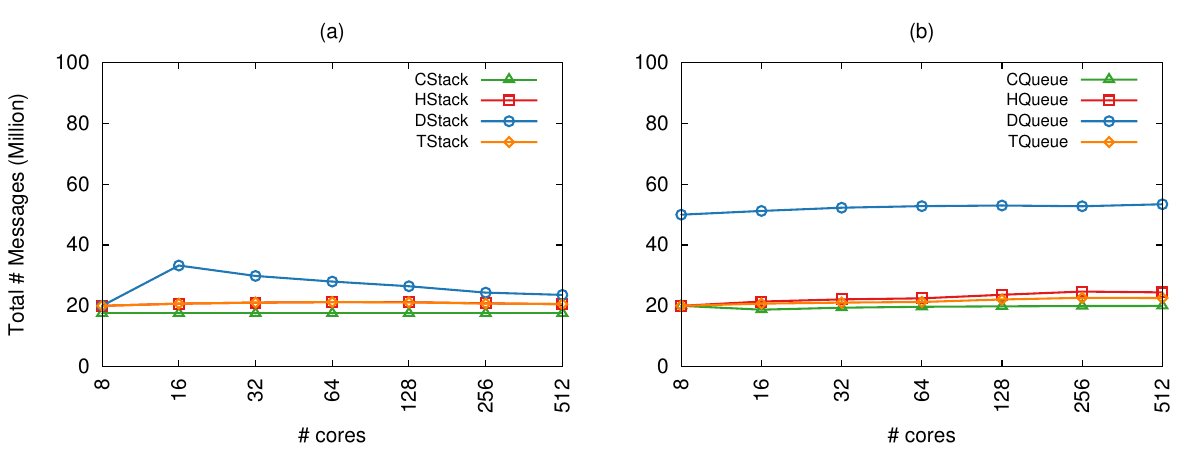}
	\caption{Total number of messages received in the proposed implementations by all servers in the system for $10^7$ operations.}
	\label{fig:msgs}
\end{figure}

\begin{figure}[!t]
\centering
\includegraphics[width=0.99\linewidth]{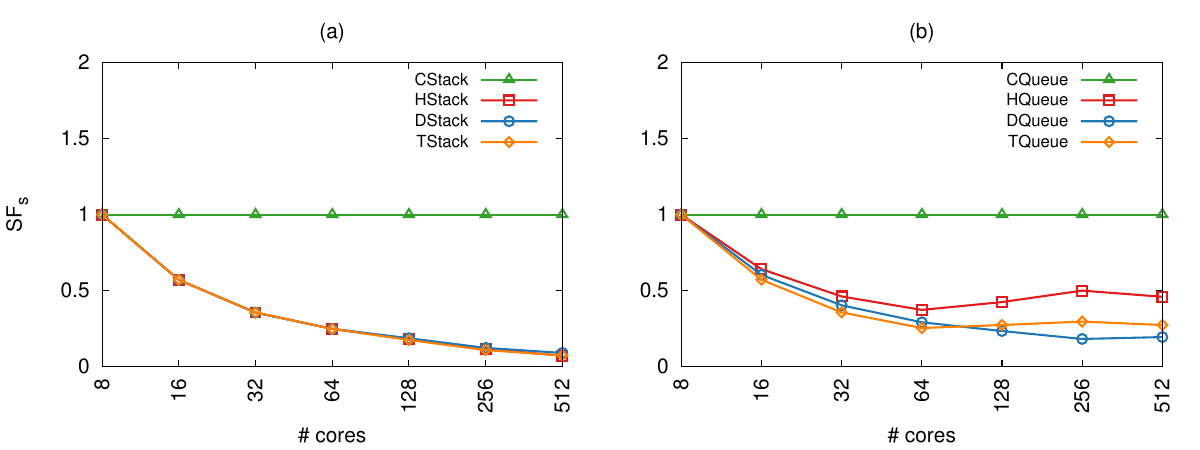}
\caption{Scalability factors of presented algorithms.}
\label{fig:sfs-stack-queue}
\end{figure}

We distill the empirical observations of this 
section into a metric. 
As observed, 
achieving load balancing in terms of evenly distributing both messages 
and the processing of requests to servers is important for ensuring scalability.
If we fix an algorithm ${\cal A}$, 
this is reflected in the number of messages received 
by each server $s$ for performing $m$ requests when executing ${\cal A}$. By 
denoting this number as $msg_s$, we can define 
the {\em scalability factor} $sf_s$ of a server $s$ as follows:
\begin{equation}
sf_s = \lim_{m \rightarrow \infty} \frac{msg_s}{m} 
\end{equation}
where we assume that a maximum number of clients repeatedly initiate requests
and these clients are scheduled fairly.
We define the scalability factor $sf$ of ${\cal A}$ as the maximum 
of the scalability factors of all servers.
\begin{equation}
sf = \max_{s}\{sf_s\} 
\end{equation}

We remark that a low scalability factor indicates good scalability behavior.
The intuition behind this metric is that the lower the $sf_s$ fraction is, 
the more requests are batched in a message that reaches a server. 
Computational overhead for the reading, handling and decoding 
of a single message is then spread over multiple requests.
As a scalability indicator, the scalability 
factor shows that it is of more interest to attempt not to minimize 
the total number of messages that are trafficked in a system, but 
to design implementations that minimize the total number 
of messages that a server has to process.
Figure~\ref{fig:sfs-stack-queue} shows graphs of the scalability factors 
of the the proposed implementations. The results agree with this
theoretical perception.

\ignore{
\lchanged{In summary, the insights from our experiments 
are the following: It is apparent that the hierarchical approach has 
significant performance benefits since it decreases the number of 
messages received by servers and combined with batching, communication 
and computation are greatly minimized. Communication in considered 
architectures is fast and not always the dominant factor. 
\here{no experiments for lists, though...:}
\sout{Experiments 
show that parallelism should not be sacrificed for less communication: 
in list search and delete, sending a message to all servers results in 
better performance than reducing parallelism (as is the case for list 
insert).} Waiting must be avoided by the servers as much as possible but 
it could be the preferable choice for clients.
}
}

\section{Implementation of Shared-Memory Primitives}
\label{sec:primitives}
\ignore{The conclusions drawn from the previous study are that in order to support most of the 
methods the library \texttt{sun.misc.unsafe} uses, we need to implement the {\em atomic 
updates}, the {\em compare-and-swap primitives}, and the pair of \texttt{pak()/unpar()} 
methods used for thread synchronization. 
In addition to these, we are going to implement the 
primitives {\em fetch-and-increment} and {\em swap}, to implement numeric operations used by 
the methods described in section \ref{sec:primitives}, i.e. \texttt{getAndAdd()}. 

These primitives need to be implemented on a lower level than the Java Virtual Machine, 
since \texttt{sun.misc.unsafe} mostly calls primitives that need either to access the 
system or hardware resource. Thus, when the system can provide these primitives, the 
implementations of the aforementioned packages will come for free since these classes are 
dependent to the methods called from the library \texttt{sun.misc.unsafe}. 
}

\lchanged{In Java, synchronization primitives  are provided as methods of the library
sun.misc.unsafe.}
\lchanged{In order to support most of the methods {\tt sun.misc.unsafe} 
uses, we need to implement the atomic updates, the compare-and-swap primitives, 
and the pair of {\tt park()/unpark()} methods used for thread synchronization. 
Once these primitives have been supported, the implementations of several data 
structures that are provided in {\tt java.util.concurrent} will come for free. 

In addition to these primitives, we have implemented the primitives 
fetch-and-increment and swap, to implement numeric operations used 
by the methods of package \texttt{java.util.concurrent.atomic}, described in Section~\ref{sec:primitives}, e.g. {\tt getAndSet()}. 
%
These primitives need to be implemented on a lower level than the Java
Virtual Machine, since {\tt sun.misc.unsafe} mostly calls primitives that 
need either to access the system or hardware resource.}

We have to point out that Formic does support reads and stores to remote memory locations 
but does not provide coherence. Thus, this remote store is different from the atomic write 
primitive we want to implement, due to the atomicity that the second provides. By adding this 
primitive, we aim to give safety when simple writes occur concurrently with other atomic 
primitives, such as \CAS, etc..

Hence, we implemented the following atomic primitives:
\begin{itemize}
	\item Read, which taks as arguments a memory address and returns the value that 
	is stored in this address.
	\item Write, which takes as arguments a memory address and a value, and stores 
	the new value to this address. 
	\item Fetch-And-Increment, which takes as arguments a memory address and a value,
	and adds the new value to the existing one stored in the address. It then returns 
	the new value. 
	\item Swap, which takes as arguments a memory address, and a new value, and replaces
	the stored value with the new one. It then returns the previous value. 
	\item Compare-And-Swap, which takes as arguments a memory address, an old value and
	a new value. If the value stored in the memory address equals to the old value, then 
	it is replaced by the new value. Then, the function returns the old value. 
\end{itemize}

\section{Atomic Accesses Support}
\label{sec:atomics}

The \texttt{java.util.concurrent.atomic} package provides a set of
atomic types (as classes).  Instances of these classes must be
atomically accessed.  In shared-memory architectures these classes are
implemented by delegating the complexity of synchronization handling
to instructions provided by the underlying architecture, e.g., memory
barriers, compare and swap etc.  In the case of Formic, however, such
instructions are not available.  As a result the atomic types need to
be implemented in software.

A naive implementation is to delegate the handling of such operations
to a manager similar to the monitor manager, presented in Section 2.3 
of Deliverable 1.1. This manager would be responsible for holding
the values associated to instances of atomic types, as well as, for
performing atomic operations on them.  
Such a manager is capable of reducing the memory traffic
regarding synchronization in some cases.  


\begin{algorithm}[!ht]
\caption{Example of getAndSet implementation.}
\label{code:java0}
\begin{code}
  \lreset
  \firstline
public final int getAndSet(int newValue) \{									\nl
\n    \For\ (;;) \{															\nl
\n      int current = get(); 				\cm{Synchronization point}	    \nl
        \If\ (compareAndSet(current, newValue)) \cm{Synchronization point}	\nl
\n            \return\ current;												\nl
\p \p   \} \p 																\nl
\} 
\end{code}
\end{algorithm}

\lchanged{For instance, in the implementation depicted in Algorithm~\ref{code:java0}}, 
the JVM constantly tries to set the value of the object to
\texttt{newValue} before its value changes between the \texttt{get()}
and the \texttt{compareAndSet} invocations.  This may result in many
failures and unnecessary synchronizations in case of high contention.
When using a manager there is no need for such a loop.  Since the
accesses are only possible by the manager, the manager can assume that
the value will never change between the \texttt{get()} and a
consequent \texttt{set()} invocation.  As a result we could implement
the logic of the above function in the manager and avoid the extra
communication from the loop iterations.  The equivalent algorithm in
the manager \lchanged{is shown in Algorithm~\ref{code:java1}.}


\begin{algorithm}[!ht]
\caption{Implementation of getAndSet through the monitor manager.}
\label{code:java1}
\begin{code}
  \firstline
public final int getAndSet(int newValue) \{		\nl
\n    int current = get();						\nl
      set(newValue);							\nl 
      \return\ current;							\nl
\p \} 
\end{code}
\end{algorithm}

A similar effect can be achieved by using the \texttt{synchronized}
classifier.  Since the code inside a \texttt{synchronized} block or
method can be seen as atomic we can implement getAndSet as above but
without the need of a centralized manager by just adding
\texttt{synchronized} to the method classifiers, 
\lchanged{as shown in Algorithm~\ref{code:java2}.}


\begin{algorithm}[!h]
\caption{Implementation of getAndSet using \texttt{synchronized}.}
\label{code:java2}
\begin{code}
  \firstline
public final synchronized int getAndSet(int newValue) \{\nl
\n    int current = get();								\nl
      set(newValue);									\nl
      \return\ current;									\nl
\p\} 
\end{code}
\end{algorithm}

Note that, in shared-memory processors, this approach is probably less
efficient than the first one when there is no contention on the
object.  In the lack of contention the first implementation would just
perform a read instruction and a compare and swap, while in the third
implementation it would also need to go through the monitor
implementation of the virtual machine (at minimum an extra compare and
swap, and a write).  On the Formic, however, where there is no compare
and swap instruction, the first implementation would require at least
two request messages (\texttt{get} and \texttt{compareAndSet}) to the
manager and the corresponding two replies per iteration.  The third
implementation would also need two requests (a monitor acquire and a
monitor release) and the corresponding two replies but with the
difference that it never needs to repeat this process, resulting in
reduced memory traffic.

To further improve this approach we avoid the use of regular monitors
and replace them with readers-writers monitors.  This way in the case
of contented reads we do not restrict access to a single thread,
allowing for increased parallelism.  As a further optimization we
slightly change the semantics of \texttt{compareAndSet} and make it
\textit{lazely} return \texttt{False} in case that the object is owned
by another writer.  This behavior is based on the heuristic that if an
object is owned by a writer it is going to be written and the value
will probably be different than the one the programmer provided as
expected to the \texttt{compareAndSet}.

\subsection{Readers-Writers Implementation}
\label{sec:rw-monitors}

The readers-writers monitors are implemented using the same monitor
manager, presented in Section 2.3 of Deliverable 1.1, but with different
operation codes.  We introduce four new operations code,
\texttt{READ\_LOCK}, \texttt{WRITE\_LOCK}, \texttt{READ\_UNLOCK}, and
\texttt{WRITE\_UNLOCK}.  For each readers-writers monitor we keep two
different queues, the readers queue and the writers queue.  The
readers queue holds the thread IDs of the threads waiting to acquire a
read-lock on the monitor.  Correspondingly, the writers queue holds
the thread IDs of the threads waiting to acquire a write-lock on the
monitor.

When a monitor is read-locked we hold the count of threads sharing
that read monitor, while for write-locked monitors we hold the thread
ID of the thread owning that write monitor.  As long as the writers
lock is empty the monitor manager services read-lock requests by
increasing the counter and sending an acknowledgement message to the
requester.  When a write-lock request arrives, it gets queued to the
writers queue and read-lock requests start being queued in the
readers-queue instead of being served, essentially giving priority
to the writer-lock requests. Eventually, when all the readers release 
the read-locks they hold the write-lock requests will start being 
served, and read-lock requests will be served  again only after 
the writers-queue gets empty again.

We chose to give priority to the write locks since in most algorithms
writing a variable is less common than reading it (e.g., polling).  In
order to implement a more fair mechanism we could set a threshold on
the number of read-lock and write-lock requests being served at each
phase.

Finally, to support the \textit{lazy fail} for \texttt{compareAndSet}
we introduce another operation code, the \texttt{TRY\_LOCK}.  This is 
a special write-lock request that if the monitor is not free, it does
not get queued in the write-queue, but a negative acknowledgement is
send back to the requester, notifying him that the monitor is not
free.

\section{Conclusions}
\label{discussion}
We have presented a comprehensive collection of data structures 
for future many-core architectures. The collection could be utilized 
by runtimes of high-productivity languages ported to such 
architectures. 
Our collection provides 
all types of concurrent data structures 
supported in Java's concurrency utilities. Other high-level 
productivity languages that provide shared memory for thread 
communication could also benefit from our library. Specifically, 
we provide several different kinds of queues, including static, 
dynamic and synchronous; our queue (or deque) implementations can 
be trivially adjusted to provide the functionality of {\em delay 
queues} (or {\em delay deques})~\cite{lea:book06}. We do not provide 
a priority queue implementation, since it is easy to adopt
a simplified version of the priority queue presented  
in~\cite{mans98} in our setting. Our list implementations
provide the functionality of sets, whereas the simple hash 
table that we utilize to design some of our data structures 
can serve as a hash-based map. 

We have outlined 
hierarchical versions of the data structures that we have implemented. 
These implementations take into consideration challenges that are 
raised in realistic scalable multicore architectures where 
communication is implicit only between the cores of an island
whereas explicit communication is employed among islands.

We have performed experimental evaluation of the implemented data 
structures in order to examine both their throughput and energy 
efficiency.  Our 
experiments show the performance and scalability characteristics 
of some of the techniques on top of a non cache-coherent hardware 
prototype. They also illustrate the scalability power of the 
hierarchical approach in such machines. We believe that the proposed 
implementations will exhibit the same performance characteristics, 
if programmed appropriately, in prototypes with similar 
characteristics as FORMIC, like Tilera or SCC. We expect this also 
to be true for future, commercially available, such machines. 

\ignore{
We have proposed a new metric, the {\em scalability factor}, 
in an effort to provide deeper insight into how to design 
data structures suitable for the considered architectures. 
We have deduced this metric by analyzing our experimental 
results and have shown experimental measurements that verify 
it.

As a further example of customizable scalable data structure, 
we have implemented a leaf-oriented (a,b)-tree, where $a, b$ 
are integers such that $a \geq 2$ and $b \geq 2a$. 
A leaf-oriented tree stores all keys currently in the dictionary 
in its leaves. Internal nodes of the tree store keys that may 
or may not be in the dictionary, used to direct a search operation 
along the path to the correct leaf. To keep the tree balance, 
we apply {\em preventive splitting} during inserts and {\em preventive 
merging} during deletes in a top-down approach (from the root 
to the leaves). To do so, we introduce the concepts of safe nodes. 
A node is {\em insert-safe} if it has less than $4$ children; 
it is insert-unsafe otherwise. A node is {\em delete-safe} if 
it has more than $2$ children; it is {\em delete-unsafe} otherwise. 
A node is always {\em search-safe}. The tree is optimized for non 
cache-coherent architectures. Details of the implementation are 
provided in the submitted {\em additional material}.

\ignore{We present a distributed implementation of a dictionary using a leaf oriented (a,b)-tree, where 
$a \geq 2$ and $b \geq 2a$; for simplicity, we describe the case where $a = 2$ and $b = 4$. 
A (2,4)-tree is a perfectly balanced tree, each node of which has either $2$, $3$ or $4$ children.
A leaf-oriented tree stores all keys currently in the dictionary in its leaves. 
Internal nodes of the tree store keys that may or may not be in the dictionary,
used to direct a search operation along the path to the correct leaf. 
To keep the tree balance, we apply {\em preventive splitting} during inserts 
and {\em preventive merging} during deletes in a top-down approach (from the root to the leaves). 
To do so, we introduce 
the concepts of safe nodes. 
A node is {\em insert-safe}
if it has less than $4$ children; it is insert-unsafe otherwise.
A node is {\em delete-safe} if it has more than $2$ children; it is {\em delete-unsafe} otherwise. 
A node is always {\em search-safe}.}

Our work is a first step towards designing customized scalable data 
structures for future non cache-coherent many-core architectures. It 
therefore reveals several interesting open problems. Our list 
implementations are highly parallel for search and delete operations 
but slow for insert. It is an interesting open problem to come up
(if possible) with a more update-efficient implementation of a list.
A simple variant of the sorted list could support range queries
for free. It is interesting to investigate whether such queries 
can be supported without sacrificing efficiency for updates. 

Investigating whether caching or replication could increase scalability 
or enhance performance is also of interest. Coping with fault-tolerance 
is another fruitful research direction. To make key-based implementations 
such our tree more scalable, we can partition the key universe into ranges 
and have one distributed (a,b)-tree for each range (similarly to what is 
suggested in~\cite{ailamaki}). It would be interesting to design different 
types of trees for our setting.

\ignore{
In our tree implementation, all requests are sent to the root node $s$, 
which might become overloaded. The hierarchical approach reduces the number 
of messages send to $s$ but it does not reduce the amount of work $s$ has 
to perform. To make our implementation more scalable, we can partition 
the key universe into ranges and have one distributed (a,b)-tree for each 
range (similarly to what is suggested in~\cite{ailamaki}). Then, the root 
of each (a,b)-tree of the forest can be stored in different servers. 
It would be interesting to design different types of trees for our setting.
}
%
}


\bibliographystyle{unsrt}
\bibliography{main}

\begin{thebibliography}{10}

\bibitem{novakovic2014scale}
Stanko Novakovic, Alexandros Daglis, Edouard Bugnion, Babak Falsafi, and Boris
  Grot.
\newblock Scale-out numa.
\newblock In {\em ASPLOS}, pages 3--18. ACM, 2014.

\bibitem{carter2013runnemede}
Nicholas Carter, Aditya Agrawal, et~al.
\newblock {Runnemede: An architecture for Ubiquitous High-Performance
  Computing.}
\newblock In {\em {HPCA}}, pages 198--209, 2013.

\bibitem{6077845}
M.~Gries, U.~Hoffmann, M.~Konow, and M.~Riepen.
\newblock Scc: A flexible architecture for many-core platform research.
\newblock {\em Computing in Science Engineering}, 13(6):79--83, Nov 2011.

\bibitem{formic}
Spyros Lyberis, George Kalokerinos, Michalis Lygerakis, et~al.
\newblock Formic: Cost-efficient and scalable prototyping of manycore
  architectures.
\newblock In {\em {FCCM}}, 2012.

\bibitem{mcilroy:oopsla10}
Ross McIlroy and Joe Sventek.
\newblock {Hera-JVM}: a runtime system for heterogeneous multi-core
  architectures.
\newblock In {\em {OOPSLA}}, pages 205--222, 2010.

\bibitem{yu1997java}
Weimin Yu and Alan Cox.
\newblock {Java/DSM: A platform for heterogeneous computing}.
\newblock {\em Concurrency: Practice and Experience}, 9(11):1213--1224, 1997.

\bibitem{zhu2002jessica2}
Wenzhang Zhu, Cho-Li Wang, and Francis~CM Lau.
\newblock Jessica2: A distributed java virtual machine with transparent thread
  migration support.
\newblock In {\em IEEE Cluster}, pages 381--388, 2002.

\bibitem{lea:book06}
Douglas Lea.
\newblock {\em Concurrent Programming in Java(TM): Design Principles and
  Patterns (3rd Edition)}.
\newblock Addison-Wesley Professional, 2006.

\bibitem{javautils}
Oracle.
\newblock Java utilities library.

\bibitem{MS96}
Maged~M. Michael and Michael~L. Scott.
\newblock Simple, fast, and practical non-blocking and blocking concurrent
  queue algorithms.
\newblock In {\em PODC}, pages 267--275, NY, USA, 1996.

\bibitem{Pugh}
William Pugh.
\newblock Skip lists: A probabilistic alternative to balanced trees.
\newblock {\em Commun. ACM}, 33(6):668--676, June 1990.

\bibitem{scherer:2006}
William~N. {{Scherer III}}, Doug Lea, and Michael~L. Scott.
\newblock Scalable synchronous queues.
\newblock In {\em {PPoPP}}, New York, US, Mar 2006.

\bibitem{NA1991}
Daniel Nussbaum and Anant Agarwal.
\newblock Scalability of parallel machines.
\newblock {\em Commun. ACM}, 34(3):57--61, March 1991.

\bibitem{DMS15}
David Dice, Virendra~J. Marathe, and Nir Shavit.
\newblock Lock cohorting: A general technique for designing numa locks.
\newblock In {\em {PPoPP}}, pages 247--256, 2012.

\bibitem{FK12}
Panagiota Fatourou and Nikolaos~D. Kallimanis.
\newblock {Revisiting the combining synchronization technique.}
\newblock In {\em {SPAA}}, pages 257--266, 2012.

\bibitem{Hendler:2004}
Danny Hendler, Nir Shavit, and Lena Yerushalmi.
\newblock A scalable lock-free stack algorithm.
\newblock In {\em {SPAA}}, pages 206--215. ACM, 2004.

\bibitem{HIST10}
Danny Hendler, Itai Incze, Nir Shavit, and Moran Tzafrir.
\newblock Flat combining and the synchronization-parallelism tradeoff.
\newblock In {\em {SPAA}}, pages 355--364, 2010.

\bibitem{Devine93}
Robert Devine.
\newblock {Design and Implementation of DDH: A Distributed Dynamic Hashing
  Algorithm}.
\newblock In {\em {FODO}}, pages 101--114, 1993.

\bibitem{hazelcast}
Hazelcast.
\newblock The leading in-memory data grid.
\newblock \url{http://hazelcast.com/}.

\bibitem{shahmirzadiphd}
Omid Shahmirzadi.
\newblock {\em High-Performance Communication Primitives and Data Structures on
  Message-Passing Manycores}.
\newblock PhD thesis, {EPFL}, 2014.
\newblock n. 6328.

\bibitem{Count1}
William Aiello, Costas Busch, Maurice Herlihy, Marios Mavronicolas, Nir Shavit,
  and Dan Touitou.
\newblock Supporting increment and decrement operations in balancing networks.
\newblock In {\em STACS 99}, volume 1563 of {\em Lecture Notes in Computer
  Science}, pages 393--403. Springer, 1999.

\bibitem{Count2}
James Aspnes, Maurice Herlihy, and Nir Shavit.
\newblock Counting networks.
\newblock {\em J. ACM}, 41:5:1020--5:1048, September 1994.

\bibitem{BAB08}
Robert~L. Bocchino, Vikram~S. Adve, and Bradford~L. Chamberlain.
\newblock Software transactional memory for large scale clusters.
\newblock In {\em {PPoPP}}, pages 247--258, 2008.

\bibitem{D2STM}
M.~Couceiro et~al.
\newblock {D2STM: Dependable Distributed Software Transactional Memory}.
\newblock In {\em {PRDC}}, 2009.

\bibitem{DhokePR15}
Aditya Dhoke, Roberto Palmieri, and Binoy Ravindran.
\newblock On reducing false conflicts in distributed transactional data
  structures.
\newblock In {\em {ICDCN}}, pages 8:1--8:10, January 2015.

\bibitem{TM2C}
Vincent Gramoli, Rachid Guerraoui, and Vasileios Trigonakis.
\newblock {TM2C: A Software Transactional Memory for Many-cores}.
\newblock In {\em {EuroSys}}, pages 351--364, 2012.

\bibitem{DiSTM}
Christos Kotselidis, Mohammad Ansari, Kim Jarvis, Mikel Lujn, Chris~C. Kirkham,
  and Ian Watson.
\newblock {DiSTM: A Software Transactional Memory Framework for Clusters.}
\newblock In {\em ICPP}, pages 51--58. IEEE Computer Society, 2008.

\bibitem{Man06}
Kaloian Manassiev, Madalin Mihailescu, and Cristiana Amza.
\newblock Exploiting distributed version concurrency in a transactional memory
  cluster.
\newblock In {\em {PPoPP}}, pages 198--208. ACM, 2006.

\bibitem{Saad2011}
Mohamed~M. Saad and Binoy Ravindran.
\newblock {Supporting STM in Distributed Systems: Mechanisms and a Java
  Framework}.
\newblock In {\em {TRANSACT}}, 2011.

\bibitem{SR11}
Mohamed~M. Saad and Binoy Ravindran.
\newblock {HyFlow: A High Performance Distributed Software Transactional Memory
  Framework}.
\newblock In {\em {HPDC}}, pages 265--266, 2011.

\bibitem{herlihy:isca93}
Maurice Herlihy and J.~Eliot~B. Moss.
\newblock Transactional memory: architectural support for lock-free data
  structures.
\newblock In {\em {ISCA}}, 1993.

\bibitem{ST1995}
Nir Shavit and Dan Touitou.
\newblock Software transactional memory.
\newblock In {\em {PODC}}, pages 204--213. ACM, 1995.

\bibitem{Attiya}
Hagit Attiya and Jennifer Welch.
\newblock {\em Distributed Computing: Fundamentals, Simulations and Advanced
  Topics (2nd edition)}.
\newblock John Wiley Interscience, March 2004.

\bibitem{lamport:cacm78}
Leslie Lamport.
\newblock Time, clocks, and the ordering of events in a distributed system.
\newblock {\em Commun. ACM}, 21(7), 1978.

\bibitem{Saad2011a}
Mohamed~M. Saad and Binoy Ravindran.
\newblock Transactional forwarding algorithm.
\newblock Technical report, Virginia Tech, 2011.

\bibitem{SaadR11}
Mohamed~M. Saad and Binoy Ravindran.
\newblock {Snake: Control Flow Distributed Software Transactional Memory}.
\newblock In {\em {SSS}}, pages 238--252, 2011.

\bibitem{BF10}
Annette Bieniusa and Thomas Fuhrmann.
\newblock {Consistency in hindsight: A fully decentralized STM algorithm}.
\newblock In {\em {IPDPS}}, pages 1--12, April 2010.

\bibitem{Combine}
Hagit Attiya, Vincent Gramoli, and Alessia Milani.
\newblock A provably starvation-free distributed directory protocol.
\newblock In {\em {SSS}}, pages 405--419, September 2010.

\bibitem{Attiya2015}
Hagit Attiya, Vincent Gramoli, and Alessia Milani.
\newblock Directory protocols for distributed transactional memory.
\newblock In {\em Transactional Memory. Foundations, Algorithms, Tools, and
  Applications}, volume 8913, pages 367--391. Springer, 2015.

\bibitem{Ballistic}
Maurice Herlihy and Ye~Sun.
\newblock Distributed transactional memory for metric-space networks.
\newblock In {\em {DISC}}, pages 324--338, 2005.

\bibitem{Spiral}
Gokarna Sharma and Costas Busch.
\newblock Distributed transactional memory for general networks.
\newblock {\em Distrib. Comput.}, 27(5):329--362, October 2014.

\bibitem{Relay}
Bo~Zhang and Binoy Ravindran.
\newblock Brief announcement: Relay: {A} cache-coherence protocol for
  distributed transactional memory.
\newblock In {\em {OPODIS}}, 2009.

\bibitem{Arrow}
Michael~J. Demmer and Maurice Herlihy.
\newblock The arrow distributed directory protocol.
\newblock In {\em DISC}, volume 1499 of {\em Lecture Notes in Computer
  Science}, pages 119--133. Springer, 1998.

\bibitem{KBCL09}
Milind Kulkarni et~al.
\newblock How much parallelism is there in irregular applications?
\newblock In {\em {PPoPP}}, pages 3--14, 2009.

\bibitem{KPR+08}
Milind Kulkarni et~al.
\newblock Optimistic parallelism benefits from data partitioning.
\newblock In {\em {ASPLOS}}, 2008.

\bibitem{LDK+08}
D.B. Larkins et~al.
\newblock Global trees: A framework for linked data structures on distributed
  memory parallel systems.
\newblock In {\em SC}, pages 1--13, Nov 2008.

\bibitem{Kroll:sigmod94}
Brigitte Kr\"{o}ll and Peter Widmayer.
\newblock Distributing a search tree among a growing number of processors.
\newblock In {\em Proceedings of the 1994 ACM SIGMOD International Conference
  on Management of Data}, pages 265--276, New York, USA, 1994.

\bibitem{Aspnes03}
James Aspnes and Gauri Shah.
\newblock {Skip Graphs}.
\newblock In {\em {SODA}}, pages 384--393, Philadelphia, USA, 2003. SIAM.

\bibitem{Gribble:2000}
Steven~D. Gribble et~al.
\newblock Scalable, distributed data structures for internet service
  construction.
\newblock In {\em {OSDI}}, pages 22--22, 2000.

\bibitem{Hilford:97}
Victoria Hilford, Farokh~B. Bastani, and Bojan Cukic.
\newblock Eh* - extendible hashing in a distributed environment.
\newblock In {\em {COMPSAC}}, 1997.

\bibitem{Martin01}
Richard~P. Martin, Kiran Nagaraja, and Thu~D. Nguyen.
\newblock Using distributed data structures for constructing cluster-based
  services.
\newblock In {\em {EASY}}, 2001.

\bibitem{DBLP:journals/pvldb/AguileraGS08}
Marcos~Kawazoe Aguilera, Wojciech~M. Golab, and Mehul~A. Shah.
\newblock A practical scalable distributed {B-tree}.
\newblock {\em {PVLDB}}, 1(1):598--609, 2008.

\bibitem{gridgain}
GridGain.
\newblock Gridgain - in-memory data fabric.
\newblock \url{http://www.gridgain.com/}.

\bibitem{Nelson:2015:LSD:2813767.2813789}
Jacob Nelson, Brandon Holt, Brandon Myers, Preston Briggs, Luis Ceze, Simon
  Kahan, and Mark Oskin.
\newblock Latency-tolerant software distributed shared memory.
\newblock In {\em Proceedings of the 2015 USENIX Conference on Usenix Annual
  Technical Conference}, USENIX ATC '15, pages 291--305, Berkeley, CA, USA,
  2015. USENIX Association.

\bibitem{D11}
David Dice, Virendra~J. Marathe, and Nir Shavit.
\newblock Flat-combining {NUMA} locks.
\newblock In {\em {SPAA}}, pages 65--74, June 2011.

\bibitem{TGT15}
Tudor David, Rachid Guerraoui, and Vasileios Trigonakis.
\newblock Asynchronized concurrency: The secret to scaling concurrent search
  data structures.
\newblock In {\em {ASPLOS}}, pages 631--644, March 2015.

\bibitem{Lynch1996}
Nancy~A. Lynch.
\newblock {\em Distributed Algorithms}.
\newblock Morgan Kaufmann Publishers Inc., San Francisco, CA, USA, 1996.

\bibitem{herlihy1990linearizability}
Maurice~P Herlihy and Jeannette~M Wing.
\newblock Linearizability: A correctness condition for concurrent objects.
\newblock {\em {TOPLAS}}, 12(3):463--492, 1990.

\bibitem{Herlihy:2008}
Maurice Herlihy and Nir Shavit.
\newblock {\em The Art of Multiprocessor Programming}.
\newblock Morgan Kaufmann Publishers Inc., San Francisco, CA, USA, 2008.

\bibitem{mellor1991algorithms}
John~M. Mellor-Crummey and Michael~L. Scott.
\newblock Algorithms for scalable synchronization on shared-memory
  multiprocessors.
\newblock {\em {TOCS}}, 9(1):21--65, 1991.

\bibitem{ultracomputer1982}
Allan Gottlieb et~al.
\newblock {The NYU Ultracomputer designing a MIMD, shared-memory parallel
  machine}.
\newblock In {\em ACM SIGARCH Computer Architecture News}, volume~10, pages
  27--42. IEEE Computer Society Press, 1982.

\bibitem{C93}
Travis~S. Craig.
\newblock {Building FIFO and priority-queueing spin locks from atomic swap}.
\newblock Technical Report TR 93-02-02, Department of Computer Science,
  University of Washington, February 1993.

\bibitem{MLH94}
Peter~S. Magnusson, Anders Landin, and Erik Hagersten.
\newblock {Queue Locks on Cache Coherent Multiprocessors}.
\newblock In {\em {Proceedings of the 8th International Symposium on Parallel
  Processing (IPDPS)}}, pages 165--171, 1994.

\bibitem{fatourou:spaa11}
Panagiota Fatourou and Nikolaos~D. Kallimanis.
\newblock A highly-efficient wait-free universal construction.
\newblock In {\em Proceedings of the Twenty-third Annual ACM Symposium on
  Parallelism in Algorithms and Architectures}, pages 325--334, NY, USA, 2011.

\bibitem{mans98}
Bernard Mans.
\newblock Portable distributed priority queues with {MPI}.
\newblock {\em Concurrency: Practice and Experience}, 10(3):175--198, 1998.

\end{thebibliography}

\end{document}